%% file: main.tex
\documentclass[11pt]{article} 
\usepackage[margin=1in]{geometry}

\usepackage{float}

\usepackage{framed}

\usepackage[linesnumbered,boxruled,vlined]{algorithm2e}
\usepackage{verbatim}

\usepackage{amssymb,amsfonts,amsmath,amsthm}

\usepackage{enumerate}

\usepackage[font=small,labelfont=bf]{caption}

\usepackage{color}

\usepackage{graphicx}
\usepackage{tabulary}
\usepackage{diagbox}
\usepackage[dvipsnames]{xcolor}
\usepackage{pgfplots}
\usepackage[colorlinks=true, linkcolor=blue, urlcolor=blue, citecolor=ForestGreen]{hyperref}

\usepackage{comment}

\usepackage{setspace}

\newcommand{\poly}{\operatorname{poly}}
\newcommand{\polylog}{\operatorname{polylog}}

\newcommand{\Alg}{\mathrm{Alg}}

\newcommand{\diam}{\mathrm{diam}}

\newcommand{\argmax}{\mathop{\mathrm{argmax}}}

\newcommand{\median}{\mathrm{median}}

\newcommand{\fail}{\operatorname{{\bf FAIL}}}

\newcommand{\ALG}{\mathcal{L}_{k+1}}

\newcommand{\cost}{\operatorname{{\bf MST}}}
\newcommand{\tester}{\operatorname{{\bf Tester}}}

\newcommand{\ecost}{\widehat{\operatorname{{\bf MST}}}}
\newcommand{\R}{\mathbb{R}}

\newcommand{\N}{\mathbb{N}}
\newcommand{\Z}{\mathbb{Z}}

\newcommand{\E}{\operatorname{{\bf E}}}
\newcommand{\Ex}{\mathop{{\bf E}\/}}
\newcommand{\Exp}{\mathop{{Exp}\/}}
\newcommand{\ex}[1]{\mathop{{\bf E}\left[ #1 \right]}}
\newcommand{\exx}[2]{\mathop{{\bf E}}_{#1}\left[ #2 \right]}

\newcommand{\Var}[1]{\mathop{{\bf Var}\left( #1 \right)}}

\newcommand{\pr}[1]{\operatorname{{\bf Pr}}\left[ #1 \right]}
\newcommand{\prb}[2]{\mathop{{\bf Pr}}_{#1}\left[ #2 \right]}

\newcommand{\Prx}{\mathop{{\bf Pr}\/}}

\newcommand{\anote}[1]{\footnote{\color{blue} Amit: #1}}

\newcommand{\MST}[1]{\operatorname{MST}}

\newcommand{\bA}{\mathbf{A}}
\newcommand{\bB}{\mathbf{B}}

\newcommand{\bE}{\mathbf{E}}

\newcommand{\bS}{\mathbf{S}}

\newcommand{\bh}{\boldsymbol{h}}
\newcommand{\bi}{\boldsymbol{i}}

\newcommand{\boldr}{\boldsymbol{r}} 

\newcommand{\bt}{\boldsymbol{t}}

\newcommand{\bx}{\boldsymbol{x}}

\newcommand{\bz}{\boldsymbol{z}}

\newcommand{\balpha}{\boldsymbol{\alpha}}

\newcommand{\btau}{\boldsymbol{\tau}}

\newcommand{\bDelta}{\boldsymbol{\Delta}}

\usepackage{footnote}

\newcommand{\CC}{\text{CC}}

\newcommand{\cB}{\mathcal{B}}
\newcommand{\cC}{\mathcal{C}}
\newcommand{\cD}{\mathcal{D}}
\newcommand{\cE}{\mathcal{E}}

\newcommand{\cH}{\mathcal{H}}
\newcommand{\cI}{\mathcal{I}}

\newcommand{\cK}{\mathcal{K}}
\newcommand{\cL}{\mathcal{L}}

\newcommand{\cO}{\mathcal{O}}

\newcommand{\cQ}{\mathcal{Q}}

\newcommand{\cS}{\mathcal{S}}

\newcommand{\cU}{\mathcal{U}}

\newcommand{\cX}{\mathcal{X}}

\newcommand{\cord}{\Lambda}
\newcommand{\eps}{\epsilon}

\newtheorem{theorem}{Theorem}
\newtheorem{lemma}{Lemma}
\newtheorem{corollary}{Corollary}

\newtheorem{fact}{Fact}
\newtheorem{proposition}{Proposition}
\newtheorem{definition}{Definition}

\newtheorem{claim}{Claim}

\usepackage[framemethod=TikZ]{mdframed}
\newcounter{Frame}

\title{Streaming Euclidean MST to a Constant Factor} 

\author {
Vincent Cohen-Addad \thanks{Google Research. \texttt{cohenaddad@google.com}}
\and
  Xi Chen\thanks{ Columbia University.
 \texttt{xichen@cs.columbia.edu}. Supported by NSF grants IIS-1838154, CCF-2106429 and CCF-2107187.}
 \and 
  Rajesh Jayaram\thanks{Google Research.
  \texttt{rkjayaram@google.com}}
  \and
  Amit Levi\thanks{\texttt{amit.levi@uwaterloo.ca}}
  \and
  Erik Waingarten\thanks{Stanford University.
  \texttt{eaw@cs.columbia.edu}}
}

\date{}

\begin{document}
\maketitle
\begin{abstract}
We study streaming algorithms for the fundamental geometric problem of computing the cost of the Euclidean Minimum Spanning Tree (MST) on an $n$-point set $X \subset \R^d$. In the streaming model, the points in $X$ can be added and removed arbitrarily, and the goal is to maintain an approximation in small space.  In low dimensions, $(1+\epsilon)$ approximations are possible in sublinear space [Frahling, Indyk, Sohler, SoCG '05]. However, for high dimensional spaces the best known approximation for this problem was $\tilde{O}(\log n)$, due to [Chen, Jayaram, Levi,  Waingarten, STOC '22], improving on the prior $O(\log^2 n)$ bound due to [Indyk, STOC '04] and [Andoni, Indyk, Krauthgamer, SODA '08]. In this paper, we break the logarithmic barrier, and give the first constant factor sublinear space approximation to Euclidean MST. For any $\epsilon\geq 1$, our algorithm achieves an $\tilde{O}(\epsilon^{-2})$ approximation in $n^{O(\epsilon)}$ space.

We complement this by proving that any single pass algorithm which obtains a better than $1.10$-approximation must use $\Omega(\sqrt{n})$ space, demonstrating that $(1+\epsilon)$ approximations are not possible in high-dimensions, and that our algorithm is tight up to a constant. Nevertheless, we demonstrate that $(1+\epsilon)$ approximations are possible in sublinear space with $O(1/\epsilon)$ passes over the stream. More generally, for any $\alpha \geq 2$, we give a $\alpha$-pass streaming algorithm which achieves a $(1+O(\frac{\log \alpha + 1}{ \alpha \epsilon}))$ approximation in $n^{O(\epsilon)} d^{O(1)}$ space. 

All our streaming algorithms are linear sketches, and therefore extend to the massively-parallel computation model (MPC). Thus, our results imply the first $(1+\epsilon)$-approximation to Euclidean MST in a constant number of rounds in the MPC model. Previously, such a result was only known for low-dimensional space  [Andoni, Nikolov, Onak, Yaroslavtsev, STOC '15], or required either $O(\log n)$ rounds or a $O(\log n)$ approximation.

\end{abstract}

\thispagestyle{empty}
\newpage
\parskip 4.0 pt
\tableofcontents
\thispagestyle{empty}
\newpage
\parskip 7.2pt 
\pagenumbering{arabic}

\section{Introduction}\label{sec:intro}
\input{intro}

\subsection{Techniques}\label{sec:tech}
\input{techniques.tex}

\subsection{Roadmap}
We introduce preliminary notation and results in Section \ref{sec:prelims}, including key facts about the Quadtree decomposition. In Section \ref{sec:estimatorMain}, we provide our main estimator for Euclidean MST, which takes parameters $\eps \in (0,1),\alpha \geq 2$, and gives an additive $O(\frac{1}{\eps \alpha})$ approximation. In Section \ref{sec:pPass}, we show how this aforementioned algorithm can be implemented via an $\alpha$-pass linear sketch with space $n^O(\eps)d^{O(1)}$, which results in $\alpha$-pass turnstile streaming algorithms, and $O(\alpha)$ round MPC algorithms. In Section \ref{sec:estimatorOnePass}, we design a modified estimator which is will be feasible to implement in a single pass via recursive $\ell_p$ sampling. In Section \ref{sec:onePassSketch}, we show how to solve the recursive $\ell_p$ sampling problem, thereby obtaining our one-pass algorithm. Finally, in Section \ref{sec:LB}, we prove a lower bound on the space required by any one-pass algorithm that achieves a small enough constant approximation.

\section{Preliminaries}\label{sec:prelims}

\input{prelims}

\section{A Sketching-Friendly Estimator for Euclidean MST}\label{sec:estimatorMain}

\input{structural_lemma}

\input{Algorithm}

\input{pPassSketching.tex}

\input{onePass}
\input{OnePassSketching.tex}

\section{Lower Bound}\label{sec:LB}
\input{LowerBound.tex}

\subsection{Acknowledgements}
The authors thanks Sepehr Assadi for suggesting a lower bound for general metric spaces which largely inspired the one described above.

\bibliography{cluster}
\end{document}

%% file: intro.tex
The minimum spanning tree (MST) problem is one of the most 
fundamental problem in combinatorial optimization, whose 
algorithmic study dates back to the work of Boruvka in 1926~\cite{boruuvka1926jistem}. Given a set of points and distances between the points, the goal is to compute a tree over the points of 
minimum total weight. 

MST is a central problem in computer science and has received a tremendous amount of attention from the algorithm design community, leading to a large 
toolbox of methods for the problem in various 
computation models and for various types of 
inputs~\cite{Charikar:2002,IT03,indyk2004algorithms,10.1145/1064092.1064116,har2012approximate,AndoniNikolov,andoni2016sketching,bateni,czumaj2009estimating,CzumajEFMNRS05,CzumajS04,chazelle00,chazellerubinfeld}.
In the offline setting, given a graph $G=(V,E)$ and a 
weight function $w$ on the edges of the graph, an exact 
randomized algorithm running in time $O(|V|+|E|)$ is 
known~\cite{10.1145/201019.201022}, while the best known deterministic 
algorithm runs in time $O(|V | + |E| \alpha(|E|, |V |))$~\cite{chazelle00}, where $\alpha$ is the inverse Ackermann function. 

The version of the 
problem where the input lies in Euclidean space has been extensively
studied, see~\cite{eppstein2000spanning} for a survey. In this setting, 
the vertices of the graph are points in $\R^d$, and the 
set of weighted edges is (implicitly given by) the set of all ${n \choose 2}$ pairs of vertices and the pairwise Euclidean distances. Despite the implicit representation, the best known bounds for computing an Euclidean MST exactly is $n^2$, even for low-dimensional inputs, while one can obtain a $(1+\eps)$-approximation in time $O(n^{2-2/(\lceil d/2\rceil +1) +\eps})$ and an $O(c)$-approximation $n^{2-1/c^2}$~\cite{Har-PeledIS13}. 

The MST problem has been also the topic of significant study in the field of \textit{sublinear algorithms}, where the goal is to estimate 
the weight\footnote{Any algorithm which returns an approximate MST requires $\Omega(n^2)$ time~\cite{indyk1999sublinear} (for general metrics), and $\Omega(n)$ memory (since this is the size of an MST). Thus sublinear algorithms focus on approximating the weight of the MST.} of the MST in a sublinear amount of space or time. Chazelle, Rubinfeld, and Trevisan~\cite{chazellerubinfeld} provided an algorithm with running
time $O(D W \epsilon^{-3})$ to estimate the weight of the 
MST to a $(1+\epsilon)$-factor for arbitrary graphs of 
maximum degree $D$ and edge weights in $[1,W]$.
This result was improved by Czumaj and Sohler~\cite{czumaj2009estimating} for metric (e.g., Euclidean) MST, who
gave a $(1+\eps)$ approximation in time $\tilde{O}(n \eps^{-7})$.
This naturally raised the question of whether these results could lead to approximation algorithms for Euclidean MST 
in related models, such as the streaming model. This was 
latter shown to be possible for low-dimensional inputs~\cite{10.1145/1064092.1064116}, 
(i.e., with space complexity exponential in the dimension). However, it has remained a major open 
question whether the same is possible for high-dimensional inputs.

In this paper, we consider the geometric streaming introduced by Indyk in his seminal paper~\cite{indyk2004algorithms}. Concretely, the input is a stream of insertions and deletions of points in $[\cord]^d$ for some integer $\cord$, where $[\cord] = \{1,2,\dots,\cord\}$. The stream defines a subset of vectors $X = \{ x_1, \dots, x_n \} \subset [\cord]^d$, and the goal is to maintain an approximation of the weight of the MST of the points in the stream,
\[ \cost(X) = \min_{T \text{ spanning $X$}} \sum_{(x_i, x_j) \in T} \| x_i - x_j \|_2 \]
 using as little space as possible. 
While the streaming model has been extensively studied and MST is a classic problem, the complexity of approximating MST in a stream remains poorly understood.
Indyk \cite{indyk2004algorithms} first presented an $O(d\log\cord)$-approximation algorithm, by drawing a connection to low-distortion embeddings of $[\cord]^d$ into tree metrics. For high-dimensional regimes (when $d = \omega(\log n)$, making exponential dependencies too costly), these techniques were improved to $O(\log(d\cord) \log n)$ in~\cite{andoni2008earth}, and further improved to $\tilde{O}(\log n)$  in~\cite{chen2022new}. 
The best known lower bound also appears in~\cite{chen2022new}, who proved that for $d$-dimensional $\ell_1$ space, 
where $d= \Omega(\log n)$, any randomized $b$-bit streaming algorithm
which estimates the MST  up to multiplicative factor $\alpha > 1$ with probability at least 2/3 must satisfy
$b \ge \log n/\alpha - \log \alpha$. For the regime we consider
in this paper where $\alpha$ is constant, this is only a $\Omega(\log n)$ lower bound.

On the other hand, when restricted to low-dimensional space, a $(1+\epsilon)$-approximation algorithm is known~\cite{10.1145/1064092.1064116}. The space complexity is $(\log (\cord) / \epsilon)^{O(d)}$, which is only sublinear when $d = o(\log n)$. 
This has left open a fundamental question: does the streaming MST problem suffer from the curse of dimensionality, where it is provably harder for high-dimensional inputs when
compared to low-dimensional inputs? 
As a result, obtaining any sub-logarithmic approximation bounds for the high-dimensional case has been a major
open problem in the field of sublinear algorithms.

In this paper, we aim at answering
the following two questions:
  \begin{quote}
\begin{center}
  {\it \textbf{(1)}
    What is the best approximation possible for Euclidean MST in a one-pass geometric stream? 
    }
\end{center}
\end{quote}

  \begin{quote}
\begin{center}
  {\it \textbf{(2)}
     How many passes are necessary to obtain a $(1+\eps)$-approximation to Euclidean MST in geometric streams?
    }
\end{center}
\end{quote}
   

\subsection{Our Results}
Our main result is a single-pass $(\eps^{-2} \log \eps^{-1})$-approximation algorithm for the weight of the Euclidean MST in $n^{O(\eps)}$ space.  Our estimator is a linear sketch and so applies to the turnstile model of streaming, where points can be both inserted and deleted from the dataset. 

\noindent
\textbf{Theorem} \ref{thm:onePassMain}.
\textit{Let $P \subset [\cord]^d$ be a set of $n$ points in the $\ell_2$ metric, represented by the indicator vector $x \in \R^{\cord^d}$, and fix any $\eps >0$. Then there is a randomized linear function $F:\R^{\cord^d} \to \R^s$, where $s = n^{O(\eps)}$, which given $F(x)$, with high probability returns a value $\ecost(P)$ such that}
\[ \cost(P) \leq \ecost(P) \leq O\left(\frac{\log \eps^{-1}}{\eps^2}\right)\cost(P) \]
\textit{Moreover, the algorithm uses total space $\tilde{O}(s)$.} 

While the memory required by our algorithm is higher than 
the $\polylog(n,d)$ space required by the $\tilde{O}(\log n)$ approximation algorithm of \cite{chen2022new}, 
our algorithm is the first to achieve a constant factor
approximation in sublinear space (namely, only $n^{O(\eps)}$ space). Moreover, we demonstrate that our approximation is tight up to constants, even in the polynomial space regime. Specifically, we rule out any one pass $(1+\eps)$-approximation algorithm for high-dimensional inputs.

\noindent
\textbf{Theorem} \ref{thm:lb}.
\textit{
    There is no one-pass streaming algorithm with memory $o(\sqrt{n})$ 
    that achieves a better than $1.1035$ approximation to the Euclidean minimum spanning tree under the $\ell_2$ metric when
    the dimension is at least $c \log n$, for a large
    enough constant $c$.
    }

Our lower bound for the single pass setting stands in stark contrast  with the $(1+\epsilon)$-approximation for the 
low-dimensional case. Thus, our result shows a strong separation between high-dimensional and low dimensional regimes. Note that such a separation is not known for the offline fine-grained complexity of MST, so designing a $(1+\eps)$-approximation to high-dimensional Euclidean MST is perhaps possible.
Nevertheless, we demonstrate that one can bypass this lower bound and still obtain a $(1+\eps)$-approximation to high-dimensional inputs if we are allowed a constant number of passes over the stream. 

\noindent
\textbf{Theorem} \ref{thm:alphapass}.
\textit{Let $P \subset [\cord]^d$ be a set of $n$ points in the $\ell_1$ metric. Fix $\eps \in (0, 1)$ and any integer $\alpha$ such that $2 \leq \alpha \leq n^\eps$.  Then there exists an $\alpha$-pass $n^{O(\eps)} d^{O(1)}$-space\footnote{The additional $d^{O(1)}$ factor in the space is a result of the fact that reduction to $O(\log n)$ dimensions is possible for the $\ell_2$ metric but not $\ell_1$.} geometric streaming algorithm which, with high probability, returns a value $\ecost(P)$ such that}
\[  \cost(P) \leq\ecost(P) \leq \left(1 + O\left(\frac{ \log \alpha +1}{\eps \alpha}\right)\right) \cdot \cost(P). \]

Our $\alpha$-pass algorithm of Theorem \ref{thm:alphapass} works more generally for the $\ell_1$ metric.  
 Specifically, for any $p \in [1,2]$, one can embed $\ell_p$ into $\ell_1$ with small distortion,\footnote{This is just the Johnson Lindenstrauss transform for $p=2$, and a $p$-stable transform for $p \in (1,2)$. See e.g. Appendix A of \cite{chen2022new} for details.} thus our results extend to $\ell_p$ norms for all $p \in [1,2]$.

All of the aforementioned algorithms are \textit{linear-sketches} (i.e., linear functions of the input). A consequence of linearity is that the sketches can be merged, and thus aggregated together in a constant number of rounds in the massively-parallel computation (MPC) model. Specifically, given a linear sketch algorithm using space $s$ and $R$ passes, one can obtain a parallel MPC algorithm with local space $s^2$, $m = O(N/s^2)$ machines, and $O(R \log_s m)$ rounds (see e.g., ~\cite{AndoniNikolov}), where $N$ is the input size to the problem ($N = nd\log \Gamma$ in our setting).
\begin{corollary}\label{cor:MPC}
Let $P \subset [\cord]^d$ be a set of $n$ points in the $\ell_1$ metric. Fix $\eps \in (0, 1)$ and any integer $\alpha$ such that $2 \leq \alpha \leq n^\eps$. There exists 
an MPC algorithm 
which uses $s := n^{O(\eps)}d^{O(1)}$ memory per machine, 
$O(nd \log \cord /s)$ total
machines and $O(\alpha \epsilon^{-1})$ rounds of 
communication to compute an estimate $\ecost(P)$ which with high probability satisfies
\[  \cost(P) \leq\ecost(P) \leq \left(1 + O\left(\frac{ \log \alpha +1}{\eps \alpha}\right)\right) \cdot \cost(P). \]
\end{corollary}

Corollary \ref{cor:MPC} complements the result of~\cite{AndoniNikolov}
who obtained a $(1+\eps)$-approximation to the MST problem
in a constant number of MPC rounds in\textit{ low-dimensional space} (specifically, they required a per machine memory of $s > 1/\eps^{O(d)}$). Previously, it was a longstanding open problem whether a constant factor approximation could be achieved in a constant number of MPC rounds, and Corollary~\ref{cor:MPC} gives a constant factor approximation to the \emph{cost}. Is it possible for a constant round MPC algorithm using sublinear memory to output a constant-approximate minimum spanning tree?



Finally, by the well-known relationship between the weight of the  minimum spanning tree, the weight of the minimum travelling salesperson tour, and the weight of the minimum Steiner tree (see, e.g.,~\cite{vazirani2001approximation}), our results also allow one to approximate these latter problems in high-dimensional space (up to another factor of two in the approximation guarantee).

\paragraph{Conceptual Contribution and Main Insights.}
The two prior best upper bounds of $O(\log^2 n)$ \cite{indyk2004algorithms,andoni2008earth} and $\tilde{O}(\log n)$~\cite{chen2022new} both followed the approach of randomized tree embeddings to estimate the MST. Specifically, by imposing a recursive, random subdivision of space, one can embed the points in $\R^d$ into a tree metric. The advantage of this embedding is the weight of the MST in a tree metric embeds isometrically into $\ell_0$, which can be estimated in a stream. 

In this paper, we follow a totally different approach, based on estimating the number of connected components in a sequence of auxiliary graphs $G_t$ for geometrically increasing $t$, where two points are connected if they are at distance at most $t$ (inspired from the classic paper \cite{chazellerubinfeld} in the sublinear \emph{time} model). Unfortunately, estimating the number of connected components involves performing a breadth-first search (BFS), which is inherently sequential and therefore requires many passes over the data. A key observation of \cite{czumaj2009estimating} (also exploited in \cite{10.1145/1064092.1064116}) is that $O(\log n)$ steps of a BFS suffice for a constant approximation to MST. At a high level, our main insights are:

\begin{enumerate}
    \item Firstly, we show that if, instead of running a BFS of $O(\log n)$ steps, we ran a shorter BFS of $\alpha$ steps (for $\alpha \leq O(\log n)$), we can still design an estimator which is a $O(\frac{1}{\alpha} \cdot \log n)$-approximation.
    \item A streaming algorithm which seeks to run a BFS around a point $p \in G_t$ for $\alpha$ steps could try to store discretized versions of all points within distance $\alpha t$ from $p$. If the number of such points was less than $n^{O(\eps)}$, we could store them in small space and simulate $\alpha$ steps of the BFS once the stream has completed.
%
    \item Finally, we prove a structural result, demonstrating that if, for each $t$ and for each point $p \in G_t$, we run $\alpha$-steps of the BFS from $p$ for the largest $\alpha$ such that at most $n^{O(\eps)}$ discretized points lie within $\alpha t$ of $p$, then the resulting estimator gives a constant factor approximation. Intuitively, there are not too many instances of $(t,p)$ where $\alpha = o(\log n)$.  
\end{enumerate}


\paragraph{Other Related Work.}
The space complexity of problems in the geometric streaming model often exhibit a dichotomy, where algorithms either use space exponential in the dimension, or suffer a larger approximation factor. 
For instance, with space exponential in the dimension, it is known how to obtain constant or even $(1+\eps)$ approximations to MST \cite{10.1145/1064092.1064116}, facility location~
\cite{lammersen2008facility,czumaj20131},  clustering problems~\cite{frahling2005coresets}, Steiner forest \cite{czumaj2020streaming}, Earth Mover's Distance \cite{andoni2009efficient}, and many others. 
On the other hand, for high-dimensional space one often suffers a $O(d \log n)$ or $O(\log^2 n)$ approximation \cite{indyk2004algorithms,andoni2008earth}. 
Recently, there has been significant effort to design algorithms with improved approximation factors for the high-dimensional regime without exponential space. For instance, in \cite{czumaj2022streaming}, the authors give a two-pass algorithm for facility location with a constant approximation. For clustering problems, such as $k$-means and $k$-median, recent work has developed $(1+\eps)$ approximate coresets in polynomial in $d$ space \cite{braverman2017clustering,hu2018nearly}. The first polynomial space algorithms for maintaining convex hulls and subspace embeddings were given in \cite{woodruff2022high}.

In the distributed computation literature, \cite{andoni2018parallel} demonstrate how to approximate the MST of a generic graph, along with other problems, in $O(\log D)$ rounds, where $D$ is the diameter of the graph, which results in $O(\log n)$-round algorithms for Euclidean MST.
Using the linear sketches of \cite{chen2022new}, one can obtain a $O(1)$-round MPC algorithm with a $\tilde{O}(\log n)$-approximation. However, until now no algorithm achieved both constant rounds and approximation. Previously, such a result was only known in the more powerful Congested Clique \cite{jurdzinski2018mst} model.

%% file: techniques.tex
\subsubsection{Background: Estimating MST via Counting Connected Components}

The central structural fact, which is the starting point for our algorithm, is that the minimum spanning tree of a graph can be approximated by counting the number of connected components of 
a sequence of auxillary graphs.\footnote{See the results of Chazelle~\cite{chazelle00} for arbitrary graphs and Czumaj and Sohler~\cite{czumaj2009estimating} for general metric spaces} Specifically, for a point set $P \subset \R^d$  with pairwise distances in $(1,\Delta)$, and for any $1 \leq t \leq \Delta$, define the $t$-\textit{threshold graph} $G_i = (P,E_t)$ as the graph where $(p,q) \in E_t$ if and only if $\|p-q\|_1 \leq t$, and let $c_t$ be the number of connected components in $G_t$. Consider the steps taken by Prim's algorithm, which adds edges to an MST greedily based on weight, and observe that the number of edges in with weight in $(t,2t]$ is precisely $c_{t} - c_{2t}$, so
\begin{equation}\label{eqn:intro1}
    \cost(P) \leq  \sum_{i=0}^{\log (\Delta)-1}2^i (c_{2^i} - c_{2^{i+1}}) = n - \Delta + \sum_{i=0}^{\log (\Delta) -1} 2^i c_{2^i} \leq 2 \cost(P) 
\end{equation}
Thus, it will suffice to obtain an estimate of $c_t$ for each $t \in \{1,2,4,\dots,\Delta\}$. For a vertex $p$ in a graph $G$, let $\CC(p,G)$ be the connected component in $G$ containing $p$.  Towards this end Chazelle, Rubinfeld, and Trevisan~\cite{chazellerubinfeld} proposed the following estimator for $c_t$:
\begin{enumerate}
    \item Sample a point $p \sim P$ uniformly.
    \item \label{intro:step2} Set $x_t(p) = \frac{1}{|\CC
    (p,G_t)|}$, and output $n \cdot x_t(p)$.\footnote{Note that the estimators of \cite{chazellerubinfeld, czumaj2009estimating} output instead an indicator random variable with roughly the same expectation, but this difference is irrelevant for us here. }
\end{enumerate}
It is easy to see that $\exx{p \sim P}{n \cdot x_t(p)} = c_t$. To implement the second step, one must run a breadth-first search (BFS) from $p$ to explore $\CC(p,G_t)$. However, a full BFS is costly, and each step of the BFS requires a full pass over the data in the streaming model.
An observation of ~\cite{chazellerubinfeld,czumaj2009estimating} is that once the BFS explores more than $X$ vertices, then $x_t(p) \leq 1/X$, so outputting $0$ whenever this occurs changes $\ex{n \cdot x_t(p)}$ by at most $n/X$. If the maximum distance in $P$ was exactly $\Delta$, then $\cost(P) \geq \Delta$, so setting $X = 2 \Delta$ results in an additive error of $\cost(P)/2$ to the estimator (\ref{eqn:intro1}), so the BFS can always be halted after $O(\Delta)$ steps. 

The key insight of \cite{czumaj2009estimating} is that, for \textit{metric} MST, one can get away with halting the BFS after $\alpha = O(\log \Delta)$ steps. Observe that if $x,y \in P$ are two points added on steps $s_1 \leq s_2$ of the BFS, repsectively, with $s_1 \leq s_2 - 2$, then $\|x-y\|_1 > t$ (otherwise $y$ would have been added on step number $s_1+1$). So after $\alpha$ steps of a BFS, we obtain a set of $\alpha/2$ points within $ C(p,G_t)$ which are pairwise distance at least $t$ apart. Any MST on these points must have cost at least $t \cdot  \alpha /4$, so letting $n_i$ be the number of connected components with diameter larger than $ \alpha$ in $G_t$, we have $\cost(P) \geq   (t/4) \cdot  n_i \cdot \alpha$.  Summing over all levels yields:
\[  \sum_{i=0}^{\log (\Delta) -1} 2^i \left(c_{2^i} - n_i \right) \geq  \sum_{i=0}^{\log (\Delta) -1} 2^i c_{2^i}  -  \frac{4 \log \Delta}{\alpha} \cost(P) \]
So for $\alpha = O(\log \Delta)$, if the algorithm instead outputs $0$ on Step \ref{intro:step2} whenever $\CC(p,G_t)$ cannot be explored in $\alpha$ steps of a BFS, then we can still obtain a constant factor approximation. Intuitively, for a connected component $C$ in $G_t$, we want to contribute an additive $t$ to (\ref{eqn:intro1}), but if the diameter of $C$ is large, then the cost of the MST edges within $C$ is much larger than $t$, so dropping $C$'s contribution to (\ref{eqn:intro1}) is negligible.



\subsubsection{A $O(\log \Delta$)-Pass Streaming Algorithm}
The above motivates the following approach for a $O(\log \Delta)$ pass streaming algorithm. For simplicity, consider the insertion-only streaming model, where on each pass the points in $P$ are added in an arbitrary order.\footnote{For the case of our $\alpha$-pass algorithm where $\alpha \geq 2$, generalizing to the \textit{turnstile} model, where points can be inserted and deleted, is relatively straightforward using known linear sketches for sparse recovery and uniform sampling. } On the first pass, we can sample the point $p \sim P$ (e.g., via reservoir sampling). Then, on each subsequent pass, we can take a single step of the BFS, by storing every point $p \in  P$ adjacent to a point stored on prior passes. 
The key issue with this approach is that is the \textit{variance} of this estimator is too large. Namely, if $P = P_1 \cup P_2$ where $P_1$ contains $n- \sqrt{n}$ points  pair-wise distance $1$ apart, and $P_2$ contains $\sqrt{n}$ points at pairwise distance $n$ apart, then clearly, we need to sample points $p \in P_2$ to obtain an accurate estimate, but doing so would require $\Omega(\sqrt{n})$ samples (and thus $\Omega(\sqrt{n})$ space). This sample complexity was acceptable in \cite{czumaj2009estimating}, since they worked in a query model and targeted $\tilde{O}(n)$ distance queries, however this will not suffice for streaming. 

\paragraph{Quadtree Discretization.} To handle the issue of large variance, we would like to ensure that large clusters of close points are not sampling with too large probability. To do this, we first perform a discretization step based on randomly shifted hypercubes (this approach is widely known as the \textit{Quadtree decomposition} \cite{IT03, indyk2004algorithms, andoni2008earth}), and map each point to the centroid of the hypercube containing it. Instead of performing the BFS on the original point set $P$ for each $t=2^i$, we define the discretized vertex set $V_t$ to be the set of centroids of the hypercubes with diameter $\delta t$, for some small $\delta >0$ and work instead with the $t$ threshold graph of $V_t$. This mapping moves points by only a small $\delta t$ distance, and has the affect of collapsing large clusters with a small diameter. The number of vertices in this graph is now $|V_t|$, thus $|V_t| \cdot x_t \in [0,|V_t|]$ is now a bounded random variable, so by taking $O(1/\epsilon^2)$ samples at each level, we guarantee a total error in the estimator of 
\[ \epsilon  \left(\sum_{i=0}^{\log \Delta - 1} 2^i |V_{2^i}|  \right)\]
By a classic result \cite{indyk2004algorithms} on the Quadtree, the term in parenthesis can be bounded above by $O( d \log \Delta \cost(P))$, thus setting $\eps = O( \frac{1}{d \log \Delta})$ yields a $O(\log \Delta)$-pass constant factor approximation algorithm.\footnote{This result implies a $O(d \log \Delta)$ approximation to MST, which can be made $O(\log n \log \Delta)$ for $\ell_2$ space. The equivalent result for $\ell_1$ was given in \cite{andoni2008earth}.}


\subsubsection{Towards Reducing the Number of Passes}
The key contribution of this work is to demonstrate that the BFS in 
the above algorithm can be (roughly) simulated using a small number 
of rounds of adaptivity (i.e., small number of passes of 
the stream). Firstly, notice that $\Omega(\log \Delta)$ 
steps of a BFS is tight for the above estimator. 
Specifically, consider the following ``Cantor'' style point 
set $P  \subset \R$ in the line, constructed as follows. 
Starting with $n$ points $x_1,\dots,x_n$, pair up 
consecutive points $\{x_i,x_{i+1}\}$, and connect them 
with an edge of weight one. Then pair up consecutive 
sets of pairs from the last level $\{x_i,x_{i+1}\}, 
\{x_{i+2},x_{i+3}\}$, and connect the middle two 
$x_{i+1},x_{i+2}$ with weight $2$. Afterward, we pair up 
consecutive sets of four points, connect the middle vertices with an edge of weight $4$, and so on (see Figure \ref{fig:Example}).
For any $0 \leq i <\log n$, the resulting graph has $n/2^{i+1}$ connected components in $G_{2^i}$, and each connected component has (metric) diameter $2^i (i+1)$, and thus (graph) diameter $(i+1)$. The total MST cost is then $\Theta(n \log n)$, with each of the $\log n$ levels contibuting equally, so we cannot ignore components with (graph) diameter larger than $\eps \log n$, since we would lose all but an $\eps$ fraction of the MST cost!


\begin{figure}
    \centering
    \includegraphics[scale = .35]{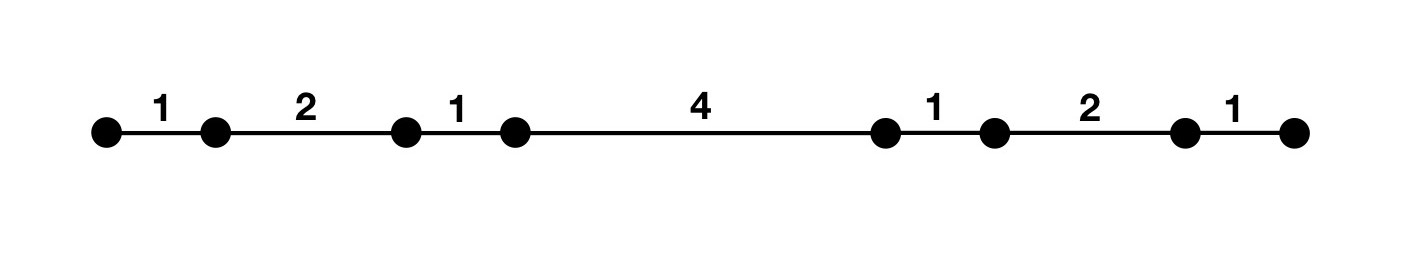}
    \caption{Construction of a challenging instance, where the connected components in $G_t$ have diameter $\log(t)$, and each level contributes equally to the MST cost. }
    \label{fig:Example}
\end{figure}

Now a tantalizing approach for reducing the number of passes would be attempting to accomplish multiple rounds of the BFS on a single pass over the stream. Specifically, since we only need to explore a connected component if it has diameter $O(\log \Delta)$, when we sample $p \in \CC(p,G_t)$ in a small diameter component of $G_t$, we know that $ \CC(p,G_t) \subseteq \cB(p, O(\log \Delta)\cdot t)$, where $\cB(p,r)$ is the \textit{metric}-ball of radius $r$ centered at $p$. We could, on a single pass, store all of $ \cB(p, O(\log \Delta)\cdot t)$ and recover the component $\CC(p,G_t)$. This would resolve the instance from above, since after discretizing to $V_t$, each component would only contain $\polylog(n)$ points. 

The difficulty, of course, is that $ \cB(p, O(\log \Delta)\cdot t)$ could contain too many points to store in small space. In particular, the ball could contain $\poly(n)$ distinct connected components from $G_t$ different from $\CC(p,G_t)$, so we cannot afford to output $x_t = 0$ in this case since we would lose the contribution of all these components without necessarily having a larger cost within the components to compensate. 
For example,  in the above instance, if after some level $i = \Omega( \log n)$, instead of merging together two connected components at a time we grouped together $n^{\Omega(1)}$ components,  then $\cB(0,\log \Delta \cdot t)$ would contain all such points, which would be too many to store.

\paragraph{A $\log \Delta$ approximation algorithm.} 
Thus, in two passes alone, we cannot hope to always recover connected components of diameter $\alpha$, even for $\alpha = O(1)$. In effect, the most we can hope to recover is a single step of the BFS is the ball $\cB(p, t)$ (i.e., the immediate neighborhood of $p$ in $G_t$). Our next algorithmic insight is that given only a single step of a BFS, we can still recover a $O(\log \Delta)$ approximation. Specifically, we prove that if on Step \ref{intro:step2} of the algorithm, after sampling $p$ we instead output $|V_t| \cdot y_t^0(p) $ where $y_t^0(p) = \frac{1}{|\cB(p,t) \cap V_t|}$, then we have 

\begin{equation}\label{eqn:overEstimate}
      \sum_{i=1}^{\log \Delta - 1}2 ^i c_{2^i} \leq \sum_{i=1}^{\log \Delta - 1}2^i |V_t| \cdot \ex{y_t^0(p)} \leq \sum_{i=1}^{\log \Delta - 1}2^ic_{2^i}   + O(\log \Delta) \cost(P)
\end{equation}

Note that while the estimator of \cite{czumaj2009estimating} is always an under-estimate (it sometimes sets $x_t = 0$ when it cannot recover the full component), the estimator using $y_t^0(p)$ of (\ref{eqn:overEstimate}) is always an overestimate, since $|\cB(p,t) \cap V_t|$ can only be smaller than $|\CC(p,G_t)|$. In fact, if for any $j \geq 0$ we define 
\[ y_t^j(p) = \frac{1}{|\CC(p,G_t \cap \cB(p,2^j \cdot t))|}\] we prove a more general statement that bounds how badly $y_t^j$ overestimates:

\noindent
\textbf{Lemma } \ref{hehelemma1} {\it
For any integer $j \geq 0$, and $t \in (1,\Delta)$, we have
\[ 0 \leq  t \sum_{p \in V_{t}} \left(y_{t}^j(p) - x_{t}(p)\right) \leq O\left( \frac{1}{2^j}\right) \cost(P) \]
}

Summing over $O(\log \Delta)$ values of $t$ is sufficent to prove (\ref{eqn:overEstimate}). 
For the case of $j=0$, the high-level approach of the proof is as follows.  We begin by fixing a maximal independent set $\cI_t$ of size $k$ in the $t/2$-threshold graph on the vertex set $V_t$.
Thus, for any two $p,q \in \cI_t$ we have $\|p-q\|_1 > t/2$, and each $p \notin  \cI_t$ is distance at most $t/2$ to some $q \in \cI_t$. By assigning each $p \notin \cI$ to such a close point in the independent set, we partition $V_t$ into clusters $C_1,\dots,C_k$ with (metric) diameter at most $t$, and such that the centers of separate clusters are separated by $t/2$. Then for any $i \in [k]$ and $p \in C_i$, we have that $C_i \subset \cB(p,t)$, and so $y_t^0(p) \leq 1/|C_i|$. It follows that 
\begin{equation}\label{eqn:tech2}
     \sum_{p \in V_{t}}  y_{t}^0(p) = \sum_{i \in [k]} \sum_{p \in C_i} y_{t}^0(p) \leq k 
\end{equation}
On the other hand, since $|\cI|= k$, there are $k$ points in $P$ which are pairwise distance $t/2$ apart, thus $\cost(P) \geq (t/4)k$, which can be combine with (\ref{eqn:tech2}) to complete the proof. 
For $j > 0$, we prove a new graph decomposition Lemma (Lemma \ref{lem:graphDecomp}), which allows us to partition each connected component into $k$ clusters $C_i$, each of diameter $\Theta(2^j)$, and such that there is an independent set of size at least $\Omega(k \cdot 2^j)$, and the proof then proceeds similarly. 
This yields a two-pass algorithm with a $O(\log \Delta)$ approximation, which is already an improvement over the $\tilde{O}(\log n)$ approximation of \cite{chen2022new} (which, however, is a \textit{one-pass} algorithm). 

\subsubsection{A Two-Pass Constant Factor Approximation Algorithm.} 
The crux of our algorithm is to combine the above approaches. For a point $p \in V_t$, let $B_t^*(p) = \cB(p, O(\log \Delta)\cdot t) \cap V_t$, and $B_t(p) =\cB(p,  t) \cap V_t$. We
\textbf{1)} attempt to store the full ball $B_t^*(p)$ to simulate the $O(\log \Delta)$-pass algorithm  and \textbf{2)} if $B_t^*(p)$ is too large to store, we simply output $y_t^0(p)$. If we can show that \textbf{2)} does not occur too often, hypothetically, if it only occurred restricted to a small number $\ell$ of levels,
then we will obtain $\ell$-approximation since we pay an additive $O(\cost(P))$ for each such level by Lemma \ref{hehelemma1}.\footnote{Note that it will not be the case that an entire level is always either \textit{good} (always possible to store $B_t^*(p)$) or \textit{bad} (always output $y_t^0(p)$). Bad points can be distributed across all levels and across the point-set. Handling this is a key challenge in our analysis.}   Since our space bound is $n^{O(\eps)} d^{O(1)}$, we will say that a set $B$ is large (at level $t$) when $|B| > d \log \Delta \cdot n^{2 \eps}$.

To ensure that $B_t^*(p)$ is not large too often, we must first relate the size of $B_t^*(p)$ to the cost of the MST. Luckily, this has already been done for us by our discretization. Namely, one can show that, in expectation over the random draw of the Quadtree, we have $\cost(B_t^*(P))  \geq O( \frac{t}{d \log \Delta} ) |B_t^*(p) |$. Now we cannot set $y_t(p) = 0$ simply because $B_t^*(p)$ is large, since there can be \textit{many} connected components in $G_t$ contained in $B_t^*(p)$. 
On the other hand, if instead the set $B_t(p)$ was large, then since $B_t(p) \subseteq \CC(p,G_t)$, this provides a certificate that the MST cost of $\CC(p,G_t)$ is significantly larger than $t$. By the same reasoning that allowed us to output $0$ whenever a connected component had (graph) diameter larger than $O(\log \Delta)$, we can similarly output $y_t(p) = 0$ whenever $B_t(p)$ is large.
This motivates the following full algorithm for level $t$, which we will show achieves a $O(\log \log \Delta)$ approximation: 

 \begin{enumerate}
     \item  Sample $p \sim V_t$ uniformly.  If $|B_t(p)| > d \log \Delta n^\eps$, output $y_t(p) =0$.
    \item Else, if $|B_t^*(p)| <d \log \Delta  n^{2 \eps}$, recover the whole ball and simulate running the $O(\log \Delta)$-pass algorithm (in one pass).
    \item \label{step:3Intro} Otherwise, set $y_t(p) = \frac{1}{|B_t(p)|}$, and output $|V_t| \cdot y_t(p)$
\end{enumerate}

The key challenge will be to bound the cost of the points for which $B_t^*(p)$ is large and for which $|B_t(p)| \leq d \log \Delta  n^\eps$,  because these are the only points for which we overestimate $x_t(p)$ by setting $y_t(p) = y_t^0(p)$ in Step \ref{step:3Intro} above. To do this, we observe that if $B_t^*(p)$ was large, then it should also be the case that $B_{\log \Delta t}(p)$ is large. One must be slightly careful with the definition of large here,  since $V_t \neq V_{\log \Delta t}$, as the latter uses a coarser discretization. Nevertheless, this will not cause a serious issue for us, and for simplicity we can assume here that for any $r \geq 0$, if $|\cB(p, r) \cap V_t| > T$, then $|\cB(p,  \lambda \cdot r) \cap V_{\lambda t}| > T/\lambda$ for any $\lambda > 1$. 

Given the above fact, the high-level idea for the remaining argument is as follows. We now know that if $B_t^*(p)$ was large, then we should have $|B_{\lambda t}(p)| > d \log \Delta n^\eps$ for every $\lambda \in (\log \Delta ,n^\eps)$, and then $y_{\lambda t} (p) = 0$ for all such points. So if $B_t^*(p)$ is large, we may overpay at levels $t,2t,\dots,\log \Delta \cdot  t,$ but after that we cannot overpay again for the point $p$ until level $n^\eps \cdot t$. Since one can  assume $\Delta = O(n)$, it follows that any one point $p$ can only overpay at most $O( \frac{1}{\eps }\log \log \Delta)$ times. This seems like a promising fact, given that we hope to obtain a  $O( \frac{1}{\eps }\log \log \Delta)$ approximation.  Let us call a point $p \in V_t$ \textit{dead} if $|B_t(p)| > d \log \Delta n^\eps$; we do not overpay for such points. We call a point $p \in V_t$ \textit{bad} if $p$ is not dead and $B_t^*(p)$ is large. Our goal is to prove the following:
\begin{equation}\label{eqn:introloglog}
    \sum_{i=0}^{\log \Delta -1} 2^i \sum_{\substack{p \in V_{2^i} \\ p \text{ is bad}}} y_t^0(p) \leq O\left( \frac{1}{\eps} \log \log \Delta \right) \cost(P)
\end{equation}

To prove (\ref{eqn:introloglog}), in Section \ref{sec:structure} we prove a key structural result, which allows us to relate the cost of the bad points to the MST cost. Specifically, we can partition the set of bad points in $G_t$ into clusters $C_1,\dots,C_k$, so that each cluster has diameter at most $t$. Let $N_t$ be the set of centers of these clusters $C_i$. Similarly to how we argued above in the proof of Lemma \ref{hehelemma1}, the cost of the bad points in $G_t$ is at most $k$. Our structural result, in essence, demonstrates that, for any possible such sets $N_{2^i}$ of pair-wise separated centers, we have $\sum_i 2^i |N_{2^i}| = O(\eps^{-1}\log \log \Delta \cost(P))$. The proof utilizes the fact that bad points $p,q$ at separate levels $t > t'$ must be distance at least $\Omega(t)$ apart, otherwise the higher point $p$ would be dead. The separation between points allows us to prove a lower bound on the MST cost as a function of the sizes $|N_t|$, which will yield the desired approximation.

\paragraph{From $\log\log \Delta$ to a constant, and from a constant to $(1+\epsilon)$.}
The $O(\log \log \Delta)$ factor in the above approximation is the result of the fact that, if $t$ is the smallest level where $p$ is bad, then we will over pay for each of the $O(\log \log \Delta)$ levels between $t$ and $t \log \Delta$. However, notice that since $t$ was the \textit{smallest} such level that was bad, this means that $|\cB(p, \frac{\log \Delta}{2} t) \cap V_t| <d \log \Delta n^{2\eps }$ (otherwise $t/2$ would have been the smallest level). In this case, we could recover the whole ball of radius $\frac{\log \Delta}{2}t $ around $p$, and simulate running 
 $\frac{\log \Delta}{2}$-passes of the BFS. Using Lemma \ref{hehelemma1}, this will give us only an additive $O(1) \cdot \frac{2}{\log \Delta} \cost(P)$ overestimate. Similarly, at the level $2^i t$, for any $0 \leq i \leq \log \Delta$, we can recover the same ball $\cB(p, \frac{\log \Delta}{2} t)$, and simulate running $\frac{\log\Delta }{2^i}$-steps of BFS in $G_{2^i t}$ to obtain an additive $O(1) \cdot \frac{2^i}{\log \Delta} \cost(P)$ over-estimate. 
The total error is geometric and sums to at most $O(\cost(P))$. Thus, our full algorithm will attempt to recover the balls $|\cB(p, 2^i t) \cap V_t|$ for each $i=1,\dots,\log \log \Delta$. We find the largest $i$ for which the ball is not large, and  then simulate the $\log(\Delta)/2^i$-pass algorithm for this $i$. Our proof will require a finer-grained notion of a point being bad, and our structural result in  Section \ref{sec:structure} is equipped to handle this generalized notion. To obtain a finer approximation using a larger $\alpha > 2$ passes, we note that in $\alpha$ passes we can always at least run $\alpha-1$ steps of the BFS. Thus, even at the smallest level $t$ where $p$ is bad, we output at least $y_t^{\log \alpha}(p)$. Thus, the same geometric sum is at most $O(\cost(P) / \alpha)$, which will result in an additive overestimate of $O(\cost(P) / \alpha)$.

\subsubsection{A One-Pass Algorithm and Recursive Precision Sampling}
The above $O(1/\eps)$-approximate algorithm can be easily implemented in two-passes over a stream (even in the turnstile model), where on the first pass we sample the point $p \sim V_t$ uniformly, and on the second pass we try to store all the points in the ball $\cB(p,2^i t)$ for $i \in [\log \log \Delta]$. However, compressing both steps into a single pass is quite non-trivial, since we cannot know which ball to begin collecting points from until $p$ is sampled. Moreover, in the linear sketching model, where we store only a linear function of the indicator vector $x \in \R^{\cord^d}$ for the point-set $P$, one cannot try something clever like sampling a point from the first portion of the stream, and collecting points from the second portion. 

It is worth noting that it is quite common to have a geometric estimator which, on the first pass generates some sample, and on the second pass recovers some information about that sample. This was the case, for instance, in recent prior works on geometric streaming \cite{chen2022new,czumaj2022streaming,DBLP:conf/icalp/CharikarW22}. The issue of compressing such estimators to a single pass has proven quite challenging, and generic techniques remain elusive. 
In this work, we make progress toward providing general-purpose techniques for achieving such compression results. 

\paragraph{A One-Pass Friendly Estimator.}
To begin, one first needs to design an estimator which is more amenable to one-pass algorithms. The whole of Section \ref{sec:estimatorOnePass} is devoted to this goal. Recall that our earlier estimator samples a point $p \sim V_t$, and then recovers information about $\log \log \Delta$ balls around that point. For the moment, suppose we only needed to recover information about the single ball $\cB(p,\lambda t)$ for some $\lambda \geq 1$. Suppose further that we could partition the points $V_t$ into clusters $C_1,\dots,C_k$, such that for every $i \in [k]$ and $p \in C_i$, we have the containment $\cB(p,\lambda t) \subset C_i$. Given this, we could hope to sample a cluster $C_i$ with probability proportional to $|C_i|$, and then recover all of $C_i$. When recovered, we could then sample a point $p \sim C_i$, and by the structure of the partition recover $\cB(p,\lambda t)$. 

Now of course such a partition will not exist in general. However, a standard approximate method for capturing the ball $\cB(p, \lambda t)$ inside a cluster is \textit{locality-sensitive hashing} (LSH). Specifically, these are hash families $\cH$ for which $h(x) = h(y)$ with good probability for points $x,y$ at distance at most $\lambda t$, and $h(x) \neq h(y)$ whenever $x,y$ are distance larger than $\lambda t/\epsilon$. However, three key issues exist with the majority of such hash functions. Firstly, the good event that  $h(x) = h(y)$ when $x,y$ are close will only occur with a small $n^{-\poly(\eps)}$ probability. Moreover, this is known to be tight for many metrics, including $\ell_1,\ell_2$ \cite{o2014optimal,andoni2017optimal}.  Secondly, even with this small probability, it is not guaranteed that the entire ball around a point $p$ is recovered. Thirdly, while we can repeat the hashing $n^\eps$ times to attempt to recover a successful attempt, given a hash bucket $h(x)$ there is, in general, no method for determining which of the repetitions was successful. The last is a significant issue for us: if we fail to recover even a single point in the ball, this could drastically effect our estimator, since removing a single point could disconnect a large connected component into many smaller parts. 

To handle this issue, in Section \ref{sec:LSH}, we design a family $\cH_{t, \eps}$ of LSH functions with \textit{verifable recovery} for $d$-dimensional $\ell_1$ space, that simultaneously solves these issues. The family has the property that,  \textbf{1)} for any $p \in V_t$, with probability at least $(1-\eps)^d$ over $h \sim \cH_{t, \eps}$ the hash bucket $h(p)$ contains all of $\cB(p,t)$, \textbf{2)} the bucket $h(p)$ has diameter $t/\eps$, and \textbf{3)} given a single repetition $h$, one can check where condition \textbf{1)} holds for $p$. The $(1-\eps)^d$ success probability is exponential in $d$, but if $P \subset (\R^{d}, \ell_2)$, then via dimensionality reduction, we can always embed $p$ into $(\R^{O(\log n)}, \ell_1)$ with at most a constant factor distortion of any distance. Thus, for the $\ell_2$ metric, we have $(1-\eps)^d \geq n^{O(\eps)}$ as needed. The problem of designing such an LSH for $(\R^{d}, \ell_1)$ for any dimension $d$ with success probability $n^{\eps}$  is an interesting open question, and such a hash function would immediately generalize our one-pass algorithm to high-dimensional $\ell_1$ space. 

We will use the hash family $\cH_{t, \eps}$ from Section \ref{sec:LSH} to approximately implement our $2$-pass algorithm in one pass. To do so, however, we will need to recover information from more than a single ball $\cB(p,\lambda t)$ containing $p$. Naively, one must recover information about $O(\log \log \Delta)$ balls. Given a point $p$, we say that $p$ is type $j$ if $|\cB(p ,2^{j+1} t)| > d \log \Delta n^{2\eps} $ but $|\cB(p ,2^j t)| \leq d \log \Delta n^{2\eps}$. Notice that, given the sizes of the balls $\cB(p ,2^{j+1} t)$ and $\cB(p ,2^{j} t)$,we can determine whether or not $p$ was of type $j$. Our solution is to \textit{guess} the type $j \in [\log \log \Delta]$ of a point. Given a guess $j$, we draw hash functions $h_1 \sim \cH_{2^j t, \eps}, h_2 \sim  \cH_{2^{j+1} t /\eps, \eps}$.\footnote{Note that the additional $(1/\eps)$ factor of slack in the second family $  \cH_{2^{j+1} t /\eps, \eps}$ is necessary to ensure that the larger hash bucket contains the smaller, since these hash functions incur a $1/\eps$ error. This additional $O(1/\epsilon)$ factor is the cause for our one-pass algorithm only achieving a $\tilde{O}(\eps^{-2})$ approximation instead of a $O(\eps)$. }   The function $h_2$ partitions $V_t$ into buckets $C_1,\dots,C_{k_2}$ based on hash value. We can further partition  each $C_i$ into $C_{i,1},\dots,C_{i,k_1}$ based on the hash value from $h_1$. Given this, we will then try to sample a bucket $C_i$ with probability $\propto |C_i|$. Given $i$, we then sample $C_{i,j}$  with probability $\propto |C_{i,j}|$. Intuitively, for a point $p \in C_{i,j}$, the bucket $C_i$ is an approximation of $\cB(p, 2^{j+1} t/\eps)$, and $C_{i,j}$ is an approximation of  $\cB(p, 2^{j} t)$.
If both hash functions $h_1,h_2$ successfully recovered the relevant ball around $p$ (which occurs with probability $(1-\eps)^{2d} = n^{-O(\eps)}$, and which we can check by the construction of the LSH), then given the sizes of these buckets we can verify if our guess for the type of $p$ was correct. If it was, we can then recover the connected component within the smaller bucket, and simulate our two-pass estimator. Analyzing the effect of undetected incorrect guesses (which can occur due to the $O(1/\eps)$ distortion in the hash functions) is quite involved, and involves detailed casework to handle.

\paragraph{Recursive Precision Sampling. }
Observe that, given the above one-pass friendly estimator, if we can sample $n^{O(\eps)}$ uniform points from both $C_i$ and $C_{i,j}$, then we can verify whether or not $j$ was the correct level. Moreover, if it was correct, with this many samples, by a coupon collector argument,  we fully recover $C_{i,j}$, and can run the estimator. 
Therefore, the above estimator motivates the following \textit{recursive} $\ell_p$ sampling problem: given a vector $x \in \R^{m^3}$, indexed by $(i,j,k) \in [m]^3$, for any $p \in (0,2]$,  design a small-space linear sketch of $x$, from which we can perform the following:
\begin{enumerate}
\item Sample $i_1 \in [m]$, where $\pr{i_1 = i} \propto \sum_{j,k} |x_{i,j,k}|^p$ for every $i \in [m]$.
\item Given $i_1$, generate $s$ samples $(i_2^1,\dots,i_2^s)$, where $\pr{i_2 = j} \propto \sum_{k} |x_{i_1,j,k}|^p$ for every $j \in [m]$.
\item Given $(i_1, i_2^1,\dots,i_2^s)$, for each $\ell \in [s]$ generate $s$ samples $(i_3^{1,\ell},\dots,i_3^{s,\ell})$, where  $\pr{i_3 = k} \propto  |x_{i_1,j_2^\ell ,k}|^p$ for every $k \in [m]$.
\end{enumerate}
Let $x_i \in \R^{n^2}, x_{i,j} \in \R^n$ be the vectors induced by fixing the first one (or two) coordinates of $x$. Note that in our setting, the coordinates $x_{i,j,k}$ correspond to the points in $V_t$,  $x_i$ corresponds to the hash bucket $C_i$, and $x_{i,j}$ corresponds to $C_{i,j}$. Since we ultimately want uniform samples of the non-zero coordinates (and not $\ell_p$ samples), we set $p=O(1/\log n)$ in the above sampling scheme to approximate this uniform distribution. Then setting the number of samples $s = n^{O(\eps)}$ is sufficient to implement our one-pass estimator. We remark that a similar problem of multi-level $\ell_p$ sampling was considered in \cite{chen2022new}, however, their solution was highly-tailored to their particular algorithm, and does not allow for generic recursive sampling, as is needed for our application. 

Section \ref{sec:onePassSketch} is devoted to solving the above sketching problem. Specifically, we prove a generic result, Lemma \ref{lem:mainLinearSketch}, which solves the above problem in $\poly(\frac{1}{p}, \log n , s^2)$ space. We focus on the first two sampling steps (sampling $i_1$, and $i_2$ conditioned on $i_1$), as the third level will be similar. 
Our approach is based on developing a recursive version of the precision sampling technique due to \cite{andoni2011streaming}.
At a high-level, given a vector $z \in \R^n$,  the approach of precision sampling is to scale each coordinate $z_i$ by a random variable $1/t_i^{1/p}$. For our purpose, following \cite{jayaram2021perfect}, we take $t_i^{1/p}$ to be an exponential random variable. By standard properties of the exponential distribution, we have

\[      \pr{i = \arg\max_{j \in [n]} \left\{ \frac{|z_i|}{t_i^{1/p}} \right\} } = \frac{|z_i|^p}{\|z\|_p^p}\]
Thus, precision sampling reduces the problem of sampling to the problem of finding the largest coordinate in a scaled vector. The latter is the focus of the \textit{sparse-recovery} literature, and good sketches exist which can accomplish this task (such as the Count-Sketch \cite{charikar2002finding}). 
The key challenge with solving the above recursive sampling problem is that the second-level samples $i_2$ must be generated \textit{conditioned} on the first sample $i_1$, namely $i_2$ is sampled from the vector $x_{i_1} \in \R^{n^2}$. Since we need two samples, we would like to use two precision sampling sketches, each using a distinct set of exponentials, to recover $i_1,i_2$. In the first sketch, we scale each block $x_i$ by $1/t_i$, in the second sketch we scale each sub-block $x_{i,j}$ by $1/t_{i,j}$. However, the sketches cannot be totally independent from each other, since otherwise if  $(i^*,j^*) =\argmax_{(i,j)} \frac{\|x_{i,j}\|_p^p}{t_{i,j}}$, then we cannot guarantee that $i^* = \argmax_{i} \|x_{i}\|_p^p/t_{i}$ is from the sample block as $i_1$. 
The solution to this dilemma is to also scale the second sketch by the exponentials from the first. Namely we first compute $i_1^* = \argmax_{i} \frac{\|x_{i}\|_p^p}{t_{i}}$, and then set 

\[ i_2^* = \argmax_{j} \frac{\|x_{i_1^*,j}\|_p^p}{t_{i_1^*} t_{i_1^*,j}} \]

In other words, in our sparse recovery procedure, we only search through indices $(i,j)$ such that $i = i_1^*$, and all such indices are scaled by the same $t_{i_1^*}$. The challenge is to show that $\frac{\|x_{i_1^*,i_2^*}\|_p^p}{t_{i_1^*} t_{i_1^*,i_2^*}}$ is large enough to be recovered from the noise from the rest of the sketch. The original arguments from the precision sampling literature \cite{andoni2011streaming, Jowhari:2011, jayaram2021perfect} guarantee that with constant probability:
\[      \argmax_{i \in [n]}  \frac{\|x_{i}\|_p^p}{t_{i}} = \Theta(\|x\|_p^p), \qquad  \text{and }  \qquad \sum_{i \in [n]}  \frac{\|x_{i}\|_p^p}{t_{i}} = \Theta(\|x\|_p^p)   \]
Our main contribution is to demonstrate that a similar fact is true for the recursive sketches. This requires a judicious choice of random events to condition on, and an analysis of their effect within a sparse recovery sketch.

A final challenge is that $\|x_{i}\|_p^p$ is not a linear function of the stream, i.e. there is no (oblivious) linear function that transforms $x$ into a vector $z$ such that $z_i = \|x_i\|_p^p$; the existence of such a $z$ to apply precision sampling to is implicitly assumed in the above discussion. We resolve this issue with a trick using the $p$-stable sketch of Indyk \cite{indyk2006stable}, by repeating each precision sampling sketch with the same exponentials, but with different $p$-stable random variables, and taking the point-wise median of each sketch. This allows us to effectively replace a block of coordinates with their $\ell_p$ norm, so that we can scale the whole block by the same exponential.

%% file: prelims.tex
\paragraph{Basic Notation for Graphs and Metric Spaces.} For any integer $n \geq 1$, write $[n] = \{1,2,\dots,n\}$, and for two integers $a,b \in \Z$, write $[a:b] = \{a,a+1,\dots,b\}$.  For $a,b \in \R$ and $\eps \in (0,1)$, we use the notation $a = (1 \pm \eps) b$ to denote the containment of $a \in [(1-\eps)b , (1+\eps)b]$. 

For any metric space $\cX=(X,d_x)$, point $p \in X$, and $r \geq 0$, define the closed metric ball $\cB_\cX(p,r) = \{y \in X \; | \; d_X(x,y) \leq r \}$. If the metric space $\cX$ is understood from context, we simply write $\cB_\cX(p,r) = \cB(p,r)$. For a subset $S \subset X$, we write $\cB_\cX(p,r,S) = \cB_{\cX}(p,r) \cap S$.  
For a set of points $S$ living in a metric space $(X,d_X)$, let $\diam(S) = \max_{x,y \in S} d_X(x,y)$. For an unweighted graph $G$ and a connected subgraph $\cC \subset G$, let $\diam_G(\cC) = \diam(S)$ where the underlying metric space is the shortest path metric in $G$.  
For a graph $G = (V,E)$ and a subset $S \subseteq V$, let $G(S)$ denote the subgraph of $G$ induced by the vertices $S$. 

Let $P \subset \R^n$ be an $n$-point set such that $\min_{x,y \in P} \|x-y\|_1 \geq 1$ and $\max_{x,y \in P} \|x-y\|_1 \leq \Delta$. 
For any $t \geq 0$, define the \emph{$t$-threshold graph} $G_t^* = (P,E_t^*)$ via $(x,y) \in E_t^* \iff \|x-y\|_1 \leq t$. If we define $c_t^*$ to be the number of connected components in the graph $G_t^*$, we have the following:

\begin{fact}[Corollary~2.2 in~\cite{czumaj2009estimating}]\label{fact:cs}
For any $\eps > 0$ and any set of points $P \subset \R^d$ whose pairwise $\ell_1$-distances are all in $[1, \Delta]$ with $\Delta = (1+\eps)^h$,
\begin{align} 
\cost(P) \leq n - \Delta + \eps \sum_{i=0}^{h-1} (1+\eps)^i c_{(1+\eps)^i}^* \leq (1+\eps) \cdot \cost(P). \label{eq:cs-est}
\end{align}
\end{fact}

\paragraph{Streaming Model and Linear Sketching.} 
In the geometric streaming model, points are inserted and deleted from discrete $d$-dimensional space $[\cord]^d$. Any such point set $P \subset [\cord]^d$ can be represented by a vector $x \in \R^{\cord^d}$, and so a geometric stream can be represented by a sequence of coordinate-wise updates to $x$. A \textit{linear sketching} algorithm is an (often randomized) linear function $F:\R^{\cord^d} \to \R^s$, such that given $F(x)$ (and a representation of the coefficents of the function $F$), one can obtain an approximation of some function of $x$. For any $\alpha \geq 2$, a $\alpha$ pass linear sketch is an algorithm which, begins with a linear sketch $F_1:\R^{\cord^d} \to \R^{s_1}$. Then, for every $i=2,\dots,\alpha$, given $(F_1,\dots,F_{i-1},F_1(x),\dots,F_{i-1}(x))$, adaptively generates a new linear function $F_i: \R^{\cord^d} \to \R^{s_i}$, and outputs its approximation based on $(F_1,\dots,F_{\alpha},F_1(x),\dots,F_{\alpha}(x))$. The space of a linear sketching algorithm is the space required to store each of the $F_i,F_i(x)$'s (we remark that all our linear sketching algorithms will be able to compactly store the linear functions $F_i$).  It is easy to see then that a $\alpha$-round linear sketching algorithm yields a $\alpha$-pass streaming algorithm in the same space. 



\subsection{Random Quadtree Decompositions}
We describe the construction of a random hierarchical 
decomposition of our input, represented by a tree whose
nodes corresponds to regions of $\R^d$.
Consider an arbitrary $d$-dimensional axis-aligned hypercube of side length $2\Delta$ containing the entire input and
such that each point is at distance at least $\Delta/2$ from the boundary. Associate this hypercube to the root of the tree. Next, partition the cube into $2^d$ child axis-aligned hypercubes of side lenght $\Delta$. Each such child hypercubes corresponds to a child node of the root node of the tree.
Then, apply this
procedure recursively: Hypercubes of side length $\Delta/2^i$ correspond to nodes in the tree with hop distance $i+1$ from the root. The partition stops
when the side length of the hypercube is 1 and so each cube contains a single input point and the depth is $O(\log \Delta)$.  
We now explain how to 
obtain a randomized decomposition: 
Set $\beta = c_0 (d \log \Delta)$ for a sufficiently small constant $c_0$. Let $L_t$ consists of the set of centers of hypercubes in $\R^d$ with side length $\frac{t}{\beta d}$, where the origin of all the hypercubes are shifted by a random vector $\vec{v} \in [\Delta/2]^d$. 
Then, define the function $f_t:\R^d \to L_t$ via $f(x) = \arg \min_{y \in L_t} \|x-y\|_1$. Notice that $\|x - f_t(x)\|_1 \leq t/\beta$ for all $x \in \R^d$.
We slightly abuse the notation and let $f_t(P)$ for any 
set $P$ of points in $\R^d$ be $\bigcup_{x \in p} \{f_t(x)\}$.

\subsubsection{Bounding the Aspect Ratio.}
We now demonstrate that, via an application of the quadtree discretization, we can assume that the aspect ratio $\Delta$ of the point set $P$ satisfies $\Delta = O(n d /\epsilon)$.

\begin{proposition}\label{prop:minDistance}
Let $P \subset \R^d$ denote an arbitrary set of $n$ points, and let $\Delta = \max_{x, y \in P} \|x - y\|_1$. Then, for $t = \eps \beta \Delta / n$, the set $P' = f_t(P) \subset \R^d$ satisfies:
\begin{itemize}
    \item The maximum distance between any two $u, v \in P'$ is at most $(1 + 2\eps / n) \Delta$, and the minimum non-zero distance between any two $u, v \in P'$ is at least $\eps \Delta / (nd)$.
    \item We have $(1-3\eps) \cost(P) \leq \cost(P') \leq (1+3\eps) \cost(P)$.
\end{itemize}
\end{proposition}

\begin{proof}
Recall that, for any two $x, y \in P$, we have $\|x - y\|_1 \leq \| f_t(x) - f_t(y)\|_1 + 2t / (\beta d)$, and $\|f_t(x) - f_t(y)\|_1 \leq \|x-y\|_1 + 2t / (\beta d)$. Therefore, for any two points $u, v \in P'$, if $x,y \in P$ are any two points where $f_t(x) = u$ and $f_t(y) = v$, then $\| u - v \|_1 \leq \|x - y\|_1 + 2t/(\beta d)$. In particular, the maximum distance between any two $u, v \in P'$ will be at most $\Delta + 2t / (\beta d) \leq (1+2\eps / (nd)) \cdot \Delta$. On the other hand, by definition of $f_t$, the minimum non-zero distance between any two $u, v \in P'$ is at least $t / (\beta d) = \eps \Delta / (n d)$. This completes the first item.

For the second item, if $T$ is a set of edges $\{ (x, y) \}$ of the minimum spanning tree of $P$, the set $T' = \{ (f_t(x), f_t(y)) : (x, y) \in T\}$ contains at most $n-1$ edges and also a spans $P'$. The cost of $T'$ is at most $\sum_{(x,y) \in T} \|x -y\|_1 + (n-1) \cdot 2t /\beta$, so $\cost(P') \leq \cost(P) + (n-1) \cdot 2t / \beta$. On the other hand, let $T'$ be a set of edges $\{ (u, v) \}$ which is a minimum spanning tree of $P'$. For each $u \in P'$, let $x_u \in P$ denote an arbitrary point $x \in P$, and consider the spanning set of edges $T$ given by the union of two classes of edges: (i) the collection $\{ (x_u, x_v) \in P\times P : (u, v) \in P'\}$, and, (ii) the set $\{ (x_u, y) \in P \times P : f_t(y) = u, u \in P' \}$. We note that there are at most $n-1$ edges in (i), so the total cost of edges in (i) is at most $\cost(P') + (n-1) \cdot 2t/\beta$. There are at most $n$ edges in (ii), and their total cost is $n \cdot t/\beta$. Together, this means $\cost(P) \leq \cost(P') + 3n t / \beta$. The upper and lower bounds together imply
\begin{align*}
    | \cost(P) - \cost(P') | &\leq 3n \cdot \frac{t}{\beta} \leq \frac{3\eps n \Delta}{n} \leq 3\eps \cdot \cost(P), 
\end{align*}
where in the final inequality, we used the fact that $\cost(P) \geq \Delta$ (since any spanning tree must connect the points at largest distance).
\end{proof}

We note that the above lemma implies that, up to an oblivious transformation of the points (i.e., letting $P' = f_t(P)$ for $t = \eps \beta \Delta / n$), the ratio of the maximum distance and the minimum non-zero distance is at most $O(nd / \eps)$ without significantly affecting the cost of the minimum spanning tree. Therefore, an algorithm may obliviously transform the input $P$ to $P'$ by applying $f_{t}$ and re-scaling in order to utilize the estimator of (\ref{eq:cs-est}). Specifically, an algorithm will (i) estimate the number of distinct points in $P'$ (since this becomes the new $n$ in Fact~\ref{fact:cs}) which will be simple using an $\ell_0$-sketch, and (ii) estimate $c_{(1+\eps)^i}^*$ for each $i$ between $0$ and $\log_{1+\eps} \Delta - 1$. The benefit to doing this is that we it now suffices to consider $\Delta \leq O(nd/\eps)$.

%% file: structural_lemma.tex
\def\calL{\mathcal{L}}
\def\calQ{\mathcal{Q}}
\def\calT{\mathcal{T}}
\def\BFS{\text{BFS}}

Let $P \subset \R^d$ be a set of points with distances between $1$ and $\Delta$, and we assume $\Delta$ to be a power of $2$ without loss of generality. In this section, we will work with the $\ell_1$ metric, which, via metric embeddings, generalizes to all $p \in [1,2]$. 
Let $\eps\in (0,1)$ be a parameter such that $\Delta^\eps \ge \log^2\Delta$.
Let $\alpha$ be a positive integer satisfying $\alpha\le \Delta^\eps$.
The goal of this section is to give a sketching-friendly estimator that gives a $(1+ O(\log(\alpha+1)\cdot(\eps\alpha)^{-1}))$-approximation of $\cost(P)$.
The reason of having two parameters $\eps$ and $\alpha$ in the estimator is that $\eps$ will corresponds to the space complexity (as in roughly $\Delta^\eps$)
and $\alpha$ will correspond to the number of passes on the data later in our algorithmic implementation of the estimator.

\subsubsection{Setup: Levels, Blocks, and the Discretized Graph}
Since we work with the $\ell_1$ metric in this section, for the remainder of the section we will write $\cB(p,r) = \{y \in \R^d \; | \; ||p-y||_1 \leq r \}$ to denote the $\ell_1$-ball of radius $r$ centered at $p$. 

We now let $\delta>0$ be a parameter defined using $\eps$ and $\alpha$ as follows:
$$
\delta=\min\left(\frac{1}{100},\frac{1}{\eps\alpha}\right)\le \frac{1}{100}.
$$
We will refer to powers of $(1+\delta)$ as \emph{levels}: $(1+\delta)^0,(1+\delta)^1, \ldots,(1+\delta)^L$, 
 where $L$ the smallest integer such that $(1+\delta)^{L-1 }\ge \Delta$.
We write $\calL$ to denote the set of levels (so $L=\Theta((\log\Delta)/\delta)$ and we usually use $t$ to denote a level.
We further divide $\calL$ into \emph{blocks} of levels: $\calQ_1, \calQ_2,\ldots$,
where $\calQ_i$ contains all levels in $\calL$ between $\Delta^{(i-1)\eps}$ and $\Delta^{i\eps}$.
So the total number of blocks is 
  $\Theta(1/\eps)$.
  
 We also write $\calT$ to denote the set of $ \Delta^{(i-1)\eps}$ for each $\calQ_i$. Given a level $t\in \calQ_i$,
  we let $V_t=V_T$ with $T=\Delta^{(i-1)\eps}$. In other words, for all $t ,t'\in \cQ_i$, the vertex sets $V_{t}= V_{t'}$ are the same. 
  and let $G_t=(V_t,E_t)$ denote the undirected (unweighted) graph such that $(p,q)\in E_t$ for $p,q\in V_t$ iff $\|p-q\|\le t$.

For each block $\calQ_i$, let $T=\Delta^{(i-1)\eps}$.
We recall $f_T$ using $\beta:= 1000\alpha d/\eps $, and cubes with side length $T/(d\beta)$ and define $V_T=f_T(P)$. 
We have the following from quadtree analysis:

\begin{lemma}\label{lem:quadtreeLemma}
With probability at least $9/10$ over the random shifting, we have 
$$
\sum_{T\in \calT} \frac{T}{\beta}\cdot (|V_T|-1)\le O\left(\frac{d}{\eps}\right)\cdot \cost(P).
$$
\end{lemma}

\begin{proof}
First, note the following simple fact about partitions and spanning trees: if $M$ is a set of edges on the vertex set $[n]$ which forms a spanning tree, then for any partition of $[n]$ into $k$ parts $P_1,\dots, P_k$, there is a subset of $M$ consisting of at least $k-1$ distinct edges $\{ i, j \}$ such $\{ i, j \}$ is ``cut'' by the partition, i.e., for some $\ell \in [k]$, $|\{ i, j \} \cap P_{\ell}| = 1$. Otherwise, removing less than $k-1$ edges of $M$, we obtain $m < k$ connected components; thus, there must be some component which contains at least two points $i, j$ from distinct parts, and by considering the path between $i$ and $j$ in $M$, we find another edge ``cut'' by $P_1,\dots, P_k$. 

The connection to Lemma~\ref{lem:quadtreeLemma} follows by fixing $T \in \calT$, and considering the partition of $P$ where $x, y \in P$ are in the same part iff $f_{T}(x) = f_{T}(y)$. If $M$ is the minimum spanning tree of $P$, the partition into $|V_T|$ parts cuts at least $|V_T| - 1$ edges. Thus, it remains to show that the expected number of edges of $M$ cut is upper bounded. We divide the edges of $M$ into two groups. The first group, $M_1$, consists of edges $(p, q)$ where there exists some $i \in [d]$ where $|p_i - q_i| \geq 2T/(\beta d)$. Notice that edges $(p, q)$ in $M_1$ are ``cut'' with probability $1$, i.e., $f_{T}(p) \neq f_{T}(q)$. Notice that the number of edges in $M_1$ is at most $O(\cost(P) \cdot (d\beta) / T)$, since each such edge contributes at least $\Omega(d\beta / T)$ to $\cost(T)$. The remaining edges $M_2 = M \setminus M_1$ are ``cut'' with probability at most $\| x - y\|_1 / (2T / (d\beta))$. Note that the expected number of such edges cut is at most $\cost(T) \cdot (d \beta / T)$. In particular, adding over all $T \in \calT$, we have
\[ \Ex\left[ \sum_{T \in \calT} \frac{T}{\beta}(|V_T| -1) \right] \leq |\calT| \cdot O(d) \cdot \cost(P). \]
The lemma then follows from the fact $|\calT|\leq O(1/\eps)$, as well as Markov's inequality.
\end{proof}


\subsection{The Estimator}

\def\CC{\text{CC}}

For each level $t\in \calL$ and $p\in V_t$, we let 
\begin{equation}\label{eqn:xt}
    x_t(p)=\frac{1}{|\CC(p,G_t)|},
\end{equation}

  where $\CC(p,G_t)$ denotes the connected component of $G_t$ that contains $p$ and
  $|\CC(p,G_t)|$ is its size.
Inspired by the work of \cite{czumaj2009estimating}, it is useful to consider an idealized estimator for $\cost(P)$ given by 
$$
n-(1+\delta)^{L+1}+\delta \sum_{t\in \calL} t
  \sum_{p\in V_t} x_t(p).
$$
The following lemma shows that the ideal estimator gives a good approximation of $\cost(P)$.

\begin{lemma}\label{prop:xBound}
We have
\[ 
(1-4\delta )\cdot \cost(P) \leq n-(1+\delta)^{L+1}+\delta \sum_{t\in \calL} 
  \sum_{p\in V_t} x_t(p) \leq \left(1+ \delta \right)\cdot \cost(P)\]
\end{lemma}
\begin{proof} 
Let $c_i\ge 1$ be the number of connected components in $G_t$ with $t=(1+\delta)^i$. Then we first note that the expression can be rewritten as (noting that $c_L=1$)
$$
n-(1+\delta)^{L+1}+\delta\sum_{i=0}^L (1+\delta)^i c_i
=(n-1)+\delta\sum_{i=0}^{L-1} (1+\delta)^i\cdot (c_{i}-1).
$$

For each level $t\in \calL$, we let $G_t^*$ denote the following graph over $P$: $p,q\in P$ have an edge iff $\|p-q\|_1\le (1-\delta/2)t$.
For each $i\in [0:L]$, let $c_i^*$ denote the number of connected components in $G_t^*$ with $t=(1+\delta)^i$.
Then we have the following claim:

\begin{claim}\label{hahaclaim}
We have
$$
(1-\delta)\cdot \cost(P)\le (n-1)+\delta\sum_{i=0}^{L-1} (1+\delta)^i\cdot (c^*_{i}-1)\le (1+ \delta)\cdot \cost(P).
$$
\end{claim}
\begin{proof}
The proof follows similar arguments used in the proof of Lemma 2.1 of \cite{czumaj2009estimating}.
Fix an MST of $P$.
For each $i\in [0:L ]$, we write $m^*_i$ to denote the number of edges in the MST of distance in $((1-\delta/2)(1+\delta)^{i },(1-\delta/2) (1+\delta)^{i+1}]$. 
Given that distances are between $1$ and $\Delta\le (1+\delta)^{L-1}$, 
$$
 \frac{\cost(P)}{1+\delta } \le \sum_{i=0}^{L } (1+\delta)^i\cdot m_i^*\le (1+\delta )\cdot \cost(P).
$$
On the other hand, given that $m_0^*+\cdots+m_L^*=n-1$, we have 
\begin{align*}
\sum_{i=0}^L (1+\delta)^i\cdot m_i^*
&=(n-1)+\sum_{i=1}^L ((1+\delta)^i-1)\cdot m_i^*\\
&=(n-1)+\delta \sum_{i\in [L]} \sum_{j=0}^{i-1} (1+\delta)^j\cdot m_i^*\\
&=(n-1)+\delta \sum_{j=0}^{L-1} (1+\delta)^j \sum_{i=j+1}^{L} m_i^*.
\end{align*}
For each $j\in [0:L-1]$, we have 
$c_j^*=m_{j+1}^*+\cdots+m_L^*+1$. 
As a result, we have 
$$
\sum_{i=0}^L (1+\delta)^i\cdot m_i^*=(n-1)+\delta
\sum_{j=0}^{L-1} (1+\delta)^j\cdot (c^*_{j}-1)
$$
The claim then follows using $\delta\le 0.01$.
\end{proof}

The next claim gives us a connection between $(c_i^*)$ and $(c_i)$:  
\begin{claim}\label{claim:haha2}
For each $i\ge 0$ we have $c_{i+1}^*\le c_i\le c_{i}^*$. 
\end{claim}
\begin{proof}
We first prove $c_i\ge c_{i+1}^*$. Let $t=(1+\delta)^i$
  and $t'=(1+\delta)t$.
So $c_i$ is the number of connected components of $G_t$ and $c_{i+1}^*$ is the number of connected components of $G_{t'}^*$. Let $T$ be the power of $\Delta^\eps$ such that $V_t=V_T$, with $t\ge T$. 
We have for every point $p\in P$, $\|p-f_T(p)\|_1\le T/\beta\le t/\beta$.
Let $C_1,\ldots,C_k$ be the connected components of $G_{t'} ^*$ (so they form a partition of $P$).
We show that $f_T(C_1),\ldots,f_T(C_k)$ must form a partition of $f_T(P)=V_T=V_t$ and there are no edges between them in $G_t$.
To see this, note than any two points $p,q\in P$ in two different components of $G_{t'}^*$ satisfy $\|p-q\|_1>(1-\delta/2)(1+\delta)t\ge (1+\delta/3)t$. As a result, we have
$$
\|f_T(p),f_T(q)\|_1\ge (1+\delta/3-2/\beta)t>t,
$$
where we used $\beta\delta\ge 10$ from our choices of $\beta$ and $\delta$.

The proof of $c_i\le c_i^*$ is similar. 
Let $C_1,\ldots,C_k$ be the connected components of $G_t$.
We show that $f^{-1}(C_1),\ldots,f^{-1}(C_k)$ form a partition of $P$ and there is no edge between them in $G_t^*$.
To see this, we take $p\in f^{-1}(C_{\ell})$ and $q\in f^{-1}(C_{\ell'})$ for some $\ell\ne \ell'$. Then  
$
\|f_T(p)-f_T(q)\|_1>t
$ and thus,
$$
\|p-q\|_1\ge (1-2/\beta)t >(1-\delta/2)t.
$$
So there is no edge between $p$ and $q$ in $G_{t}^*$.
This finishes the proof of the claim.
\end{proof}
  
Using the above claim, we have 
$$
(n-1)+\delta\sum_{i=0}^{L-1} (1+\delta)^i\cdot (c_{i}-1)\le 
(n-1)+\delta\sum_{i=0}^{L-1} (1+\delta)^i\cdot (c^*_{i}-1)\le (1+ \delta)\cdot \cost(P).
$$
On the other hand, we have (noting that $c_0^*=n$ and $c^*_L=1$)
\begin{align*}
 \sum_{i=0}^{L-1}(1+\delta)^i\cdot (c_i-1)
 &\ge \sum_{i=1}^{L} (1+\delta)^{i-1} \cdot (c^*_i-1)\\&=\sum_{i=0}^{L-1} (1+\delta)^{i-1} \cdot (c^*_i-1)- \frac{n-1}{1+\delta}
 \\&\ge (1-\delta)\sum_{i=0}^{L-1}(1+\delta)^{i-1}\cdot (c_i^*-1)- ({n-1})
\end{align*}
This implies that 
$$
(n-1)+\delta\sum_{i=0}^{L-1} (1+\delta)^i\cdot (c_{i}-1)
\ge (1-\delta)\cdot\cost(P)-\delta^2 \sum_{i=0}^{L-1}(1+\delta)^{i-1}\cdot (c_i^*-1)-\delta(n-1).
$$
But $\cost(P)\ge n-1$ trivially and we also have 
$$
\delta\sum_{i=0}^{L-1} (1+\delta)^i\cdot (c^*_{i}-1)\le (1+\delta)\cdot\cost(P).
$$
Finally we have 
\begin{align*}
(n-1)+\delta\sum_{i=0}^{L-1} (1+\delta)^i\cdot (c_{i}-1)
&\ge (1-\delta)\cdot\cost(P)-\delta(1+\delta)\cdot \cost(P)-\delta\cdot \cost(P)\\[-1.5ex]&\ge (1-4\delta)\cdot\cost(P).
\end{align*}
Rewriting the LHS finishes the proof of the lemma.
\end{proof}

Sketching the ideal estimator, however, turns out to be challenging. 
Instead, we will introduce an alternative estimator and show in the rest of this section that it leads to good approximation of $\cost(P)$ as well.
We start with the definition of $y_t(p)$ and $z_t(p)$, which we use to replace $x_t(p)$ in the ideal estimator.
Given $t\in \calL$ and $p\in V_t$,  $y_t(p)\ge 0$ is defined as follows: 
\begin{enumerate}
\item Let $n_j = |\cB(p,2^ j  t, V_t) |$ for each $j=0,1, \dots,  \log L$;
\item If $n_0\ge \beta^2 \Delta^{10\eps}$, set $y_t(p)=0$;
\item Otherwise, letting $j^*$ be the largest $j\in [0:\log L]$ with $n_j< \beta^2 \Delta^{10 \eps}$, we set
$$y_t(p) = \frac{1}{|\CC(p, G_t(\cB(p,2^{j^*}  t,V_t))) |},$$
where we write $G_t(S)$ for some $S\subseteq V_t$
  to denote the subgraph of $G_t$ induced on $S$.
\end{enumerate}
Let $\BFS(p,G_t,\alpha)$ denote the set of vertices of $G_t$
  explored by running BFS on $p$ for $\alpha$ rounds (i.e., it contains all vertices in $G_t$ with shortest path distance at most $\alpha$ from $p$.
Next we define $z_t(p)$: 
\begin{enumerate}
\item If $\BFS(p,G_t,\alpha)\ge \beta^2 \Delta^{10\eps}$, set $z_t(p)=0$;
\item Otherwise, set
$$
z_t(p)=\frac{1}{| \BFS(p,G_t,\alpha) |}.
$$
\end{enumerate}
Our alternative estimator is
\begin{equation}  \label{eqn:mainEstimatorZ}
Z:=n-(1+\delta)^{L+1}+\delta\sum_{t\in \calL} t\sum_{p\in V_t}\min\big(y_t(p),z_t(p)\big).
\end{equation}

For each $p\in V_t$, we say $p$ is \emph{dead}
  if either $y_t(p)=0$ on line 2
  or $z_t(p)=0$ on line 1, and we say $p$ is \emph{alive} otherwise.
It is clear that $y_t(p)\le x_t(p)$ when $p$ is dead and $y_t(p)\ge x_t(p)$ when $p$ is alive. 
Therefore, to show that $Z$ is a good approximation of $\cost(P)$, it suffices to upperbound
$$
\sum_{t\in \calL} t \sum_{\text{alive}\ p\in V_t}\Big(\min\big(y_t(p),z_t(p)\big)- x_t(p)\Big)
\quad\text{and}\quad
\sum_{t\in \calL} t \sum_{\text{dead}\ p\in V_t}x_t(p).
$$
These are our goals in Section \ref{sec:upper} and \ref{sec:lower}, respectively, where we prove the following two lemmas:

\begin{lemma}\label{mainanalysislemma1}
Suppose the conclusion of Lemma \ref{lem:quadtreeLemma} holds. We have
\begin{equation}\label{mainlowerbound}
\sum_{t\in \calL} t \sum_{\text{alive}\ p\in V_t}\Big(\min\big(y_t(p),z_t(p)\big)-x_t(p)\Big)\le O\left(\frac{\log \alpha +1}{\delta\eps\alpha}\right)\cdot \cost(P).
\end{equation}
\end{lemma}
\begin{lemma}\label{mainanalysislemma2}
Suppose the conclusion of Lemma \ref{lem:quadtreeLemma} holds. We have 
$$\sum_{t\in\calL} t \sum_{\text{dead}\ p\in V_t}x_t(p)\le O\left(\frac{L}{\Delta^{9\eps}}\right)\cdot \cost(P).
$$
\end{lemma}
Our main theorem in this section follows from these two lemmas:

\begin{theorem}\label{thm:mainEstimator}
Suppose the conclusion of Lemma \ref{lem:quadtreeLemma} holds. Then for any $\eps \in (0,1)$, and positive integer $\alpha \leq \Delta^\eps$, we have  
$$
\cost(P)\le 
\left(1+\frac{1}{\Delta^{8\eps}}\right)\cdot Z\le \left(1+O\left(\frac{\log\alpha +1}{\eps\alpha}\right)
\right)\cdot\cost(P).
$$
Where $Z$ is as defined in (\ref{eqn:mainEstimatorZ}) with parameters $\eps,\alpha$.
\end{theorem}
\begin{proof}
The first inequality follows from
$$
\delta\cdot \frac{L}{\Delta^{9\eps}}
=O\left(\frac{\log \Delta}{ \Delta^{9\eps}}\right) =o\left(\frac{1}{ \Delta^{8\eps}}\right) ,
$$
using the assumptions that  $\Delta^\eps\ge \log ^2\Delta$. The second inequality used  $\alpha\le \Delta^\eps$ and $1/\eps\le \log\Delta$.
\end{proof}

Note that when $\alpha$ is set to $1$, we always have $\min(y_t(p),z_t(p))=y_t(p)$ for all $p\in V_t$.
This leads to the following corollary:

\begin{corollary}\label{foronepass}
Suppose the conclusion of Lemma \ref{lem:quadtreeLemma} holds. Then we have 
$$\Omega(1)\cdot \cost(P)\le n-(1+\delta)^{L+1}+\delta\sum_{t\in \calL} t\sum_{p\in V_t} y_t(p) \le O\left(\frac{1}{\eps}\right)\cdot \cost(P).
$$
\end{corollary}
  
\subsection{A Structural Lemma}\label{sec:structure}
\def\calP{\mathcal{P}}
We prove a structural lemma that will play a crucial role in the proof of Lemma \ref{mainanalysislemma1}. 
We start with the definition of bad points: 


\begin{definition}\label{def:bad}
Fix a level $t\in \calL$.
We say a point $p\in V_t$ is \emph{bad} at level $t$ if 
$$
|\mathcal{B}(p, t,V_t)| \geq\beta^2  \Delta^{10 \epsilon}
$$
We say $p$ is \emph{$j$-bad at level $t$} for some nonnegative integer $j$ if both of the following hold: 
$$
|\mathcal{B}(p,  2^{j+1}  t,V_t)| \geq\beta^2  \Delta^{10 \epsilon}\quad\text{and}\quad
|\mathcal{B}(p,  2^{j}  t,V_t)| < \beta^2  \Delta^{10 \epsilon}. $$
\end{definition}

We prove the following key structural lemma:

\begin{lemma}\label{lem:structural}
Fix a $j\ge 0$, a block $\calQ$   
 and any subset $S$ of $\calQ $ with the property that for every $t,t' \in S$ with $t>t'$ we have $t\geq 4t'$. 
Let $(N_t\subseteq V_T: t\in S)$ be a tuple of sets such that for every $t\in S$: (1) Every $p\in N_t$ is $j$-bad at level $t$; and (2) Every $p,q\in N_t$ satisfy $\|p-q\|_1\ge  t/2$.
Then we have
\[\sum_{t \in S} \left|N_t\right| \cdot  t  \leq O\left(\cost (P)\right).\]
\end{lemma}
We prove two simple claims before proving Lemma \ref{lem:structural}:

\begin{claim}\label{claim:1}
Let $x_1,\dots,x_k \in \R^d$ and $r_1 , \dots ,r_k>0$ be such that $$\cB(x_i , r_i ) \cap \cB(x_j, r_j) = \emptyset,\quad \text{for all $i\ne j \in [k]$.}$$ Then we have 
\[\cost(\{x_1,\dots,x_k\}) \geq \sum_{i\in [k]} r_i \]
\end{claim}
\begin{proof}
Consider the optimal MST over the points $x_1,\dots,x_k$. Since the ball of radius $r_i$ around $x_i$ contains no other points $x_j$ for $j \neq i$, it follows that some edge passes through the boundary of the ball $\cB(x_i,r_i)$ to connect $x_i$ to points outside of  $\cB(x_i,r_i)$. Thus each point $x_i$ pays at least $r_i$ on an edge to be connected to the rest of the MST, which completes the proof. 
\end{proof}

\begin{claim}\label{prop:verySmart}
Fix any block $\calQ$ and any $t,t' \in \calQ$ with  $t \geq  4t'$. 
If $x\in V_t$ is $j$-bad at level $t$ and $y\in V_{t'}$ 
  is $j$-bad at level $t'$, then we must have $\|x-y\|_1\ge t/2$.
\end{claim}
\begin{proof}
 First we note that since $t,t'\in \calQ$, we have $V_t=V_{t'}=V_T$.
Assume for a contradiction that $\|x-y\|_1< t$. Since $x$ is bad at level $t$, we know that 
\[|\cB(x, 2^j t,V_t))| < \beta^2 \Delta^{10 \eps} .\]
On the other hand, using $4 t'\leq t$ and $\|x-y\|_1< t/2$, 
we have 
\[\cB(y, 2^{j+1} t',V_{t'})=\cB(y, 2^{j+1} t',V_{t }) \subseteq \cB(y, 2^{j-1}t ,V_t) \subseteq \cB(x,  2^j t,V_t)\]
but since $y$ is $j$-bad at level $t'$, we have $|\cB(y, 2^{j+1} t',V_{t'})| \geq \beta^2 \Delta^{10 \eps}$, a contradiction.
\end{proof}

We are now ready to prove Lemma \ref{lem:structural}.

\begin{proof}[Proof of Lemma \ref{lem:structural}]
First note that
  for any two distinct levels $t, t' \in S$ with $t>t'$, we have  $t\ge 4t'$. It follows from Claim \ref{prop:verySmart} that $\|x-y\|_1\ge  t/2$ for every $x\in N_t$ and $y\in N_{t'}$ and thus,
$\cB(x, t/5) \cap \cB(y , t'/5 ) = \emptyset$. Next, for $x,y\in N_t$ for the same $t$, we have from the assumption that 
$\|x-y\|_1\ge t/2$ and thus, $\cB(x, t/5) \cap \cB(y ,  t/5)=\emptyset$.
It follows from Claim \ref{claim:1} that 
\[\sum_{t \in S} \left|N_t\right| \cdot  t  \leq O\big( \cost (\cup_{t\in S} N_t )\big)\]
So it suffices to show that $\cost(\cup_{t\in \cS} N_t )\le O(\cost(P))$.
To see this is the case, note that (1) 
we have shown above that every two points $x,y\in \cup_{t\in S}N_t$ has $\|x-y\|_1\ge  t/2$;
and (2) For every point $p_i\in \cup_{t\in S} N_t\subseteq V_T$, there is a point $q_i\in P$ such that $\|p_i-q_i\|_1\le T/\beta\le t/\beta$. As a result, we have $\|p_i-p_j\|_1=(1\pm 3/\beta)\cdot \|q_i-q_j\|_1$ for any $p_i,p_j\in \cup_{t\in S}N_t$, and recall that $\beta\ge 1000$.
This finishes the proof of the lemma.
\end{proof}

%% file: Algorithm.tex
\def\calI{\mathcal{I}}

\subsection{Upper Bounding the Estimator}\label{sec:upper}

We prove Lemma \ref{mainanalysislemma1} in this subsection.

We write $B_t\subseteq V_t$ to denote the set of bad points at level $t$, and $B_t^j$ to denote the set of $j$-bad points at level $t$, for each $j\ge 0$.
Then $B_t,B_{t,0},B_{t,1},\ldots,$ form a partition of $V_t$. 
We also write $B_{t,+}$ to denote the union of $B_{t,j}$ over all $j\ge \log L$.

Furthermore we classify $j$-bad points in $B_{t,j}$, $j\in 0$, into those that are \emph{complete} and \emph{incomplete}:
\begin{enumerate}
\item We say $p\in B_{t,j}$, a $j$-bad point at level $t$, is complete if at least one of the following holds:
$$
\CC(p, G_t(\cB(p,2^{j } t,V_t))) = \CC(p, G_t )\quad\text{or}\quad \BFS(p,G_t,\alpha)=\CC(p,G_t).
$$
\item Otherwise we say the point is incomplete. From the definition we have
$$
\CC(p, G_t(\cB(p,2^{j } t,V_t))) \subsetneq \CC(p, G_t )\quad\text{and}\quad \BFS(p,G_t,\alpha)\subsetneq \CC(p,G_t).
$$
\end{enumerate}
For each $B_{t,j}$ we write $B_{t,j}^*$ to denote the set of incomplete points, and $B_{t,+}^*$ to denote the union of $B_{t,j}^*$ with $j\ge \log L$.
With these definitions, to upper bound (\ref{mainlowerbound}), it suffices to bound
\begin{equation}\label{sumsum}
\sum_{t\in\calL} t
\sum_{j=0}^{\log L-1} \sum_{p\in B_{t,j}^*} y_t(p)+\sum_{t\in \calL}t \sum_{p\in B_{t,+}^*} \Big(y_t(p)-x_t(p)\Big).
\end{equation}

The following graph-theoretic lemma will be heavily used in our analysis:
 
\begin{lemma}\label{lem:graphDecomp}
Let $\lambda \geq 2$ 
 be an integer and  $G = (V,E)$ be a connected unweighted undirected graph with $\diam(G) \geq \lambda$. Then there exists an integer $k \geq 1$ such that the following holds:
\begin{flushleft}\begin{itemize}
    \item There is a partition $C_1, \dots,C_k$ of $V$ such that the subgraph of $G$ induced by each $C_i$ is  connected and has $\diam(C_i) \leq  \lambda$ for all $i \in [k]$.
    \item  There is an independent set $\cI \subseteq V$ such that  $|\cI \cap C_i| \geq \max\{1,\lfloor \lambda/12 \rfloor\}$ for each $i \in [k]$\\ \emph{(}and thus,  $|\cI|\ge k\cdot \max\{1,\lfloor\lambda/12\rfloor \}$\emph{)}.
\end{itemize}
\end{flushleft}
\end{lemma}
\begin{proof}
We greedily construct the partition as follows. First, suppose $\lambda \geq 6$. Specifically, we begin by setting $S = V$, and $i=1$, and then repeat the following procedure until $S = \emptyset$:
\begin{enumerate}
    \item Select $x \in S$ arbitrarily.
    \item Run a breadth-first search for $\lfloor \lambda/6 \rfloor $ steps starting at $x$ in the subgraph of $C$ induced by $S$. Set $C_i$ to be the set of points seen on the BFS.
    \item Set $S \leftarrow S \setminus C_i$, $i \leftarrow i+1$ and proceed. 
\end{enumerate}
Let $C_1,\dots,C_k$ be the partition of the vertices of $G$ formed from the above procedure, and $x_i \in C_k$ be the point chosen on the $i$-th round above. Since each point in $C_k$ can be reached from $x_i$ in at most $\lambda/6$ steps, the diameter of each $C_i$ is at most $\lambda/3$. Call a cluster $C_i$ \textit{complete} if the BFS ran for exactly $\lfloor \lambda/6 \rfloor$ steps, adding at least one new vertex to $C_i$ at each step. Call $C_i$ \textit{incomplete} otherwise, meaning that the BFS explored all vertices that were connected to $x_i$ in the subgraph of $C$ induced by $S$. Let $k_0 \leq k$ be the number of complete clusters, and order the clusters so that $C_1,\dots,C_{k_0}$ are all big.

For each complete cluster $i$, there necessarily exists a path $x_i=p_{i,0},p_{i,1},p_{i,2},\dots,p_{i,\lfloor \lambda/6 \rfloor}$ such that $p_{i,j}$ was added to $C_i$ on the $j$-th step of the BFS. Define the set $\cI_i = \{p_{i,j} \; | \; j < \lfloor \lambda/ 12 \rfloor, \; j \pmod 2 =1 \}$, and define $\cI = \cup_{i =1}^{k_0} \cI_i$.   We first claim the following:

\begin{claim}
The set $\cI$ is an independent set in $G$.
\end{claim}
\begin{proof}
We first prove that each individual set $\cI_i$ is an independent set, and then show that no edges cross from $\cI_i$ to $\cI_j$ for $i,j \in [k_0]$. For the first fact, fix any $p_{i,j_1},p_{i,j_{2}} \in \cI_i$, where WLOG $j_1 < j_2$. Then if there were an edge $(p_{i,j_1},p_{i,j_{2}})$ in the graph $C$, then since $p_{i,j_1}$ was explored on the $j_1$-th step of the BFS, then on the $j_1+1$'st step the vertex $p_{i,j_2}$ would have been added to $C_i$. But since $p_{i,j_2}$ was first added to $C_i$ on the $j_2$-th step, and $j_2 > j_1 + 1$, it follows that such an edge could not have existed. 

Finally, fix $i_1,i_2 \in [k_0]$, and fix any $p_{i_1,j_1} \in \cI_{i_1}$ and $p_{i_2,j_2} \in \cI_{j_2}$. WLOG, the cluster $C_{i_1}$ was constructed before the cluster $C_{j_2}$ in the greedy procedure above. Now since $j_1 < \lfloor \lambda/6 \rfloor$ by construction of the sets $\cI_{i}$, and since $i_1$ is complete, if $(p_{i_1,j_1}, p_{i_2,j_2})$ was indeed an edge of $G$, it would be the case that on the $(j_1 + 1)$-st step of the BFS procedure which constructed $C_{i_1}$, the vertex $p_{i_2,j_2}$ would have been added to $C_{i_1}$. Since $p_{i_2,j_2} \notin C_{i_1}$, it follows that such an edge could not have existed, which completes the proof of the claim.
\end{proof}

We next claim that for any incomplete cluster $C_i$, there is a complete cluster $C_{i'}$ such that an edge exists between a vertex in $C_i$ and a vertex in $C_{i'}$. To see this, let $C_i$ be an incomplete cluster that halted its BFS on step $\tau < \lfloor \lambda/6 \rfloor$. Then it must be the case that there was some vertex $u \in V \setminus C_i$ and a vertex $v \in C_i$ such that $(u,v) \in E$. If this were not the case, then since $G$ and $C_i$ are bot connected, it would follow that $C_i = G$, which cannot hold since $\diam(G) \geq \lambda$, and we know that $\diam(C_i) < \lambda/2$, since every vertex in $C_i$ can be reached from $x_i$ in at most $\tau < \lambda/6$ steps. Now if the BFS ended at step $\tau$, the vertex $v$ must have already been added to a cluster $C_{i'}$ that was constructed before $C_i$. Moreover, $C_{i'}$ could not have been incomplete, since if $v$ was added to $C_{i'}$ on step $\rho <\lfloor \lambda/6 \rfloor$ step of the BFS constructing $C_{i'}$, then $u$ would have been added to $C_{i'}$ on step $\rho+1$, which did not occur, which completes this claim.

Given the above, we can merge each incomplete cluster $C_i$ into to a complete cluster $C_{i'}$ such that there is an edge between $C_i$ and $C_{i'}$. Let $\hat{C}_1,\dots,\hat{C}_{k_0}$ be the resulting partition of $V$, where $\hat{C}_{i'}$ is the result of merging all incomplete clusters $C_i$ into the complete cluster $C_{i'}$. Notice that every vertex in $\hat{C}_i$ can be reached from $x_i$ in at most $\lambda/2$ steps. This follows since every point from $C_i$ could be reached from $x_i$ in at most $\lambda/6$ steps, each cluster $C_j$ merged into $C_i$ to construct $\hat{C}_i$ had an edge into $C_i$, and the diameter of each such $C_j$ merged into $C_i$ was at most $\lambda/3$. It follows that the diameter of $\hat{C}_i$ is at most $\lambda$ as desired. Moreover, notice that $|\cI_i| \geq \lfloor \lambda / 12 \rfloor$ for each $i \in [k_0]$, and $\cI_i \subset \hat{C}_i$. Thus, the partition $\hat{C}_1,\dots,\hat{C}_{k_0}$ and independent set $\cI = \cup_{i=1}^{k_0} \cI_i$ satisfy the desired properties of the Lemma (with the parameter $k_0$).

Finally, suppose we are in the case that $\lambda < 6$. Then pick an arbitrary maximal independent set $\cI$ in $G$, and let $k$ be the size of $\cI$. Since $\cI$ is maximal, for each $v \in V \setminus \cI$ there is a $u \in \cI$ such that $(u,v) \in E$. Thus assign each vertex $v \notin \cI$ to an arbitrary vertex in $\cI$ adjacent to $v$, resulting in a partition $C_1,\dots,C_k$, and notice that this partition and the resulting independent set $\cI$ satisfy the desired properties of the Lemma.
\end{proof}

We start by dealing the second sum in (\ref{sumsum}). 
To this end, we define
\begin{equation}
    x_t^j(p) =  \frac{1}{|\CC(p,G_t(\cB(p, 2^j t,V_t)))|}
\end{equation}
for any $j \geq 0$ and $p \in V_t$.
Note that by definition we always have $x_t^j(p)\ge x_t(p)$.
We prove the following lemma:


\begin{lemma}\label{hehelemma1}
For any $t\in \calL$ and $j\ge 0$, we have 
\begin{equation}\label{eq:usesoon} t \sum_{p \in V_t} \Big(x_t^j(p) -x_t(p)\Big) \leq \frac{36 }{2^j}\cdot \cost(P). 
\end{equation}
\end{lemma}
\begin{proof}
Let $C$ be a connected component of $G_t$ with shortest path diameter less than $2^j$. We have $x_t^j(p)=x_t(p)$ for all $p\in C$ and thus, the contribution to the sum on the LHS from $p\in C$ is $0$. 
Below we let $C_1,\dots,C_\ell$  denote  connected components of $G_t$ with shortest path diameter at least $2^j$.
For each such $C_i$, decompose each $C_i$ into $k_i$ components $C_{i,1},C_{i,2},\dots,C_{i,k_i}$ via Lemma \ref{lem:graphDecomp} using the diameter bound $\lambda = 2^j$. We  claim that 
\begin{equation}\label{sumsumsum}
\sum_{p\in   C_{i,h}} x_t^j(p)\le 1
\end{equation}
  for all $i,h$, from which we have that the sum on the LHS of (\ref{eq:usesoon}) can be bounded from above by $\sum_{i\in [\ell]} k_i$.
To show (\ref{sumsumsum}), notice that for any $p ,q\in C_{i,h}$, since the shortest-path diameter of $C_{i,h}$ is at most $2^j$, it follows that there is a path of length at most $2^j$ from $p$ to $q$ in $G_t$. Since each such edge corresponds to a distance of at most $t$ in $(\R^d,\ell_1)$, it follows by the triangle inequality that $\|p-q\|_1 \leq 2^j t$. 
It follows that for all $p \in C_{i,h}$, we have $C_{i,h} \subseteq \CC(p,G_t(\cB(p,2^j t,V_t)))$. Thus, 
$$\sum_{p \in   C_{i,h}} 
x^j_t(p)\le |C_{i,h}|\cdot \frac{1}{|C_{i,h}|} = 1.$$

Next, we prove that $$\cost(P) \geq \frac{t 2^j}{36}\cdot \sum_{i\in [\ell]} k_i.$$ To see this, note that if $\cI_i$ is the independent set of $C_i$ of size at least $k_i \cdot \max\{1, \lfloor 2^j/12 \rfloor \}$ as guaranteed by Lemma \ref{lem:graphDecomp}, then $\cup_{i\in [\ell]} \cI_i$ is an independent set of $G_t$ of size  $s = \max\{1, \lfloor 2^j/12 \rfloor \} \cdot \sum_{i\in [\ell]} k_i$. 
This holds simply because the $C_i$'s are distinct connected components of $G_t$. Thus, there exists a set of $T$ points in $V_t$ that are pairwise distance at least $t$ apart.
Moreover, for each of these points $p$, there is a point $q\in P$ such that $\|p-q\|_1\le t/\beta$ (and $\beta\ge 1000$).
It follows that 
$$
\cost(P)\ge (s-1)\cdot (1-2/\beta)\cdot t\ge \frac{st}{3}\ge \frac{t2^j}{36}\cdot \sum_{i\in [\ell]}k_i  
$$ and thus,
$$
 t\sum_{i\in [\ell]} k_i\le \frac{36}{2^j}\cdot \cost(P).
$$
This finishes the proof of the lemma.
\end{proof}

We use the lemma above to  bound the second sum in (\ref{sumsum}):

\begin{corollary}
We have 
$$
\sum_{t\in \calL} t \sum_{ p\in B_{t,+}^*}\Big( y_t(p)-x_t(p)\Big)\le O\left(\frac{1}{\delta\eps\alpha}\right)\cdot \cost(P).
$$
\end{corollary}
\begin{proof}
Note that $y_t(p)=x_t^{\log L}(p)$ for $p\in B_{t,+}^*$. Using $x_t^j(p)\ge x_t(p)$, we have from Lemma \ref{hehelemma1} that
$$
t\sum_{p\in B_{t,+}^*}\Big(y_t(p)-x_t(p)\Big)
\le t\sum_{p\in V_t} \Big(x_t^{\log L}(p)-x_t(p)\Big)\le \frac{36}{L}\cdot \cost(P).
$$
Summing over $t\in \calL$ we have 
$$
\sum_{t\in\calL} t\sum_{p\in B_{t,+}^*} \Big(y_t(p)-x_t(p)\Big) 
\le  O(\cost(P)).
$$
The statement follows using $\delta\le 1/(\eps\alpha)$ from the definition of $\delta$ and thus, $1/(\delta\eps\alpha)\ge 1$.
\end{proof}
Before dealing with $B_{t,j}^*$ with $j<\log L$, we prove the following lemma about $j$-bad points:

\begin{proposition}\label{prop:5}
Let $p\in V_t$ be a point that is $j$-bad at level $t$ for some $j\ge 1$, and let $q\in V_t$ be a point with $\|p-q\|_1\le 2^{j-1} t $. 
Then $q$ is $j'$-bad at level $t$ for some $j'$ that satisfies $|j-j'|\le 1$.
\end{proposition}
\begin{proof}
First we notice that
\[ \cB(p, 2^{j+1} t, V_t) \subseteq \cB(q, (2^{j+1} + 2^{j-1}) t,V_t) \subseteq  \cB(q, 2^{j+2} t,V_t)\]
and thus, $|\cB(q,2^{j+2}t,V_t)|\ge \beta^2\Delta^{10\eps}$.
On the other hand, we have that 
\[\cB(q, 2^{j-1} t,V_t) \subseteq \cB(p,  (2^{j-1}+2^{j-1})  t,V_t) \subseteq \cB(p, 2^{j} t, V_t) \] 
and thus, $|\cB(q,2^{j-1}t,V_t)|<\beta^2 \Delta^{10\eps}$. This completes the proof.

\end{proof}

The next two lemmas deal with the sum of $y_t(p)$ over $p\in B_{t,j}^*$:

\begin{lemma}\label{lem:offlineAlgoMain}
Fix $t\in \calL$ and $j:0\le j<\log L$.
There are three sets $N_{t,j}^{j-1}, N_{t,j}^{j},N_{t,j}^{j+1}\subseteq V_t$ such that  
\begin{enumerate}
\item $N_{t,j}^{\ell}$, for $\ell\in \{j-1,j,j+1\}$, is a set of $\ell$-bad points at level $t$ with pairwise distance at least $t$; 
\item They satisfy
\end{enumerate}
\[\sum_{p\in B_{t,j}^*} y_{t}(p)  \leq O\left(\frac{1}{2^j}\right)\cdot \Big(|N_{t,j}^{j-1}|+|N_{t,j}^{j}|+|N_{t,j}^{j+1}|\Big).\]
Moreover, when $j=0$ we require $N_{t,j}^{j-1}$ to be empty.
\end{lemma}
\begin{proof}
We first deal with the special case of $j =0$. Define $\tilde{G}_{t/2} = (V_t, \tilde{E}_{t/2})$ be the $t/2$ threshold graph on the vertices in $V_t$. In other words, $(x,y) \in  \tilde{E}_{t/2}$ if and only if $\|x-y\|_1 \leq t/2$ (notice that $\tilde{G}_{t/2} = G_{t/2}$ whenever $t,t/2$ are in the same block $\cQ_i$). Let $\cI \subseteq  B_{t,j}^*$ be any maximal independent set of the points in $B_{t,j}^*$ in the graph $\tilde{G}_{t/2}$, and suppose $|\cI| = k$. Consider an arbitrary partition $C_1,\dots,C_k$ of $B_{t,j}^*$, where each $C_i$ is associated to a unique $y_i \in \cI$, and for any $y \in B_{t,j}^* \setminus \cI$, we add $y$ to an arbitrary $C_i$ for which $\|y_i - y\|_1 \leq t/2$ (such a $y_i \in \cI$ exists because $\cI$ is a maximal independent set). Then we can just add $\cI$ to $N_{t,j}^j$, which results in $|N_{t,j}^j| \leq |\cI|$. Moreover, notice that the points in $\cI$ are pairwise distance at least $t/2$ since they form an independent set in $\tilde{G}_{t/2}$.
We now claim
\[\sum_{p \in B_{t,j}^* } y_t(p) \leq k\]
To see this, notice that each $C_i$ has diameter at most $t$, thus for any point $p \in C_i$ we have $C_{i} \subseteq \CC(p,G_t(\cB(p, t,V_t)))$. Therefore, $y_t(p) \leq 1/|C_i|$ for all $p \in C_i$, from which the claim follows.

We now deal with the general case when $j\ge 1$.
Let $\cC_1,\ldots,\cC_\ell$ be connected components of $G_t$ that contain at least one point in $\smash{B_{t,j}^*}$.
This implies that 
  each $\cC_h$ has diameter at least $2^j$ given that points in $B_{t,j}^*$ are incomplete.
For each $\cC_i$ we can now partition it into $C_{i,1}, \dots,C_{i,k_i}$  via Lemma \ref{lem:graphDecomp}, setting the diameter bound on the $C_{i,j}$'s to be $\lambda = 2^{j-1 }$, but only keep $C_{i,j}$'s  which overlaps with $B_{t,j}^*$. Thus 
\[\sum_{p \in B_{t,j}^* }  y_{t}(p) \le  \sum_{i=1}^\ell \sum_{j=1}^{k_i}\sum_{p \in C_{i,j}\cap B_{t,j}^*}  y_{t}(p).\]
Similar to arguments used in the proof of Lemma \ref{hehelemma1}, we claim that $\smash{\sum_{p \in C_{i,j}}  y_{t}(p) \leq 1}$. To see this, note that $C_{i,j} \subseteq \CC(p,G_t(\cB(p,2^j t,V_t)))$ since the diameter of $C_{i,j}$ in the shortest path metric is at most $2^j$ and thus,  the diameter of the set of points $C_{i,j}\subset \R^d$ is at most $2^j t$. Thus $y_t(p) \leq  {1}/{|C_{i,j}|}$, from which the claim follows. Using this, we have  
\begin{equation}\label{eqn:yBound1}
    \sum_{p \in B_{t,j}^* } y_{t}(p) \leq \sum_{i\in [\ell]} k_i.
\end{equation}



 Let $s= \sum_{i\in [\ell]} k_i\max(1,2^{j-1 }/12)$ and 
$X = \{x_1,\dots,x_{s}\}$ be the independent set of $G_t$ obtained from the prior lemma (by taking the union of independent sets from each connected component $\cC_i$).
Given that they form an independent set of $G_t$ we have that their pairwise distance is at least $t$.
Furthermore, each $x\in X$ is in one of the $C_{i,j}$'s and  thus, satisfies
  $\|x-y\|_1\le 2^{j-1}t$ for some point $y\in B_{t,j}^*$.
It follows from Proposition \ref{prop:5} that  every $x\in X$ is $\ell$-bad at level $t$
for some $\ell\in \{j-1,j,j+1\}$. Partitioning $X$ into $\smash{N_{t,j}^{j-1},N_{t,j}^j}$ and $\smash{N_{t,j}^{j+1}}$ accordingly, and using (\ref{eqn:yBound1}), finishes the proof of the lemma.\end{proof}

Let $\tau=\lceil (\log \alpha)+1 \rceil$.
For $j$ with $2^j<\alpha$, the following lemma gives a better bound for $B_{t,j}^*$:

\begin{lemma}\label{lem:offlineAlgoMain2}
Fix $t\in \calL$ and $j\ge 0$ with $2^j<\alpha$.
There are sets $N_{t,j},N_{t,j}^{0}, \ldots,  N_{t,j}^{\tau}\subseteq V_t$ such that  
\begin{enumerate}
\item 
$N_{t,j}$ is a set of bad points at level $t$ such that all $p,q\in N_{t,j}$ satisfy $\cB(p,t,V_t)\cap \cB(q,t,V_t)=\emptyset$;
\item 
$N_{t,j}^{\ell}$, for each $\ell\in [0:\tau]$, is a set of $\ell$-bad points at level $t$ with pairwise distance at least $t$; 
\item Together they satisfy
\end{enumerate}
\[\sum_{p\in B_{t,j}^*} y_{t}(p)  \leq O\left(\frac{1}{\alpha}\right)\cdot \Big(|N_{t,j}|+|N_{t,j}^{0}|+\cdots +|N_{t,j}^{\tau}|\Big).\]
\end{lemma}
\begin{proof}
The statement trivially follows from Lemma \ref{lem:offlineAlgoMain} when $\alpha$ is bounded from above by a constant.
So we assume below that $\alpha$ is sufficiently large.
Let $\cC_1,\ldots,\cC_\ell$ be connected components of $G_t$ that contain at least one point in $\smash{B_{t,j}^*}$.
This implies that 
  each $\cC_h$ has diameter at least $\alpha$ given that points in $B_{t,j}^*$ are incomplete: Here we used that $\BFS(p,G_t,\alpha)=\CC(p,G_t)$ also implies that $p$ to be complete.

For each $\cC_i$ we can now partition it into $C_{i,1}, \dots,C_{i,k_i}$  via Lemma \ref{lem:graphDecomp}, setting the diameter bound on the $C_{i,j}$'s to be $\lambda = \alpha/4$, but only keep $C_{i,j}$'s  which overlaps with $B_{t,j}^*$. Thus 
\[\sum_{p \in B_{t,j}^* }  y_{t}(p) \le  \sum_{i=1}^\ell \sum_{j=1}^{k_i}\sum_{p \in C_{i,j}\cap B_{t,j}^*}  y_{t}(p).\]
We claim that $\smash{\sum_{p \in C_{i,j}}  y_{t}(p) \leq 1}$, which follows from $C_{i,j} \subseteq \BFS(p,G_t,\alpha)))$
for every $p\in C_{i,j}$.
Then    
\begin{equation}\label{eqn:1}
    \sum_{p \in B_{t,j}^* } y_{t}(p) \leq \sum_{i\in [\ell]} k_i.
\end{equation}  
Let $s= \sum_{i\in [\ell]} k_i\cdot \alpha/48 $ and 
$X = \{x_1,\dots,x_{s}\}$be the independent set of $G_t$ obtained from the prior lemma (by taking the union of independent sets from each connected component $\cC_i$).
Given that they form an independent set of $G_t$ we have that their pairwise distance is at least $t$.

The following claim is a generalization of Proposition \ref{prop:5}:

\begin{claim}
Let $p\in V_t$ be a point that is $j$-bad at level $t$ for some $j\ge 0$ with $2^j<\alpha$, and let $q\in V_t$ be a point with $\|p-q\|_1\le \alpha t/4 $. 
Then $q$ is either bad at level $t$ or $j'$-bad at level $t$ for some $j'\le \tau$.
\end{claim}
\begin{proof}
First we notice that
\[ \cB(p, 2^{j+1} t, V_t) \subseteq \cB(q, (2\alpha + \alpha/4) t,V_t) \subseteq  \cB(q, 2^{\tau+1} t,V_t)\]
and thus, $|\cB(q,2^{\tau+1}t,V_t)|\ge \beta^2\Delta^{10\eps}$.
The claim follows.

\end{proof}
The claim leads to a partition of $X$
  into $N_{t,j},N_{t,j}^0,\ldots,N_{t,j}^\tau$ that finishes the proof.
\end{proof}

Finally we bound the contribution from $B_{t,j}^*$ with $j:0\le j<\log L$:

\begin{lemma}\label{thm:mainupper}
Suppose the conclusion of Lemma \ref{lem:quadtreeLemma} holds, we have
$$
\sum_{t\in\calL} \sum_{j=0}^{\log L-1}\sum_{p\in B_{t,j}^*}t\cdot y_t(p)\le  O\left(\frac{\tau+1}{\delta\eps\alpha}\right)\cdot \cost(P).
$$
\end{lemma}
\begin{proof}
We start by bounding the sum over $j$ such that $2^j\ge \alpha$.
By Lemma \ref{lem:offlineAlgoMain}, we have
\begin{align*}
\sum_{t\in\calL}& \sum_{j:2^j\ge \alpha}\sum_{p\in B_{t,j}^*}t\cdot y_t(p)\\&\le
\sum_{t\in \calL} \sum_{j:2^j\ge \alpha}
O\left(\frac{t}{2^j}\right)\cdot \Big(|N_{t,j}^{j-1}|+|N_{t,j}^j|+|N_{t,j}^{j+1}|\Big)\\
& =\sum_{j:2^j\ge \alpha}
\sum_{t\in \calL} O\left(\frac{t}{2^j}\right)\cdot |N_{t,j+1}^j|+
\sum_{j:2^j\ge \alpha}
\sum_{t\in \calL} O\left(\frac{t}{2^j}\right)\cdot |N_{t,j }^j|+
\sum_{j:2^j\ge \alpha}
\sum_{t\in \calL} O\left(\frac{t}{2^j}\right)\cdot |N_{t,j-1}^j|
\end{align*}
For the first sum, we further split $t\in \calL$ into blocks $\calQ_i$:
$$
\sum_{j:2^j\ge \alpha}
\sum_{t\in \calL} O\left(\frac{t}{2^j}\right)\cdot |N_{t,j+1}^j| 
=\sum_{j:2^j\ge \alpha}
\sum_{\calQ} \sum_{t\in \calQ} O\left(\frac{t}{2^j}\right)\cdot |N_{t,j+1}^j| 
$$
Recall that each $N_{t,j+1}^j$ is a set of $j$-bad points at level $t$ with pairwise distance at least $t$.
Given that there are $O(1/\eps)$ many blocks and each block $\calQ$ can be further split into $O(1/\delta)$ many $S$ such that every $t,t'\in S$ with $t>t'$ satisfy $t\ge 4t'$,  
  it follows from the structural lemma (Lemma \ref{lem:structural}) that 
$$
\sum_{j:2^j\ge \alpha}
\sum_{t\in \calL} O\left(\frac{t}{2^j}\right)\cdot |N_{t,j+1}^j|\le 
\sum_{j:2^j\ge \alpha} O\left(\frac{1}{\delta\eps 2^j}\right) \cdot O(\cost(P))\le O\left(\frac{1}{\delta\alpha\eps}\right)\cdot O(\cost(P)).
$$
The other two sums can be bounded similarly.
For $j$'s with $2^j<\alpha$, it follows from
  Lemma \ref{lem:offlineAlgoMain2} that
\begin{align*}
\sum_{t\in\calL}&\sum_{j:2^j< \alpha}\sum_{p\in B_{t,j}^*}t\cdot y_t(p) \\ &\le
\sum_{t\in \calL} \sum_{j:2^j<\alpha}
  O\left(\frac{t}{\alpha}\right)\cdot \Big(|N_{t,j}|+|N_{t,j}^0|+\cdots+
  |N_{t,j}^\tau|\Big)
  \\ &= \sum_{j:2^j<\alpha} \sum_{t\in \calL}O\left(\frac{t}{\alpha}\right) \cdot |N_{t,j}|+
  \sum_{\ell=0}^{\tau} \sum_{j:2^j<\alpha}
  \sum_{t\in \calL} O\left(\frac{t}{\alpha}\right)\cdot |N_{t,j}^\ell|.
\end{align*}
The second sum can be similarly bounded from above by $O((\tau+1)/(\delta\alpha\eps))\cdot O(\cost(P))$.

For the first sum, we use the following claim:

\begin{claim}\label{useagain}
Suppose the conclusion of Lemma \ref{lem:quadtreeLemma} holds.
Let $N_t\subseteq V_t$ be a set of bad points 
  at level $t$ such that all $p,q\in N_t$ satisfy $\cB(p,t,V_t)\cap \cB(q,t,V_t)=\emptyset$.
Then we have 
$$
|N_t|\le O\left(\frac{1}{\Delta^{9\eps}t}\right)\cdot \cost(P).
$$
\end{claim}
\begin{proof}
Let $\calQ_i$ be the block that contains $t$ and let $T=\Delta^{(i-1)\eps}$. So $V_t=V_T$.
By Lemma \ref{lem:quadtreeLemma}, we have 
\begin{equation*}
|V_T|\le O\left(\frac{\beta d}{T\eps}\right)\cdot \cost(P)\le 
O\left(\frac{\beta^2}{T}\right)\cdot \cost(P).
\end{equation*}
On the other hand,
since every $p\in N_t$ is bad at level $t$, we have 
$|\cB(p , t,V_t) | \geq  \beta^2 \Delta^{10\eps}.$
Furthermore, we have 
$ 
\cB(p , t,V_t) \cap  \cB(q, t,V_t) = \emptyset
$ 
for any $p,q\in N_t$. As a result, we have 
$$
|V_T|=|V_t|\ge \sum_{p\in N_t} |\cB(p ,t,V_t)|\ge |N_t|\cdot \beta^2\Delta^{10\eps}.
$$
Combining the two inequalities, we have 
$$
|N_t|\le \frac{1}{\beta^2\Delta^{10\eps}} \cdot O\left(\frac{\beta^2}{T}\right)\cdot \cost(P) \le O\left(\frac{1}{\Delta^{9\eps}t}\right)\cdot \cost(P)
$$
using $t\le \Delta^\eps T$.
\end{proof}
Using the above claim, we have
$$
\sum_{j:2^j<\alpha} \sum_{t\in \calL}O\left(\frac{t}{\alpha}\right) \cdot |N_{t,j}|\le \tau\cdot O\left(\frac{tL}{\alpha}\right)\cdot O\left(\frac{1}{\Delta^{9\eps}t}\right)\cdot \cost(P)\le O\left(\frac{\tau}{\delta\eps\alpha}\right)\cdot \cost(P) 
$$
using our choices of parameters.
\end{proof}
Lemma \ref{mainanalysislemma1} follows by combining Lemma \ref{hehelemma1} and Lemma \ref{thm:mainupper}.

\subsection{Lower Bounding the Estimator}\label{sec:lower}

We prove Lemma \ref{mainanalysislemma2} in this subsection.
Trivially we have
$$
\sum_{t\in \calL}t \sum_{\text{dead}\ p\in V_t}x_t(p)\le 
\sum_{t\in \calL}t \sum_{\text{bad}\ p\in V_t}x_t(p)+\sum_{t\in \calL}t\sum_{ p\in V_t: z_t(p)=0}x_t(p).
$$
Lemma \ref{mainanalysislemma2} follows  from the following two lemmas, which handle the two sums, respectively.
 


\begin{lemma}\label{prop:mainlower}
Suppose the conclusion of Lemma \ref{lem:quadtreeLemma} holds. We have
\[ \sum_{\text{bad}\ p\in V_t}   x_{t}(p) \leq O\left( \frac{1}{\Delta^{9\eps}t}\right)\cdot \cost(P). \]
\end{lemma}
\begin{proof}

Let $\cC_t$ be the set of all connected components in $G_t$. 
Let $\cC_t^*\subseteq \cC_t$ be the set of connected components that contain at least one bad point $p\in V_t$ at level $t$.
Then we have
$$
\sum_{\text{bad}\ p\in V_t}   x_t(p)\le \sum_{C\in \cC_t^*} \sum_{p\in C}  x_t(p)=|\cC_t^*|.
$$
Let $|\cC_t^*|=k$ and let $ C_1,\ldots,C_k$ be the connected components in $\cC_t^*$, and let $p_h\in V_t$ be a bad point at level $t$ in $C_h$ for each $h\in [k]$.
The lemma follows by applying Claim \ref{useagain} on $\{p_h:h\in [k]\}$. 
\end{proof}

\begin{lemma}\label{prop:mainlower2}
Suppose the conclusion of Lemma \ref{lem:quadtreeLemma} holds.
We have
\[ \sum_{p\in V_t:z_t(p)=0}   x_{t}(p) \leq   \frac{1}{\Delta^{9\eps}t}\cdot \cost(P). \]
\end{lemma}
\begin{proof}The flow of the proof is very similar. 
Let $k$ be the number of connected components of $G_t$ that 
  contain at least one point $p\in V_t$ with $z_t(p)=0$.
Then we have 
  the sum on the LHS is at most $k$.
On the other hand, each such connected component must contain at least one $z$-dead point and thus, has size at least $\beta^2 \Delta^{10\eps}$.
As a result, we have $|V_T|=|V_t|\ge k\beta^2\Delta^{10\eps}$.
Combining these with arguments in Claim \ref{useagain} finishes the proof.
\end{proof}






\subsection{Variance in Estimating $Z$}

In this section, we will assume Theorem~\ref{thm:mainEstimator}, and show that, (after a fixed random shift satisfying Lemma~\ref{lem:quadtreeLemma},) we may estimate $Z$ by randomly sampling vertices $p_1,\dots, p_k \sim V_{t}$ for each $t \in \calL$ for $k$ which is not too large.

\begin{algorithm}[ht]
\DontPrintSemicolon
	\caption{MST Estimator}\label{alg:mainOutside}
	\KwData{Point set $P \subset \R^d$, and an error parameter $\eps_0 > 0$.} 
\For{$t \in \cL$}{
Sample $p_1,p_2,\dots,p_k \sim V_t$, where we will set $k= \poly( \Delta^{\eps} d\beta)$.\;
Compute $\hat{Z}_t = (1/k)\sum_{i=1}^k \min\{y_t(p_i), z_t(p_i)\}$ (we will give the specifics of how to estimate $y_{t}(p_i)$ and $z_{t}(p_i)$ with linear sketching in the subsequent sections). \\
Compute estimate an estimate of $\hat{n}_t$ which is $1$ if $|V_t|$ is one, and otherwise satisfies $\hat{n}_t = (1 \pm \eps_0)|V_t|$. \\
Compute an estimate $\hat{n}_0$ which will be $(1\pm \eps_0) |P|$. 
}
Output
\[ \hat{Z} = \hat{n}_0 - (1+\delta)^{L+1} + \delta \sum_{t \in \calL} t\cdot \hat{n}_{t} \cdot \hat{Z}_t.  \]
\end{algorithm}

\begin{lemma}\label{lem:expectation-bound}
Suppose that:
\begin{itemize}
    \item The conclusion of Lemma~\ref{lem:quadtreeLemma} holds,
    \item For every $t \in \calL$, $\hat{n}_t = 1$ if $|V_t| = 1$ and $\hat{n}_t = (1\pm \eps_0) |V_t|$, and
    \item $\hat{n}_0 = (1\pm \eps_0) n$.
\end{itemize}
Then, over the random draw of $p_1,\dots,p_k \sim V_t$ over all $t \in \calL$
\begin{align*}
    \left| \ex{\hat{Z}} - Z \right| \leq O\left(\frac{\eps_0}{\Delta^{\eps} \beta d \log\Delta} \right) \cdot \cost(P).
\end{align*}
\end{lemma}

\begin{proof}
We first show the upper bound. We will break up $\calL$ into two parts, $\calL_1$ consists of $t \in \calL$ where $|V_t| = 1$, and $\calL_2 = \calL \setminus \calL_1$. Then, we have
\begin{align*}
    \ex{\hat{Z}} &\leq (1+\eps_0) n - (1+\delta)^{L+1} + \delta \sum_{t \in \calL_1} t \ex{\hat{Z}_t} + \delta \sum_{t \in \calL_2} (1+\eps_0) |V_t| \ex{\hat{Z}_t} \\
                &= Z + \eps_0 n + \delta \eps_0 \sum_{t \in \calL_2} t |V_t|,
\end{align*}
where we have used the fact that $y_t(p)$ and $z_t(p)$ are both at most $1$, as well as the definition of $Z$. Finally, note that because the minimum non-zero distance is $1$, $\cost(P) \geq n-1$. Furthermore, since all $t \in \calL_{2}$ have $|V_t| \geq 2$, we may always upper bound $|V_t| \leq 2 ( |V_t| - 1)$. Thus, we have
\[ \ex{\hat{Z}} - Z \leq \eps_0 \cost(P) + 2\delta \eps_0 \sum_{t \in \calL} t (|V_t| - 1).  \]
To complete the proof, we note that we may sum over $\calL$ by  summing over $\calT$ and multiplying by $\Delta^{\eps} \cdot O(\log_{1+\delta}(\Delta^{\eps}))$. The term $\Delta^{\eps}$ accounts for the fact that for each $t \in \calL$, there is some $T \in \calT$ where $t / T \leq \Delta^{\eps}$. The term $O(\log_{1+\delta}(\Delta^{\eps}))$ accounts for the fact that there are that many such $t \in \calL$ which would use the same block $T \in \calT$. Recall that for any $t \in \calL$, we have $V_t = V_T$ for the largest $T = \Delta^{(i-1) \eps}$ which is at most $t$. Thus,
\[\sum_{t \in \calL} t \left( |V_t| - 1 \right) \leq \Delta^{\eps} \cdot O(\eps \log \Delta) \sum_{T \in \calT} T (|V_T| - 1) = O(\Delta^{\eps}) \cdot O(\beta d/\delta \cdot \log\Delta) \cdot \cost(P). \]
In other words, $\ex{\hat{Z}} - Z$ is at most
\[ \eps_0 \left(1 + O\left(\Delta^{\eps} \cdot \beta d \cdot \log \Delta \right)\right) \cdot \cost(P), \]
by using Lemma~\ref{lem:quadtreeLemma}. We note that the lower bound proceeds analogously, except we consider upper bounding $Z - \ex{\hat{Z}}$, by $\eps_0 n + \delta \eps_0 \sum_{t} t (|V_t|-1)$.
\end{proof}

\begin{lemma}\label{lem:variance-bound}
Similarly to Lemma~\ref{lem:expectation-bound}, suppose that:
\begin{itemize}
    \item The conclusion of Lemma~\ref{lem:quadtreeLemma} holds,
    \item For every $t \in \calL$, $\hat{n}_t =1$ if $|V_t| = 1$ and $\hat{n}_t = (1\pm \eps_0) |V_t|$, and
    \item $\hat{n}_0 = (1\pm \eps_0) n$. 
\end{itemize}
then, over the random draw of $p_1,\dots, p_k \sim V_t$ over all $t \in \calL$, 
\begin{align*}
    \Var{\hat{Z}} \leq O\left(\dfrac{\Delta^{2\eps} \log^2 \Delta \cdot (d \beta)^2}{k} \right) \cost(P)^2.
\end{align*}
\end{lemma}

\begin{proof}
We will use the following facts to upper bound the variance. First, the sampling is independent for $t \neq t' \in \calL$, which will allow us to handle the variance of each term individually. Second, whenever $t \in \calT$ is such that $|V_t| = 1$, then $\Var{Z_t} = 0$; indeed, all samples $p_1,\dots, p_k \sim V_t$ will be the same. So, it will again be useful to consider dividing $\calL$ into $\calL_1 = \{ t \in \calL : |V_t| = 1\}$, and $\calL_2 = \calL \setminus \calL_1$. Finally, we always have $y_t(p)$ and $z_t(p)$ is at most $1$. Putting things together, we have
\begin{align*}
    \Var{\hat{Z}} &\leq \delta^2 \sum_{t \in \calL_2} t^2 \cdot \hat{n}_t^2 \cdot \frac{1}{k} \leq O(1)\cdot \frac{\delta^2}{k} \cdot \left( \sum_{t\in \calL} t \cdot (|V_t| - 1)^2 \right)^2 \\
    &= O\left(\dfrac{\Delta^{2\eps} \log^2 \Delta \cdot (d \beta)^2}{k} \right) \cost(P)^2,
\end{align*}
where in the second line, we used the same inequality as in Lemma~\ref{lem:expectation-bound} to upper bound $\sum_{t \in \calL} t (|V_t| - 1)$ by $O(\Delta^{\eps} \log \Delta \beta d / \delta) \cdot \cost(P)$.
\end{proof}
\newcommand{\ignore}[1]{}
\ignore{
\begin{proposition}
Fix $\gamma = O(\frac{1}{\beta  })$
Suppose that for each $t=2^i$ for $i=1,2,\dots,\log \Delta$, we have $|y_t - \ex{y_t}| \leq \gamma$. Then conditioned on $\ecost(P) \leq \beta \cost(P)$ we have $|Z - \ex{Z}| \leq \cost(P)/10$.
\end{proposition}
\begin{proof}
We can bound the error in our estimator via
\begin{align*}
  |  \gamma \sum_{i=1}^{\log\Delta} 2^i n_{2^i} | &\leq |  2\gamma \sum_{i=1}^{\log\Delta} 2^i |V_t|  \\
  & \leq \gamma \ecost(P) \\
  \leq \frac{\cost(P)}{10}
\end{align*}
\end{proof}

Putting together the expectation and variance bound, and taking $O(\frac{\log \Delta}{\gamma^2})$ total samples, we obtain the following:
}

%% file: pPassSketching.tex
\section{A $\alpha$-Pass Linear Sketch for MST}\label{sec:pPass}

In this section, we demonstrate how for any $\alpha \geq 2$, the Algorithm from Section \ref{sec:estimatorMain} can be simulated via an $\alpha$-pass linear sketch in $n^{O(\eps)} d^{O(1)}$ space. As a corollary, we will obtain the first MPC algorithms which achieve a $1 + O(\log\alpha / (\alpha \eps))$ approximation in $\alpha +1$ rounds of communication, using space per machine $\Delta^{O(\eps)} d^{O(1)}$ (and up to another multiplicative $(1\pm \eps')$-error, recall that we may assume $\Delta \leq nd / \eps'$) and $O(n)$ total space.  
In particular, the following theorem is implied by Theorem~\ref{thm:mainupper}, as well as the implementation details which we specify below. Recall that in the linear sketching model, the input $P \subset [\cord]^d$ is represented by an indicator vector $x \in \R^{\cord^d}$.

\begin{theorem}\label{thm:alphapass}
Let $P \subset [\cord]^d$ be a set of $n$ points in the $\ell_1$ metric. Fix $\eps \in (0, 1)$ and any integer $\alpha$ such that $2 \leq \alpha \leq n^\eps$.  Then there exists an $\alpha$-round $n^{O(\eps)} d^{O(1)}$-space linear sketch which, with high probability, returns an estimate $\hat{Z}\geq 0$ which satisfies
\[  \cost(P) \leq \hat{Z} \leq \left(1 + O\left(\frac{ \log \alpha +1}{\eps \alpha}\right)\right) \cdot \cost(P). \]
\end{theorem}

\begin{proof}
In Subsection~\ref{sec:y-est-imp} and Subsection~\ref{sec:z-est-imp}, we give a $1$- and $\alpha$-round linear sketch which produces a sample the quantities $y_t(p)$ and $z_t(p)$ for $p \sim V_t$. These sketches may be used with along with a $\ell_0$ estimation sketch \cite{kane2010optimal} (see Lemma \ref{lem:l0}) to estimate $\hat{n}_t$ and $\hat{n}_0$ to accuracy $1\pm \eps_0$,  using $\tilde{O}(\eps_0^{-2} d \log \cord)$ space. Notice here that $|V_0| = n = \|x\|_0$, and moreover that $|V_t| = |V_t| = \|v^t\|_0$, where $v^t$ is the vector indexed by points in $V_t$, such that for any $a \in V_t$ we have $v^t_a = \sum_{p \in P, f_T(p) = a} x_p$. 
It is simple to enforce the property that $\hat{n}_t = 1$ if $|V_t| = 1$, by rounding to the nearest integer.

By Lemma~\ref{lem:expectation-bound} after we repeat the estimator from Algorithm~\ref{alg:mainOutside} with $k = \poly(\Delta^{\eps} d\beta)$ and $\eps_0 = 1 / \poly(\Delta^{\eps} \beta d)$ our estimate will concentrate around $Z$ with variance bounded by Lemma~\ref{lem:variance-bound}. By Theorem~\ref{thm:mainEstimator}, $Z$ is the desired estimate for $\cost(P).$
\end{proof}

\subsection{Preliminaries: Sketching Tools}

The implementation will be a straight-forward consequence of the following two lemmas, which give sketching algorithms for $k$-sparse recovery and $\ell_0$ sampling.
There are many known methods for $k$-sparse recovery -- for our purposees, we can use either the Count-Sketch of \cite{charikar2002finding} or the Count-Min Sketch of \cite{cormode2005improved}. 
\begin{lemma}[$k$-Sparse Recovery \cite{charikar2002finding, cormode2005improved}]
For $k, m \in \N$ where $k\leq m$ and $\delta > 0$, let $t = O(k\log(m/\delta))$. There exists a distribution $\cS^{(r)}(m,k)$ supported on pairs $(\bS, \Alg_{\bS}^{(r)})$, where $\bS$ is a $t\times m$ matrix, and $\Alg_{\bS}^{(r)}$ is an algorithm. Suppose $x \in \R^m$ is any vector, then with probability at least $1 - \delta$ over $(\bS, \Alg_{\bS}^{(r)}) \sim \cS^{(r)}(m, k)$, the following occurs:
\begin{itemize}
    \item If $x$ is a $k$-sparse vector, then $\Alg_{\bS}^{(r)}(\bS x)$ outputs a list of at most $k$ pairs $(i_1, x_{i_1}), \dots, (i_k, x_{i_k})$ consisting of the indices and values of non-zero coordinates of $x$.
    \item If $x$ is not a $k$-sparse vector, then $\Alg_{\bS}^{(r)}(\bS x)$ outputs ``fail.''
\end{itemize}
\end{lemma}

\newcommand{\indd}{\mathbf{1}}

\begin{lemma}[$\ell_0$ Sampling ~\cite{Jowhari:2011}]
For any $m \in \N$ and any $\delta \in (0, 1)$, let $t = O(\log^2 m \log \delta^{-1})$, and fix any constant $c \geq 1$. There exists a distribution $\cS^{(s)}(m)$ supported on pairs $(\bS, \Alg_{\bS}^{(s)})$ where $\bS$ is a $t \times m$ matrix, and $\Alg_{\bS}^{(s)}$ is an algorithm which outputs an index $i \in [m]$ or ``fail.'' Suppose $x \in \R^m$ is any non-zero vector, then $\Alg_{\bS}^{(s)}(\bS x)$ outputs ``fail'' with probability at most $\delta$, and for every $i \in [m]$,
\begin{align*}
    \Prx_{(\bS, \Alg_{\bS}^{(s)}) \sim \cS^{(s)}(n)}\left[ \Alg_{\bS}^{(s)}(\bS x) = i \mid \Alg_{\bS}^{(s)}(\bS x) \neq \text{``fail''}\right] = \frac{\indd\{ x_i \neq 0 \}}{\|x\|_{0}} \pm m^{-c}
\end{align*}
\end{lemma}

We remark that the small $1/\poly(\cord^d)$ variational distance from the $\ell_0$ sampling sketch above can be ignored by ours algorithms, as the total number of samples is significantly smaller than $m^c = \cord^{dc}$.

\begin{lemma}[$\ell_0$ Estimation ~\cite{kane2010optimal}]\label{lem:l0}
For any $m \in \N$ and any $,\eps , \delta \in (0, 1)$, let $t = \tilde{O}(\eps^{-2} \log m \log(1/\delta))$. There exists a distribution $\cS^{(e)}(m)$ supported on pairs $(\bS, \Alg_{\bS}^{(e)})$ where $\bS$ is a $t \times m$ matrix, such that given $x \in \in R^m$, the algorithm $\Alg_{\bS}^{(e)}(\bS x)$  outputs a value $R \in \R$ such that $R = (1 \pm \epsilon) \|x\|_0$ with probability at least $1-\delta$.
\end{lemma}

\subsection{Implementing the $y_{t}(p)$ Estimator}\label{sec:y-est-imp}

We now outline how to utilize sparse recovery and $\ell_0$-sampling in order to produce a sample $y_{t}(p)$, where $t \in \cL$ is a fixed level, and $p$ is drawn uniformly at random from $V_t$. For all the following sketches, we set $\delta = \cord^{-dc}$ for any arbitrarily large constant $c$. Moreover, we will set the precision parameter for the $\ell_0$ estimation sketch to $O(1/\alpha)$ with a small enough constant. 

\begin{enumerate}
    \item In the first round, we initialize an $\ell_0$-sampling sketch $(\bS, \Alg_{\bS}^{(s)}) \sim \cS^{(s)}(\cord^d)$. We may implicitly consider the vector $x$ in $\R^{\cord^d}$ where coordinates are indexed by points $p \in [\cord]^d$, and $x_p$ denotes the number of points $q \in P$ where $f_t(q) = p$. If the algorithm does not fail, the sample $p$ produced by $\Alg_{\bS}^{(s)}(\bS x)$ will be uniform among $V_t$. 
    
    \item In the second round, we initialize $\log L + 1$ sketches for $k$-sparse recovery, where $k = \beta^2 \Delta^{10\eps}$ and $m = \cord^d$. We let $(\bS_j, \Alg_{\bS_j}^{(r)}) \sim \cS^{(r)}(m, k)$ for each $j \in \{ 0, \dots, \log L\}$.
    \begin{itemize}
        \item For each $j \in \{0, \dots, \log L\}$, we will use the $j$-th sketch to recover $\cB(p, 2^j t, V_t)$, as long as the number of such points is at most $\beta^2 \Delta^{10\eps}$. We do this considering the vector $x^{(j)} \in \R^{\cord^d}$ where coordinates are indexed by points $p' \in [\cord]^d$, and $x_{p'}^{(j)}$ contains the number of points $q \in P$ where $f_T(q) = p'$ if $p' \in \cB(p, 2^j t, V_t)$, and is zero otherwise. We initialize them such  that they will succeed with high probability.
        \item Whenever $n_j \leq k$, we are able to recover the $k$-sparse vector $x^{(j)}$ from $\Alg_{\bS_j}^{(r)}(\bS_j x^{(j)})$, and whenever $n_j > k$, the algorithm outputs ``fail.'' Given this information, one can compute $y_t(p)$ exactly.
    \end{itemize}
\end{enumerate}

It is not hard to see that the sketch may be implemented in just two rounds, and uses space $\beta^2 \Delta^{10\eps} \cdot \poly(d \log\Delta)$. 
\subsection{Implementing the $z_t(p)$ Estimator}\label{sec:z-est-imp}

We now show how to implement $z_t(p)$. This case will use $\alpha \geq 2$ rounds of linear sketching in order to implement the breath-first search. Just as above, we may generate the sample $p$ drawn uniformly from $V_t$ by one round of $\ell_0$ sampling. Thus, we show how to implement the breath-first search, where we assume that we have already implemented the first $\ell < \alpha - 1$ rounds of the breath-first search (where in the beginning, we think of $\ell = 0$).
\begin{itemize}
    \item We assume that we maintain a set $E_{\ell} \subset [\cord]^d$ of ``explored'' vertices of $V_t$ at the end the $\ell$-th round (we initialize $E_{0}$ to $\{ p \}$).
    \item If $|E_{\ell}| \geq \beta^2 \Delta^{10\eps}$, we set $z_t(p) = 0$ and terminate. Otherwise, we let $k = \beta^2 \Delta^{10\eps} - |E_{\ell}| - 1$, $m = \cord^d$, and we initialize a $k$-sparse recovery sketch $(\bS, \Alg_{\bS}^{(r)}) \sim \cS^{(r)}(m, k)$ to succeed with high probability. We consider the vector $x \in \R^{\cord^d}$ where coordinates are indexed by points $q \in [\cord]^d$ and $x_q$ contains the number of points $p' \in P$ where $f_t(p') = q$, $q \notin E_{\ell}$, and there exists some $p'' \in E_{\ell}$ where $\|p' - p''\|_1 \leq t$.
    \item We run $\Alg_{\bS}^{(r)}(\bS x)$, and if it outputs ``fail,'' we are guaranteed that $|\BFS(p, G_t, \alpha)| \geq \beta \Delta^{10\eps}$. Otherwise, we recover $x$, and we let $E_{\ell+1}$ denote the subset of $V_t$ consisting of $E_{\ell}$ and all $p \in V_t$ where $x_p$ is non-zero.
\end{itemize}
After $\alpha-1$ rounds, we have either failed and we know that $\BFS(p, G_t, \alpha) \geq \beta^{2} \Delta^{10\eps}$, or we have recovered all of $\BFS(p, G_t, \alpha)$. Given this information, we recover $z_t(p)$. Similarly to above, the sketch can be implemented with $\alpha$ rounds and uses space $\beta^2 \Delta^{10\eps} \cdot \poly(d\log \Delta)$.

%% file: onePass.tex
\section{A Modified Estimator for One-Pass Streaming}\label{sec:estimatorOnePass}

In this section, we give a modified implementation of the Estimator from Section \ref{sec:estimatorMain}, which will have the property that it can be implemented in a single pass over a stream. In particular, we will implement it via a single-round linear sketch. Our construction will first require the development of  a specific class of locality sensitive hash functions, which we do in Section \ref{sec:LSH}. We then provide the modified estimator in Section \ref{sec:modifiedEstimator}, and our linear sketch that implements it in Section \ref{sec:onePassSketch}.

\paragraph{Setup.} Since the linear sketches in this section will be used primarily for one-pass streaming, we assume the points live in a discretized universe $P \subset [\cord]^d$, and we will use the $\ell_1$ metric. Notice that all distances are therefore in the range $[1,d \Lambda]$. In what follows, we will fix a precision parameter $\eps >0$  such that $\eps^{-1}$ is a power of $2$. Throughout, we assume that $\epsilon$ is smaller than some constant.  The parameter $\eps$ will govern the space usage of the algorithm when we implement it with a linear sketch in Section \ref{sec:onePassSketch}. Specifically,  the algorithm of that section will have space complexity at least $1/(1-\eps)^d$, but since if the original metric was high dimensional $\ell_2$ space, namely $(\R^d, \ell_2)$, we can always first apply dimensionality reduction into $(\R^{O(\log n)}, \ell_1)$ with at most a constant distortion. We can also make the assumption  that $\min_{x\neq y \in P} \|x-y\|_1 \geq 1/\epsilon^3$, and that $d\cord$ is a power of $1/\epsilon$. Notice that we can enforce this by scaling the input by $O(1/\eps^2)$, and rescaling $\cord$ by the same factor. 

\paragraph{Diameter Approximation.} Our algorithms will require an approximation of $\diam(P)$. Specifically, we assume that we have a value $\Delta$\footnote{For the remainder of Sections \ref{sec:estimatorOnePass} and \ref{sec:onePassSketch}, $\Delta$ will refer to this approximate diameter. } such that 

\[ \diam(P) \leq \Delta \leq  \frac{4}{\eps} \diam(P)\]

We demonstrate in Section \ref{sec:diam} how such an approximation can be recovered in space equivalent to the space of the rest of the algorithm. We can then run $O(\log (\cord d))$ copies of the algorithm, one for each guess of $\Delta$ in powers of $2$, and verify at the end which to use once we obtain the estimate  $\Delta$. Thus, by scaling up, we can also assume that $\Delta$ is a power of $2$.

Our algorithm will operate on the same graphs $G_t = (V_t, E_t)$ as defined in Section \ref{sec:estimatorMain}, with the setting of $\delta =1 $  so that $\cL = \{0,1,2,4,\dots,\Delta-1\}$ contains $\log \Delta$ powers of $2$. Given the diameter approximation, we have the following.

\begin{proposition}\label{prop:xConstFactor}
Let $P \subset ([\cord]^d, \ell_1)$, and suppose we have a value $\Delta$ such that $\diam(P) \leq \Delta \leq   \frac{8}{\eps} \diam(P)$, and let $\cL = \{1,2,4,\dots,d \Delta\}$. Suppose further that $\min_{x\neq y \in P} \|x-y\|_1 \geq 1/\eps^3$. Then we have

\[ \cost(P) /2 \leq \sum_{t \in \cL }  t \sum_{p \in V_t} x_t(p) \leq O(\eps^{-1}) \cost(P)\]
where $x_t(p) = |\CC(p,G_t)|^{-1}$ is as in (\ref{eqn:xt})
\end{proposition}
\begin{proof}

We first prove the upper bound. 
 Let $G_t^*$ be the true threshold graph for the points in $P$. Let $c_i$ be the number of connected components in $G_t$, and let $c_i^*$ be the number of connected components in $G_t^*$. By Lemma \ref{prop:xBound} (where $L = O(\log \Delta)$ is as in that Lemma), we have 

 \[  \sum_{t \in \cL } t \sum_{p \in V_t} x_t(p) - 2^{L+1} \leq n-  2^{L+1} +\delta \sum_{t\in \calL} t
  \sum_{p\in V_t} x_t(p) \leq 2 \cdot \cost(P)\]

Since $2^{L+1} \leq 8 \Delta \leq O(\eps^{-1}) \cost(P)$, the upper bound of the proposition follows. For the lower bound, by Corollary 2.2 of \cite{czumaj2009estimating}, we have
\[ \cost(P) \leq n/\eps^2  - 2\Delta + \sum_{t \in \cL , t \geq \eps^{-2}} t c_t^* \leq \eps\cost(P) -1 + \sum_{t \in \cL , t \geq \eps^{-2}} t c_t^* \]
Where we used that the minimum distance being $\eps^{-3}$ implies $\cost(P) \geq (n-1)/\eps^3$
Since we have $c_t \geq c_{t(1+\beta)}$. it follows that 

\[ \sum_{t \in \cL }  t \sum_{p \in V_t} x_t(p) \geq \frac{1}{2}\sum_{t \in \cL , t \geq \eps^{-2}} t c_t^* \geq \cost(P)/2\]
which completes the proof. 
\end{proof}

Thus, it will suffice to produce an approximation $R$ which satisfies:

\begin{equation}\label{eqn:goal}
\sum_{t\in \calL} t  \sum_{p\in V_t} x_t(p) \leq R \leq \sum_{t\in \calL} t  \sum_{p\in V_t} x_t(p) + O(\eps^{-2} \log \eps^{-1})  \cost(P)  
\end{equation}
Thus, the goal of this section will be to obtain an estimator which satisfies \ref{eqn:goal}.

\input{NewLSH.tex}


\begin{algorithm}[!h]
\DontPrintSemicolon
	\caption{Main Algorithm for level $t$}\label{alg:onepassMain}
	\KwData{Graph $G_t = (V_t,E_t)$, diameter approximation $\Delta$} 
		$(z,p) \leftarrow \fail$, set $T \leftarrow \Delta^\eps \lfloor t/\Delta^\eps \rfloor$, and set $\tau = \log_{1/\eps}(\log(\Delta)) + 3$\;
	\While{$(z,p)  \neq  \fail$}{
	Sample $j \sim [-1 \; : \; \tau]$ uniformly at random\; 
	\If{$j$=-1}{
$(z,p) \leftarrow$  Algorithm \ref{alg:onepassDead}\; 
	}
	\ElseIf{$j = \tau$}{
		$(z,p) \leftarrow$ Algorithm \ref{alg:onepassComplete} \; }
	\Else{
		 $(z,p) \leftarrow$ Algorithmn \ref{alg:onepassBad}\;
	}
	}
	Output $(z,p)$
\end{algorithm}

\begin{algorithm}[!h]
\DontPrintSemicolon
	\caption{Procedure for case $j=-1$}\label{alg:onepassDead}
	\KwData{Graph $G_t = (V_t,E_t)$} 
	Sample hash functions: $h_1 \sim \cH_{2\eps}(t/\eps), h_2 \sim \cH_{2\eps}(t/\epsilon^3)$, and set $h_1^{-1}(b) = \{p \in V_t \; | \; h_1(p) = b \}$\;
	Sample hash bucket $b$ with probability proportional to $\left| h_2^{-1}(b) \right|$, then $c$ with probability proportional to $\left| h_2^{-1}(b) \cap h_1^{-1}(c)\right|$ \;
		\If{$|h_2^{-1}(b) \cap h_1^{-1}(c)| > \beta^2 \Delta^{10\eps}$ }{ 
	    	Sample a uniform $p \sim h_2^{-1}(b) \cap h_1^{-1}(c)$\label{step:5Alg3}\;
		\If{$\tester(h_2,p) = 1$  and  $\tester(h_1,p) = 1$}{
	 set $z_t(p)= 0$, output $(z_t(p),p)$. \; 
		}\Else{ Output $\fail$ (could not recover a good point)\;
		}
		}
		\Else{
		Output $\fail$ ($j$ was not the correct level)\;
	}
\end{algorithm}

\begin{algorithm}[!h]
\DontPrintSemicolon
	\caption{Procedure for case $0 \leq j <\tau$}\label{alg:onepassBad}
	\KwData{Graph $G_t = (V_t,E_t)$} 
	Sample hash functions: $h_1 \sim \cH_{2\eps}(\frac{t}{\eps^{j}}), h_2 \sim \cH_{2\eps}(\frac{t}{\eps^{j+2}})$, and set $h_i^{-1}(b) = \{p \in V_t \; | \; h_i(p) = b \}$\;
	Sample hash bucket $b$ with probability proportional to $\left| h_2^{-1}(b) \right|$, then $c$ with probability proportional to $\left| h_2^{-1}(b) \cap h_1^{-1}(c)\right|$ \;
		\If{$|h_2^{-1}(b))| \leq \beta^2 \Delta^{10\eps}$}{
	Output $\fail$ ($j$ was not the correct level)\;
		}
		\ElseIf{$|h_2^{-1}(b) \cap h_1^{-1}(c)| > \beta^2 \Delta^{10\eps}$ }{ 
		Output $\fail$ ($j$ was the wrong level)\;
		}
		\Else{
		Recover the hash bucket $B = h_2^{-1}(b) \cap h_1^{-1}(c)$, and sample a uniform $p \sim B$. \;
		\If{$\tester(h_2,p) = 1$  and  $\tester(h_1,p) = 1$}{
		Set $z_t(p) = \frac{1}{|CC(p , G_t(B))|}$ and output $(z_t(p),p)$ \; 
		}\Else{ Output $\fail$ (could not recover a good point)\;
		}
	}
\end{algorithm}

\begin{algorithm}[!h]
\DontPrintSemicolon
	\caption{Procedure for case  $j = \tau$}\label{alg:onepassComplete}
	\KwData{Graph $G_t = (V_t,E_t)$} 
	Sample hash functions: $h_1 \sim \cH_{2\eps}(\frac{t}{\eps^{j}}), h_2 \sim \cH_{2\eps}(\frac{t}{\eps^{j+2}})$, and set $h_i^{-1}(b) = \{p \in V_t \; | \; h_i(p) = b \}$\;
	Sample hash bucket $b$ with probability proportional to $\left| h_2^{-1}(b) \right|$, then $c$ with probability proportional to $\left| h_2^{-1}(b) \cap h_1^{-1}(c)\right|$ \;
	\If{$|h_2^{-1}(b) \cap h_1^{-1}(c)| > \beta^2 \Delta^{10\eps}$ }{ 
		Output $\fail$ ($j$ was the wrong level)\;
		}
		\Else{
	Recover the full hash bucket $B = h_2^{-1}(b) \cap h_1^{-1}(c)$ \;
				Sample a uniform $p \sim B$\;
			\If{$\tester(h_2,p) = 1$  and  $\tester(h_1,p) = 1$}{
		Set $z_t(p) = \frac{1}{|CC(p , G_t(B))|}$ and output $(z_t(p),p)$ \; 
		}
		}
\end{algorithm}


\subsection{The Modified Estimator}\label{sec:modifiedEstimator}
We now define our modified estimator, which we will be able to simulate in a single-round linear sketch. Fix a precision parameter $\eps \in (0,1)$ such that $1/\eps$ is a power of $2$, and set $\tau = \log_{1/\eps}(\log(\Delta)) + 3$. The algorithm samples a integer $j \sim \{-1, \dots, \tau\}$ uniformly at random, and then runs Algorithm \ref{alg:onepassDead} if $j=-1$, Algorithm \ref{alg:onepassComplete} if $j = \tau$, and runs Algorithm \ref{alg:onepassBad} otherwise. For the remainder of the section, we will fix a level $t \in [1,\Delta]$, and set $T = \Delta^\eps \lfloor t/\Delta^\eps \rfloor$. Recall that the graph $G_t$ has vertex set $V_t = V_T$. 

We now introduce a modification of the earlier definitions, so that they increase in powers of $(1/\eps)$ instead of $2$. For the remainder of the section, we will employ the terminology used in the definitions below. 

\begin{definition}
Fix any point $p \in V_t$, and set $T = \Delta^\eps \cdot \lfloor t/\Delta^\eps \rfloor$. 
\begin{itemize}
    \item We say that $p$ is \textit{dead} if $|\cB(p,t/\eps,V_t)| > \beta^2 \Delta^{10 \eps}$. If $|\cB(p,t,V_t)| > \beta^2 \Delta^{10 \eps}$, then we call $p$ \textit{very dead}. If $p$ is not dead as say that $p$ is \textit{alive}. 
   
    \item For any integer $\ell \geq 0$, we say that $p$ is of type $\ell$ if $|\cB(p, t/\eps^{\ell+1}, V_t)| > \beta^2 \Delta^{10\eps} $ and $|\cB(p, t/\eps^{\ell}, V_t)|  \leq \beta^2 \Delta^{10\eps} $. 
    
     \item The point $p$ is called $\ell$-\textit{complete} if it is of type $\ell+3$ some $0 \leq \ell <\tau$ and if $CC(p,G_t) \subseteq \cB(p,t/\eps^{\ell})$. The point $p$ is called complete if $p$ is $\ell$-complete for some $\ell$. 
    \item We say that $p$ is nearly complete if $|\cB(p, (1/\eps)^{\tau+1} t , V_t)| \leq \beta^2 \Delta^{10\eps} $.
\end{itemize}
\end{definition}

Notice by the definition, and the fact that $(1/\eps)$ is a power of $2$ we have the following observation:
\begin{fact}\label{fact:alltypes}
For any point $p \in V_t$, the point $p$ is either dead, of type $j$ for some $j \in [\tau]$, or is nearly-complete. Similarly, for any point $p \in V_t$, the point $p$ is either very dead, type $j$ for some $j \in [0 \; : \; \tau]$, or is nearly-complete. Moreover, for any point $p \in V_t$ of type $j$, the point $p$ is $j'$ bad (Definition \ref{def:bad}) for some 
\[(1/\eps)^{j} \leq  2^{j'} \leq (1/\eps)^{j+1}\]
\end{fact}

\paragraph{Definitions of Distributions and Surviving.}
Let $(z,p)$ be the random output of Algorithm \ref{alg:onepassMain}, and let $\cD^t$ be the distribution over the output of $(z,p)$, and let $\hat{\cD}^t$ be the distribution over the output of a single round of the while loop in Algorithm \ref{alg:onepassMain}. In what follows, when $t$ is understood from context we drop the superscript. Notice that $\hat{\cD}$ is supported on $[0,1] \times V_t \cup \{\fail\}$, and $\cD$ is supported on $[0,1] \times V_t$ (i.e., $\cD$ is $\hat{\cD}$ conditioned on not outputting $\fail$). Notice that, a priori, it is not clear that $\cD$ is well defined, since we must first show that while loop terminates at  each step with non-zero probability (which we will do). Also define $\hat{\cD}_j$ to be the distribution of the output (which may be $\fail$) of a single round of the while loop in Algorithm \ref{alg:onepassMain} given that the index $j \in \{-1,0,1,\dots,\tau\}$ being sampled. 



For any point $q \in V_t$ and $j \in [-1 : \tau]$, define $\pi_{j,q} = (\tau+2)^{-1}\prb{(z,p) \sim \hat{\cD}_j}{p=q}$, and let $\pi = \sum_{j=-1}^{\tau}\sum_{q \in V_t}\pi_{j,q} < 1$.   Note that $\pi_{j,q}$ is the probability that $j \in [-1 : \tau]$ is sampled uniformly and $q$ is returned from $\hat{\cD}_j$, and $\pi$ is the probability that the algorithm does not return $\fail$.

Our algorithm will sample from the hash families $\cH_{2\eps}(t)$ for varying values of $t$. Since the precision parameter $2\eps$ will be fixed for the entire section, we drop the subscript and simply write $\cH(t)$. For a given hash function $h \sim \cH(t)$ and point $p$, we say that $p$ \emph{survives} $h$ if $\tester_t(h,p ) = 1$. For two hash functions $h_1 \sim \cH(t_1), h_2 \sim \cH(t_2)$, we say that $p$ \emph{survives} $(h_1,h_2)$ if $\tester_{t_1}(h_1,p) = 1$ and $\tester_{t_2}(h_2,p) = 1$. If we define $h_i^{-1}(b) = \{p \in V_t \; | \; h_i(p) = b \}$ for $i \in \{1,2\}$ and hash bucket $b$, then notice that if $p$ survives $h_i$, by definition we have 
\[\cB(p,t_i,V_t) \subseteq h_{i}^{-1}(h(p)) \subseteq \cB(p,t_i/\epsilon,V_t)\]
Where the last containment follows from the diameter bound on the hash buckets from Lemma \ref{lem:newLSH}.


\input{OnePassUpper.tex}

\input{OnePassLower.tex}


\subsection{Main Theorem for the Modified Estimator}
Given the upper and lower bounds on the expectation of the random variable returned by Algorithm \ref{alg:onepassMain} (Lemma \ref{lem:onePassUpperBound} and \ref{lem:onePassLowerBound})

The main result of this section is the following theorem.

\begin{theorem}\label{thm:OnePassEstimator}
Let $P \subset [\cord]^d$ be a set of $n$ points equipped with the $\ell_1$ metric, and let $\Delta$ be such that $\diam(P) \leq \Delta \leq O(\eps^{-1}) \diam(P)$. 
    For each $t \in \cL$, let $Z_{t}^1,Z_{t}^2,\dots, Z_{t}^s \sim \cD^t$ be $s$ independent samples of the value $Z$ output from Algorithm \ref{alg:onepassMain} at level $t$, and let $Z_t = \frac{1}{s} \sum_{i=1}^s Z_t^i$. Let $\hat{V}_{t}$ be any value such that $\hat{V}_{t} = (1 \pm \frac{1}{2}) |V_{t}|$, and set 

    \[ R = \frac{4}{\eps} \sum_{t \in \cL}\hat{V}_{t} \cdot t \cdot Z_t \]
Then, setting $s = \tilde{O}( d^4 \Delta^{2 \eps} \log^4(\Delta) \log( n) )$, we have 
\[ \cost(P) \leq  R \leq O\left(\frac{\log\eps^{-1}}{\eps^2}\right)\cost(P) \]
    with probability $1-\poly(1/n)$. Moreover, each iteration of the while loop in Algorithm \ref{alg:onepassMain} returns a valid sample $z$ with probability at least $(1+\eps)^{-2d}$. 
\end{theorem}
\begin{proof}

    By Lemmas \ref{lem:onePassUpperBound} and \ref{lem:onePassLowerBound}, we have

    \[ \left(1-\frac{1}{2}\right)\cost(P) \leq \ex{R} \leq O\left(\frac{\log \eps^{-1}}{\eps^2}\cost(P)\right)\]
    Since $Z_i^t$ is a bounded random variable in $[0,1]$, by Chernoff bounds and a union bound, setting $s = \Omega(\log (n \log(\Delta))/\delta^2)$ we have that $|Z_t - \ex{Z_t}| \leq \delta $ with high probability for all $t \in \cL$. Thus 

    \[ |R - \ex{R}| \leq 2 \delta \sum_{t \in \cL} t |V_t| \leq   2 \Delta^{\eps} \delta \beta^2 \cost(P) \]
    Where the last inequality follows from Lemma \ref{lem:quadtreeLemma}.  
    Setting $\delta = O(\frac{\Delta^{-\eps}}{\delta^2}) =  O(\Delta^{- \eps} d^2 \log^2 \Delta)$ completes the proof. 
\end{proof}


%% file: NewLSH.tex
\subsection{Locality Sensitive Hash Functions with Verifiable Recovery}\label{sec:LSH}
In this section, we prove the existence of a specific type of hash function.

\begin{lemma}\label{lem:newLSH}
Let $\cX = (\R^d, \ell_1)$ be $d$-dimensional space equipped with the $\ell_1$ metric, and fix any maximum radius $\Gamma>0$ and precision parameter $\eps \in (0,1)$. Then for any $t \in (0,\Gamma)$, there is a family $\cH_\eps(t)$ of randomized hash functions $h: \cB_{\cX}(0,\Gamma) \to \{0,1\}^{\tilde{O}(d \log \Gamma)}$  with the following properties:
\begin{enumerate}
    \item For any $p,q \in \cB_{\cX}(0,\Gamma)$, if $h(p) = h(q)$ then we have $\|p-q\|_1 \leq 2 t/\eps$ deterministically. \label{property:1}
    \item For any $p \in \cB_{\cX}(0,\Gamma)$, with probability exactly $(1-\eps)^{d}$ over the draw of $h \sim \cH$, we have $h(p) = h(y)$ simultaneously for all $y \in \cB_{\cX}(p,t)$.  \label{property:2}
    \item Given $h,p$, one can check whether or not the prior condition holds. \label{property:3}
\end{enumerate}
Moreover, with probability at least $1-\Gamma^{-d}$, assuming we evaluate $h$ on a set of at most $\Gamma^d$ points within $\cB(0,\Gamma)$, a single hash function $h \in \cH_\eps(t)$ requires at most 
$R = O(\Gamma^{d} \log^3(\Gamma^d))$ bits of randomness, but can be stored with a pseudo-random generator in at most $\tilde{O}(d \log^3 (\Gamma^d))$ bits-of space such that the above hold up to variational distance $\Gamma^{-100d}$. 
\end{lemma}

\begin{proof}

The hash function is as follows. In what follows, all metric balls are with respect to the $\ell_1$ metric over $\R^d$. We select an infinite sequence independently and uniformly at random points $p_1,p_2,\dots, \sim \cB(0,2\Gamma)$ (we show that the sequence can be bounded later on). Note that we select these points from the continuous uniform distribution over $\cB_\cX(0,\Gamma)$, and will later deal with issues of discretization. 
We then set $h(x) = i$ where $i$ is the smallest index in $[k]$ such that $\|p_i-x\|_1 \leq t/\eps$.

We first claim that $h(x) = O(\Gamma^{d} \log(\Gamma^d))$ for all $x \in \cB_\cX(0,\Gamma)$. To see this, notice that $h(x)$ is fixed the moment a point $p \in \cB(x,t/\eps)$ arrives. Since for any constant $\delta > 0$ we have 
\[ \frac{\textsc{Vol}\left(\cB(x,\delta)\right)}{\textsc{Vol}\left(\cB(0,2\Gamma)\right)} = \left(\frac{\delta}{2 \Gamma}\right)^d \]
Thus, after $ \Omega((2\Gamma/\delta)^{d} \log( \Gamma^d))$ sampled points $p_i$, we expect at least $\Omega(\log(\Gamma^d))$ points to land in the ball $\cB(x,\delta)$. By Chernoff bounds, at least one point will land in the ball with probability at least $1 - (\frac{\delta}{ \Gamma^d})^2$. Thus, if we create a $\delta$-net $N_\delta$ over $\cB(0,2\Gamma)$, which will have at most $O((2\Gamma/\delta)^d)$ points in it, we can union bound over all such points so that $h(x) \leq T$, for some $T = O(\Gamma^{d} \log( \Gamma^d))$ (taking $\Gamma$ larger than some constant), with high probability in $\Gamma^d$. Then given any $p \in \cB(0 ,\Gamma)$, there is a $y \in N_\delta$ with $\|p-y\|_1 \leq \delta$, and so for any $t > 2 \eps \delta$, we will have that $\cB(y,\delta) \subset \cB(p,t/\eps)$, in which case $h(p) \leq T$ also follows.

We now prove that $\cH_\eps(t)$ has the desired properties. First, note that Property \ref{property:1} holds trivially by the construction of the hash function. For Property $\ref{property:2}$, we note that if $h(x) = i$ and $\|p_i-x\|_1 \leq t/\eps - t$, then it follows that $h(x) = h(y)$ simultaneously for all $y \in \cB_{\cX}(x,t)$. In fact, note that the converse is also true: if $\cB(p,t) \subseteq \cB(x,t/\eps)$ then $\|p-x\|_1 \leq t/\eps - t$. 

\[ \frac{\textsc{Vol}\left(\cB(p,t/\eps(1-\eps))\right)}{\textsc{Vol}\left(\cB(p,t/\eps)\right)} = (1-\eps)^d \]
As desired. The final Property \ref{property:3} follows, since given $h,x$, once can simply check if $\|x-p_i\|_1 \leq t/\eps - t$, where $h(x) = i$. Now we deal with issues of discretization. Note that to determine $h(x)$ for all $x \in \cB(0,\Gamma)$, it suffices to determine how many bits are required to check whether or not $\|p-x\|_1 \leq \frac{t}{\eps }$. If after $s \log \Gamma$ bits are generated for each coordinate, where $s = O(\log \Gamma^d)$, we still cannot determine the validity of this test, it must be that $\|p-x\|_1 = \frac{t}{\eps } \pm 2^{-s}$. Since we have 
\begin{align*}
    \textsc{Vol}\left(\cB(p,t/\eps +2^{-s} )\right) - \textsc{Vol}\left(\cB(p,t/\eps -2^{-s} )\right) &\leq \frac{2^d}{d!}( (t/\eps +2^{-s} )^d - (t/\eps -2^{-s} )^d )\\
    & \leq \frac{2^d}{d!} t/\eps ( (1+2^{-s})^d -  (1-2^{-s})^d)) \\
    & \leq \frac{2^d}{d!}  (t/\eps) \cdot6 \cdot 2^{-s} d \\
   & \leq \Gamma^{-10d}
\end{align*}
Where we used that the volume of a ball of radius $r$ in the $\ell_1$ norm is exactly $(2^d/d!) r^d$. 
Thus we can union bound over all $\Gamma^d$ points which we ever will evaluate $h$ on, to ensure that all such points can have their hash determined in at most $O(\log^2 \Gamma^d)$ bits of space. Since we generate at most $T = O(\Gamma^{d} \log( \Gamma^d)$ points, the bound on the total amount of randomness required follows. 

\paragraph{Derandomization.} We now derandomize the hash function. Notice that the only step which needs to be fooled is the second property since the latter two are deterministic. Let $\cU\subset \cB_\cX(0,\Gamma)$ be the universe of at most $\Gamma^d$ points on which we plan to employ our hash function. We will employ the standard technique of Nisan's Psuedorandom generator (PRG) for this purpose~\cite{nisan1990pseudorandom}. To utilize the PRG, we need only show that there exists a tester with space $S= O(d \log^2 \Gamma^d)$ which, given $x \in \cU$, reads the randomness used to generate $p_1,p_2,\dots$ in a stream, and outputs the value of $h(x)$. To see this, the tester can check if $\|p_i - x\|_1 \leq t/\eps$ while only storing the randomness required to generate $p_i$. Since this is at most $O(d \log^2 \Gamma^d)$ bits for a given $p_i$, the proof of the derandomization then follows from Nisan's PRG \cite{nisan1990pseudorandom}.\end{proof}

\begin{definition}
Let $h \sim \cH_\eps(t)$ be a hash function from the above family, and let $x_1,x_2,\dots$ be the ordered set of sampled points that define $h$. For any $p \in \R^d$, write $\tester_t(h,p)$ to be the Boolean function which outputs $1$ if and only if $\|p - x_{h(p)}\|_1 \leq t$. Note that $\mathbb{E}_{h \sim \cH_\eps(t)}[\tester_t(h,p)] = (1+\eps)^{-d}$ for every $p \in \R^d$. Since $t$ is implicitly defined by the hash function $h \sim \cH_\eps(t)$, we drop the subscript when the hash family that $h$ is drawn from is fixed, and write $\tester_t(h,p) = \tester(h,p)$.
\end{definition}

%% file: OnePassUpper.tex
\subsection{Upper Bounding the Modified Estimator} We are now ready to begin upper bounding the cost of our estimator. 
For any point $q \in V_t$, define $z_t^j(q) = \exx{(z,p) \sim \hat{\cD}_j}{ z \; | \; p=q   }$. In other words, $z_t^j(p)$ is the expectation of the $z$-marginal of the distribution $(z,p) \sim \hat{\cD}_j$ conditioned on $p=q$; if this event has a zero probability of occurring, set $z_t^j(q) = 0$. Thus we have $z_t^j(q) \in [0,1]$. 
Notice then that, if $(z,p) \sim \cD^t$ is be the random variable returned by Algorithm \ref{alg:onepassMain}, we have
\begin{equation}
\exx{(z,p)\sim \cD^t}{Z} = \frac{1}{\pi}\sum_{j=-1}^{\tau}\sum_{q \in V_t}\pi_{j,q}z_t^j(q)
\end{equation}
We also define $ z_t(q)  = \sum_{j=-1}^{\tau} z_t^j(q)$ and 
\begin{equation}\label{eqn:xteps}
    x_{t,\eps}^j(p) = \frac{1}{|CC(p,G_t(\cB(p, \frac{1}{\eps^j}t)))|}
\end{equation}
Notice that $x_{t,\eps}^j(p)$ is a reparameterization of $x_t^j(p)$ as defined Section \ref{sec:upper}, where in particular $x_{t,\eps}^j(p) = x_t^{j \log(1/\eps)}(p)$. Since the definition of type-$j$ goes in powers of $1/\epsilon$ in this section, it will be more convenient  to work with $x_{t,\eps}^j(p)$.
We begin with a proposition that demonstrates under which circumstances any point $p \in V_t$ can be returned as the output of Algorithm \ref{alg:onepassMain}.

\begin{proposition}\label{prop:structureBad}
Fix any point $q \in V_t$, and suppose that $\prb{(z,p) \sim\hat{\cD}_{j}}{p=q} > 0$ for some $j \geq 0$. Then we have $|\cB(q,t/\eps^{j} ,V_t)| \leq \beta^2 \Delta^{10 \eps}$, and if $j < \tau$ then we also have $|\cB(q,t/\eps^{j+3}, V_t)| > \beta^2 \Delta^{10\eps}$.
\end{proposition}

\begin{proof}
For the first claim, note that if  $\prb{(z,p) \sim\hat{\cD}_{j}}{p=q} > 0$ for some $j \geq 0$, then it must be the case that $q$ survived both $(h_1,h_2)$, where where $h_1 \sim \cH(t/\eps^j),h_2 \sim \cH(t/\eps^{j+2})$ are independently drawn hash functions. For any $j \geq 0$,  both Algorithms \ref{alg:onepassBad} and \ref{alg:onepassComplete} only return $(z,q)$ if, letting $b=h_2(q),c=h_1(q)$, we have $|h_2^{-1}(b) \cap h_1^{-1}(c)| \leq \beta^2 \Delta^{10 \eps}$. But since $q$ survived $(h_1,h_2)$, we have $\cB(q,t/\eps^j,V_t) \subset h_2^{-1}(b) \cap h_1^{-1}(c)$, so 
\[|\cB(q,t/\eps^j,V_t)| \leq |h_2^{-1}(b) \cap h_1^{-1}(c)|  \leq \beta^2 \Delta^{10 \eps}\]
For the second claim, in Algorithm \ref{alg:onepassBad} (which covers the case of $j < \tau$) we only return $q$ if $|h_2^{-1}(b))| > \beta^2 \Delta^{10\eps}$. But since $h_2^{-1}(b)$ has diameter at most $t/\eps^{j+3}$ by Lemma \ref{lem:newLSH} (using the precision parameter $2\eps$ in the construction of the hash families), it follows that $h_2^{-1}(b) \subset \cB(q,t/\eps^{j+3}, V_t)$, so
\[ |\cB(q,t/\eps^{j+3}, V_t)| \geq |h_2^{-1}(b) |  > \beta^2 \Delta^{10\eps}\]
which completes the proof. 
\end{proof}

We now prove a Lemma, which will allow us to bound the overall expectation of the variable $z$ sampled from Algorithm \ref{alg:onepassMain}. 

\begin{lemma}\label{lem:onePassMain}
Fix any point $q \in V_t$. Then the following four facts are is true:
\begin{enumerate}
    \item For any $j \in \{-1,0,1,\dots,\tau\}$, we have $\prb{(z,p) \sim\hat{\cD}_j}{p=q} \leq (1-2\eps)^{2d} |V_t|^{-1}$.
\item If $p$ is dead,  then $\prb{(z,p) \sim\hat{\cD}_{-1}}{p=q} =(1-2\eps)^{2d} |V_t|^{-1}$. If $p$ is very dead, then $\prb{(z,p) \sim\hat{\cD}_j}{p=q} =0$ for any $j \neq -1$. 
      \item We have $\sum_{j=-1}^{\tau}\prb{(z,p) \sim\hat{\cD}_j}{p=q} \leq 4 (1-2\eps)^{2d} |V_t|^{-1}$, and further that at most $4$ of these summands are non-zero (at most three of which can be for $j \geq 0$). 
      
      \item For each point $q \in V_t$, we have $\prb{(z,p) \sim\hat{\cD}_j}{p=q} \geq (1-2\eps)^{2d} |V_t|^{-1}$ for at least one $j \in \{-1,0,1,\dots,\tau\}$.
 
\end{enumerate}
\end{lemma}
\begin{proof}
We prove each property in order:
\paragraph{Property 1.} 
First note that for any point $q \in V_t$, the algorithm only outputs the point $q$ if $q$ is sampled uniformly from $V_t$ and $q$ survives both $h_1,h_2$, where $h_1 \sim \cH(t_1),h_2 \sim \cH(t_2)$ are independently drawn hash functions for two values $t_1 > t_2$. Now $q$ survives $(h_1,h_2)$  with probability exactly $(1-2\eps)^{2d}$, since each of the two events is independent and occurs with probability $(1-2\eps)^d$ by Lemma \ref{lem:newLSH} (using the precision parameter $2\eps$). Conditioned on this event now (that $p$ survives $(h_1,h_2)$), one can then fix $h_1,h_2$, which fixes the partition of the points into hash buckets $h_1^{-1}(b) \cap h_{2}^{-1}(c)$ for all possible values of $b,c$. Given this partition, each algorithm samples the point $q$ uniformly (via first sampling its hash bucket with probability proportional to its size, and then sampling $q$ from that hash bucket). Thus the probability that $q$ is returned is at most $(1-2\eps)^{2d} |V_t|^{-1}$ as needed. Notice that the probability could be lower, since we may still output $\fail$ even conditioned on these two events, but this completes the first part of the proposition.

\paragraph{Property 2.}
For the second part, suppose $q$ was dead. Note that in all three Algorithms \ref{alg:onepassDead}, \ref{alg:onepassBad}, or \ref{alg:onepassComplete}, whenever $(z,q)$ is output, it is always the case that $q$ survives $(h_1,h_2)$ and $|h_2^{-1}(b) \cap h_1^{-1}(c)|\leq \beta^2 \Delta^{10\eps}$, where $h_1 \sim \cH(t'), h_2 \sim \cH(t'/\eps^2)$ are two independently drawn hash functions for some value $t'$ (depending on which algorithm is used). Since $q$ survives $(h_1,h_2)$, letting $b=h_2(q),c=h_1(q)$, it follows that for this value of $t'$, we have $\cB(q,t',V_t) \subseteq (h_2^{-1}(b) \cap h_1^{-1}(c))$. First, for the case of $j=-1$, the value is $t' = t/\eps$, and this implies that $\cB(q,t/\eps,V_t) \subseteq (h_2^{-1}(b) \cap h_1^{-1}(c))$. Thus
\[ |(h_2^{-1}(b) \cap h_1^{-1}(c))| \geq |\cB(q,t/\eps,V_t)| > \beta^2 \Delta^{10 \eps}\]
Where the last inequality is the definition of $q$ being dead. Thus, if a dead point $q$ is sampled and survives $(h_1,h_2)$ (which occurs with probability $(1-2\eps)^{2d} |V_t|^{-1}$ as already argued), then it will be returned by Algorithm \ref{alg:onepassDead}. Next, we show that a very dead point is not output by any other $j \geq 0$. By Proposition \ref{prop:structureBad}, this only occurs when $|\cB(q,t , V_t)| \leq \beta^2 \Delta^{10 \eps}$ contradicting the fact that $q$ was very dead.

\paragraph{Property 3.} Fix any $\tau \geq j_1> j_2 \geq 0$ such that $j_1 - j_2 > 3$. We claim that $\prb{(z,p) \sim\hat{\cD}_j}{p=q}=0$ for at least one of $j \in \{j_1,j_2\}$, from which the final property will hold (with a factor of $4$ instead of $3$, accounting for the case of $j=-1$). To prove this claim, suppose otherwise.
By Proposition \ref{prop:structureBad}, we know that $|\cB(q,t/\eps^{j_1},V_t)| \leq \beta^2 \Delta^{10 \eps}$, and moreover that $|\cB(q,t/\eps^{j_2 + 3}, V_t)| > \beta^2 \Delta^{10 \eps}$, which implies that $|\cB(q,t/\eps^{j_2 + 3}, V_t)| > |\cB(q,t/\eps^{j_1},V_t)|$.
But since $j_1 - j_2 >  3$, we have $\cB(q,t/\eps^{j_2+3}, V_t ) \subset \cB(q, t/\eps^{j_1}, V_t)$, so this would be impossible, yielding a contradiction.

\paragraph{Property 4.}
To prove the fourth and final property, note that for all three algorithms, the point $q$ is sampled and survives two separate hash functions $(h_1,h_2)$ with probability exactly $(1-2\eps)^{2d}|V_t|^{-1}$. So we need only show that there exists at least one level $j$ such that, conditioned on these two events, $(z,q)$ is output from $\hat{D}_j$. Note that if $q$ is dead, this is already demonstrated by the second property of the Lemma. So next, suppose that $q$ is of type $\ell$ for some $\ell \in \{1,\dots,\tau\}$, we show that $q$ is returned from $\hat{\cD}_{j}$ where $j = \ell-1$ with probability at least $(1-2\eps)^{2d}|V_t|^{-1}$. With this probability $q$ is sampled and survives $(h_1,h_2)$, where $h_1 \sim \cH(t/\eps^{j-1}), h_2 \sim \cH(t/\eps^{j+1})$. Let $b = h_2(q), c=h_1(q)$. It suffices to show that 
$|h_2^{-1}(b))| > \beta^2 \Delta^{10\eps}$ and $|h_2^{-1}(b) \cap h_1^{-1}(c)| \leq \beta^2 \Delta^{10\eps}$. For the first part, notice that $q$ surviving $h_2$ implies that $\cB(q, t/\eps^{j+1},V_t) \subseteq h_2^{-1}(b))$, thus
\[|h_2^{-1}(b))| \geq |\cB(q, t/\eps^{j+1},V_t) | > \beta^2 \Delta^{10 \eps} \]
Where the last inequality is because $q$ is of type $\ell = j+1 $. Next, since $h_1^{-1}(c)$ has diameter at most $t/\eps^{j}$, we have $h_2^{-1}(b) \cap h_1^{-1}(c) \subseteq \cB(q, t/\eps^{j},V_t) $, so
\[|h_2^{-1}(b) \cap h_1^{-1}(c)|  \leq |\cB(q, t/\eps^{j},V_t) | \leq \beta^2 \Delta^{10 \eps} \]
Which proves the claim. 
By Fact \ref{fact:alltypes}, the only case of point $q$ left to consider is the case where $q$ is nearly complete, for which we show that $q$ is returned by Algorithm \ref{alg:onepassComplete} (when $j=\tau$). In this case, given that $q$ is sampled and survives $(h_1,h_2)$, where $h_1 \sim \cH(\frac{t}{\eps^{j}}), h_2 \sim \cH(\frac{t}{\eps^{j+2}})$,  we need only show that $|h_2^{-1}(b) \cap h_1^{-1}(c)| \leq \beta^2 \Delta^{10 \eps}$. Since $q$ is nearly complete, we have $|\cB(q,\frac{1}{\eps^{\tau+1} }t, V_t)| \leq \beta^2 \Delta^{10\eps}$. Since $h_1^{-1}(c)$ has diameter at most $\frac{1}{\eps^{\tau+1} }t$, it follows that $(h_2^{-1}(b) \cap h_1^{-1}) \subseteq h_2^{-1}(b)  \subseteq \cB(q,\frac{1}{\eps^{\tau+1} }t , V_t)$, so
\[|h_2^{-1}(b) \cap h_1^{-1}(c)| \leq \left|\cB(q,\frac{1}{\eps^{\tau+1} }t , V_t)\right| \leq \beta^2 \Delta^{10 \eps}\] 
which completes the proof. 
\end{proof}

\begin{corollary}\label{cor:upperBound}
Let $Z$ be the random variable returned by Algorithm \ref{alg:onepassMain}. Then we have
\[ \ex{Z} \leq \frac{1}{|V_t|} \sum_{q \in V_t}z_t(q) \] 
\end{corollary}
\begin{proof} 
Applying Lemma \ref{lem:onePassMain}, we have that $\pi_{j,q} \leq (\tau+2)^{-1}(1-\eps)^{2d}|V_t|^{-1}$ for all $j,q$, and moreover that $\pi \geq (\tau+2)^{-1}(1- 2\eps)^d  $. It follows that 
\begin{align}
    \ex{Z} = \frac{1}{\pi}\sum_{j=-1}^{\tau}\sum_{q \in V_t}\pi_{j,q}z_t^j(q) \leq\frac{1}{|V_t|}\sum_{j=-1}^{\tau} \sum_{q \in V_t}z_t^j(q) = \frac{1}{|V_t|}\sum_{q \in V_t}z_t(q) 
\end{align}
as desired. 
\end{proof}
\begin{proposition}\label{prop:someBounds}
Fix any point $p \in V_t$. Then the following bounds apply:
\begin{enumerate}
    \item If $p$ is very dead, then $z_t(p) = 0$. 
    \item If $p$ is of type $\ell$ for any $0 \leq \ell\leq \tau$, then $z_t(p) \leq  3 x_{t,\eps}^{\max\{ \ell-3, 0\}}(p)$, 
    and if $p$ is nearly complete, we have $ z_t(p)  \leq  3 x_{t,\eps}^{\tau  -3}(p)$.
    \item If $p$ is complete, then we have $ z_t(p) \leq  3 x_{t}(p)$.
\end{enumerate}
\end{proposition}
\begin{proof}
    First note that if $p$ is very dead, we have $z_t^j(p)=0$ for all $j \geq 0$ by Lemma \ref{lem:onePassMain}, and since the output $z$ of Algorithm \ref{alg:onepassDead} is always either $\fail$ or $0$, we have $z_t^{-1}=0$ as well. Thus, $z_t^{-1}(p)=0$  for all $p \in V_t$, which completes the first claim. 
    
    Now suppose $z_t^j(p) > 0$ for any point $p$. This implies that $q$ survived both $(h_1,h_2)$, where where $h_1 \sim \cH(t/\eps^j),h_2 \sim \cH(t/\eps^{j+2})$ are independently drawn hash functions. For a fixed $h_1,h_2$ such that $z_t^j(p) >0$ conditioned on $h_1,h_2$, we have that the output of the algorithm is $|CC(p,G_t(h_1^{-1}(c) \cap h_2^{-1}(b)))|^{-1}$, where $c = h_1(p), b=h_2(p)$. Since $p$ survived $h_1,h_2$ to be output, it follows that $\cB(p,t/\eps^j, V_t) \subset (h_1^{-1}(c) \cap h_2^{-1}(b)))$, thus $z_t^j(p) \leq x_{t,\eps}^j(p)$. 
    
    We now claim that if $p$ is of type $\ell$, then $z_t^{j}(p) = 0$ for every $j < \max\{0,\ell-3\}$ (notice that $z_{t,\eps}^{-1}(p)$ is always zero for all $p \in V_t$, if we didn't output $\fail$). By Proposition \ref{prop:structureBad}, if $z_t^{j}(p) >0$ for some $0 \leq j < \tau$, then $|\cB(p,t/\eps^{j+3}, V_t)| > \beta^2 \Delta^{10 \eps}$. But since $p$ is of type $\ell$, we have $|\cB(p,t/\eps^{j}, V_t)| \leq \beta^2 \Delta^{10 \eps}$, from which $\ell-j \leq 3$ follows. Since  $z_t^{-1}(p)=0$  for all $p \in V_t$, and $x_{t,\eps}^a \leq x_{t,\eps}^b$ whenever $b \geq a$, along with the fact  from Lemma \ref{lem:onePassMain} that at most three of the $z_t^{j}(p)'s$ can be non-zero for $j\ge 0$, it follows that  $z_t(p) = \sum_{j=-1}^{\log \log \Delta} z_t^j(p) \leq  3 x_{t,\eps}^{\max\{ j-3, 0\}}(p)$. Notice that the same bound applies to nearly-complete points, using that if $p$ is nearly complete, then we also have $|\cB(p,t/\eps^{j}, V_t)| \leq \beta^2 \Delta^{10 \eps}$ for $j = \tau$. 

    Finally, assume $p$ is $\ell$-complete for some $\ell \geq 0$. Then by definition, it is of type $j=\ell+3$, so similarly to the previous item $z_t(p) = \sum_{j=-1}^{\tau} z_t^j(p) \leq  3 x_{t,\eps}^{\ell}(p)$. But since $p$ is $\ell$-complete, $CC(p,G_t) \subseteq \cB(p,t/\eps^{\ell})$, and thus $x_{t,\eps}^{\ell}(p) = \frac{1}{|CC(p,G_t(\cB(p,t/\eps^{\ell})))| } = x_{t}(p)$, from which the claim follows.

\end{proof}

Now define $Y_t^j \subset V_t$ be the set of points $p$ which are of type $j$ at level $t$, such that $p$ is not complete. Let $B_{t}^j$ be the set of leaders as defined in Section \ref{sec:structure} (for any level $t$). Let $Y_t^{near}$ be the set of nearly complete points which are not complete. We first prove a reparameterization of Proposition \ref{prop:5}

\begin{claim}\label{claim:prop5}
Fix any $j \geq 1$. Let $p\in Y_t^j$ and let $q\in V_t$ be a point such that $\|p-q\|_1\le (1/\eps)^{j-1} t $. 
Then $q$ is $j'$-bad (Definition \ref{def:bad}) at level $t$ for some $j'$ that satisfies $|j\log \eps^{-1}-j'|\le \log \eps^{-1}$.
\end{claim}
\begin{proof}First assume $j \geq 1$. 
Setting $\eta = 1/\eps$, we have
\[ \cB(p, \eta^{j+1} t, V_t) \subseteq \cB(q, (\eta^{j+1} + \eta^{j-1}) t,V_t) \subseteq  \cB(q, \eta^{j+2} t,V_t)\]
and thus, $|\cB(q,\eta^{j+2}t,V_t)|\ge \beta^2\Delta^{10\eps}$.
On the other hand, we have that 
\[\cB(q, \eta^{j-1} t,V_t) \subseteq \cB(p,  (\eta^{j-1}+\eta^{j-1})  t,V_t) \subseteq \cB(p, \eta^{j} t, V_t) \] 
and thus, $|\cB(q,\eta^{j-1}t,V_t)|<\beta^2 \Delta^{10\eps}$. It follows that $q$ is at least $(j-1)\log \eta $-bad, and at most $O(j+1)\log \eta$ bad.
\end{proof}


\begin{lemma}\label{lem:moreBounds}
Fix any level $t\in \calL$ and $0\le j<\tau$. Define the index set $\cO_j =[\min\{0,j-\log \eps^{-1}\} : j + \log\eps^{-1}]$. 
Then for each $\ell \in \cO_j$, there is a set $N_t^{\ell}\subseteq V_t$ which consists only of points which are $\ell$-bad at level $t$ (Definition \ref{def:bad}) with pairwise distance at least $t/2$, such that the following holds:

\[\sum_{p\in Y_{t}^j}z_{t}(p)    \leq 240 \cdot \eps^{\max\{0,j-3\}} \cdot \sum_{\ell \in \cO_j} \left|N_t^{\ell}\right|\]
\end{lemma}

\begin{proof}
By Proposition \ref{prop:someBounds}, we have
\[   \sum_{p \in Y_t^j }  z_{t}(p)  \leq \sum_{p \in Y_t^j } 3 x_{t,\eps}^{\max\{0,j-3\}}(p)  \]
So it will suffice to bound the left hand side above. 
We first deal with the special case of $j \leq 3$. Define $\tilde{G}_{t/2} = (V_t, \tilde{E}_{t/2})$ be the $t/2$ threshold graph on the vertices in $V_t$. In other words, $(x,y) \in  \tilde{E}_{t/2}$ if and only if $\|x-y\|_1 \leq t/2$ (notice that $\tilde{G}_{t/2} = G_{t/2}$ whenever $t,t/2$ are in the same block $\cQ_i$). Let $\cI \subseteq  Y_t^j$ be any maximal independent set of the points in $Y_t^j$ in the graph $\tilde{G}_{t/2}$, and suppose $|\cI| = k$. Consider an arbitrary partition $C_1,\dots,C_k$ of $Y_t^j$, where each $C_i$ is associated to a unique $y_i \in \cI$, and for any $y \in Y_t^j \setminus \cI$, we add $y$ to an arbitrary $C_i$ for which $\|y_i - y\|_1 \leq t/2$ (such a $y_i \in \cI$ exists because $\cI$ is a maximal independent set). By Fact \ref{fact:alltypes}, we know that each $y_i \in Y_t^j$ is $j'$ bad for some $j \log \eps^{-1} \leq j' \leq (j+1) \log \eps^{-1}$, so we can add each point in $\cI$ to a set $N_t^\ell$ for some $\ell \in \cO_j$, resulting in $\sum_{\ell \in \cO_j} \|N_t^{\ell}| \geq k$. Moreover, notice that the points in $\cI$ are pairwise distance at least $t/2$ since they form an independent set in $\tilde{G}_{t/2}$.
It will suffice to then show that 
\[\sum_{p \in Y_t^j } x_{t,\eps}^0(p) \leq k\]
To see this, notice that each $C_i$ has diameter at most $t$, thus for any point $p \in C_i$ we have $C_{i} \subseteq \CC(p,G_t(\cB(p, t,V_t)))$. Therefore, $x_{t,\eps}^0(p) \leq 1/|C_i|$ for all $p \in C_i$, from which the claim follows.

We now consider the case of $j \geq 4$. 
Let $\cC_1,\ldots,\cC_\ell$ be the connected components of $G_t$ that contain at least one point in $Y_t^j$.
 We first claim that each $\cC_i$ has diameter at least $\eps^{j-3}$.  To see this, note that if $\cC_i$ contains a single point $p \in Y_{i}^j$ which is of type $j$ and which is not complete, this implies $CC(p,G_t(\cB(p,\frac{1}{\eps^{j-3}} t))) \neq CC(p,G_t)$, so there exists a point $q \in CC(p,G_t)$ with $\|p-q\|_1 \geq \frac{1}{\eps^{j-3}} t$.  

For each $\cC_i$ we can now partition $\cC$ into $C_{i,1}, \dots,C_{i,k_i}$  via Lemma \ref{lem:graphDecomp}, setting the diameter bound on the $C_{i,j}$'s to be $\lambda = \frac{1}{2 \eps^{j-3}} \geq 2$ (taking $\eps \leq 1/4$), but only keep $C_{i,j}$'s  which overlaps with $Y_t^j$. Thus 

\[\sum_{p \in Y_t^j }  x_{t,\eps}^{j-3}(p)  \le  \sum_{i=1}^\ell \sum_{j=1}^{k_i}\sum_{p \in C_{i,j}}   x_{t,\eps}^{j-3}(p) \]

We claim that $\sum_{p \in C_{i,j}} x_{t,\eps}^{j-3}(p) \leq 1$. To see this, since the diameter of $C_{i,j}$ in the shortest path metric is at most $\eps^{-(j-3)}$, it follows that  $C_{i,j} \subseteq \CC(p,G_t(\cB(p,(\frac{1}{\eps})^{j-3} t,V_t)))$. Thus the $\ell_1$ diameter of the set of points $C_{i,j}\subset \R^d$ is at most $\eps^{-(j-3)} t$. Thus $x_{t,\eps}^{j-3}(p) \leq  {1}/{|C_{i,j}|}$, from which the claim follows. Using this, we have  
\begin{equation}\label{eqn:YBound2}
   \sum_{p \in Y_t^j } x_{t,\eps}^{j-3}(p)  \leq \sum_{i=[\ell]} k_i.
\end{equation}

Now let $s= \sum_{i\in [\ell]} \frac{k_i}{24\eps^{j-3}} $ and 
$X = \{x_1,\dots,x_{s}\}$ be the independent set of $G_t$ obtained from Lemma \ref{lem:graphDecomp} 
 (by taking the union of independent sets from each connected component $\cC_i$). Given that they form an independent set of $G_t$ we have that their pairwise distance is at least $t$. Furthermore, each $x\in X$ is in one of the $C_{i,j}$'s and is thus, satisfies
  $\|x-y\|_1\le 
  \frac{1}{2 \eps^{j-3}} t$ for some point $y\in Y_t^j$.
It follows from Claim \ref{claim:prop5} that that every $x\in X$ is $\ell$-bad at level $t$
for some $\ell\in \cO_j$. Partitioning $X$ into $N_t^{\ell}$ for each $\ell \in \cO_j$ accordingly and using (\ref{eqn:YBound2}) finishes the proof.\end{proof}

We are now ready to prove the main upper bound on the cost of our estimator.

\begin{lemma}\label{lem:onePassUpperBound}
    Let $P \subset [\cord]^d$ be a dataset with the $\ell_1$ metric. Then letting $\cD^t$ be the distribution of the output of Algorithm \ref{alg:onepassMain} when run on the graph $G_t$, we have:
    \[ \sum_{t \in \cL} 2^i  \cdot |V_{t}| \cdot \exx{Z_{t} \sim \cD^{t}}{Z_{t}} \leq O\left(\frac{\log \eps^{-1}}{\eps} \cost(P) \right) \]
\end{lemma}
\begin{proof}
By Corollary \ref{cor:upperBound}, we have

 \[ \sum_{t \in \cL} t \cdot |V_{t}| \cdot\exx{Z_{t} \sim \cD^{t}}{Z_{t}} \leq  \sum_{i=0}^{\log(\Delta) }  t  \sum_{q \in V_{t}}z_t(q) \] 
 
We first consider the case of complete points. Let $\cK_t$ be the set of complete points at level $t$. By proposition \ref{prop:someBounds}, the total cost of all such points is at most 
\[\sum_{t \in \cL} t \sum_{p \in \cK_t} z_t(q) \leq  \sum_{t \in \cL} t \sum_{p \in \cK_t}3 x_t(p) \leq  \sum_{t \in \cL} t \sum_{p \in V_t }3 x_t(p) \leq   O(\eps^{-1} \cost(P))\]
where the last inequality follows from Proposition \ref{prop:xConstFactor}.

Now it suffices to upper bound $z_t(q)$ for the following cases: \textbf{1) } $p \in Y_t^j$ for some $0 \leq j \leq \tau$, \textbf{2)} $p$ is nearly complete, and 
\textbf{3)} $p$ is complete. We will do this seperately for for each block of indices $\cQ_l \subset \cL$ (as defined in Section \ref{sec:estimatorMain}). Specifically, for $l \in [1/\eps]$ we first bound the above sum for $t \in \cQ_l$ by $O(\log(1/\eps) \cdot \cost(P))$ for points in the above two cases, from which the final bound will follow (using that we can assume $\Delta = \poly(n)$).

Firstly, for the points $p \in Y_{2^i}^j$. 
By Lemma \ref{lem:moreBounds}, for each  $0 \leq j \leq \tau$, we obtain a set $\{N_t^{\ell}(j)\}_{\ell \in \cO_j}$ of points which are all $\ell$-bad and pair-wise distance $t/2$ from each other. Thus

\[ \sum_{t \in \cQ_l} \sum_{j = 0}^\tau \sum_{p\in Y_{t}^j} t z_{t}(p)    \leq 240 \sum_{t \in \cQ_l} t\sum_{j = 0}^\tau  \eps^{\max\{0,j-3\}} \cdot \sum_{\ell \in \cO_j}  \left|N_t^{\ell}\right|\]
Notice that in the above sum, since $\cO_j$ contains only indices $\ell$ which are distance at most $\log \eps^{-1}$ from $j$, it follows that for each $\ell \geq 0$, there are at most$O(\log \eps^{-1})$ sets $N_t^{\ell}$ counted in the above sum, and each time the set is counted it appears with a factor of at at most $\min\{1, \eps^{-j-3 -\log \eps^{-1}} \}$. 
We can then apply  Lemma \ref{lem:structural} a total of $O(\log \eps^{-1})$ times, breaking $\cQ_i$ into even and odd powers of $2$, to obtain 
\[ \sum_{t \in \cQ_l} \sum_{j = 0}^\tau \sum_{p\in Y_{t}^j}t z_{t}(p)  \leq O(\log \eps^{-1}) \cost(P) \] 
as needed.

 Finally, we bound the cost of the nearly complete points $p$  over all levels. By Proposition  \ref{prop:someBounds}, letting $\cK_t'$ be the set of nearly complete points at level $t$,  we have 
\begin{align*}
    \sum_{t \in \cL}   \sum_{q \in \cK_t'} tz_t^\ell(q) &\leq  \sum_{t \in \cL}\sum_{q \in V_t} 3 t x_{t,\eps}^{\tau - 3}(p) \\
    &  \leq O(\cost(P)) + \sum_{t \in \cQ_l} \sum_{t \in V_t} t x_{t}(p)  \leq  O(\cost(P)))
\end{align*}
where the second to last last inequality follows from Lemma \ref{hehelemma1}, using that $(1/\eps)^{\tau-3} = \log \Delta$, and the last inequality follows from Proposition \ref{prop:xConstFactor}. 
\end{proof}

%% file: OnePassLower.tex
\subsection{Lower Bounding the Modified Estimator}
 \paragraph{Setup and Notation. } Fix any levels $t \in \cL$.
 We first define a function $\rho_t: \cup_{\substack{ t' \in \cL \\ t' < t}} V_{t'} \to V_t$ in the following way: first, if $t,t' \in \cQ_i$ are in the same block, we have $V_{t} = V_{t'} = V_T$ where $T = \Delta^\eps \lfloor t / \Delta^\eps \rfloor$, and then the mapping is just the identity. Otherwise, for $p \in V_{t'}$
we set $\rho_{t}(p)= f_t(f^{-1}_{t'}(p))$ which is well defined because the construction of the $V_t$'s satisfies $f_{t'}(p) = f_{t'}(q)$ implies that $f_t(p) = f_t(q)$ for $t'\le t$.  We abuse notation and write $f_t(S) = \{f_t(p) | p \in S\}$ for any set $S \subset V_{t'}$ and $t \geq t' \in \cL$.

Note that this allows us to extend the definition of $x_t(p)$ to holding for all $p \in \cup_{t' \leq t/2}V_{t'}$, by simply defining $x_t(p) = x_t(\rho_t(p))$. 
For any set of connected components $\cC$, we abuse notation and write $q \in \cC$ to denote that $q \in C$ for some $C \in \cC$.  For a connected component $C$ in $G_t$, call $C$ \textit{dead} if it contains at least one point $p \in C$ which is dead at level $t$. If $C$ contains only alive points, call $C$ alive. 
Next, for any $t \in \cL$, we define $\Gamma_t \subset V_t$ to be the set of connected component in $G_t$ which are alive (i.e., contain only alive points). We will show that the set of dead points are precisely those for which we can possibly output a value in our estimator which is less than the true cost. Thus, our main goal will be to bound the contribution of the dead points to the total MST cost.

 We now will prove a sequence of structural facts, which will allow us to prove the main result of this section, which is Lemma \ref{lem:onePassLowerBound}. We begin with the following claim.
 
\begin{claim}\label{claim:LowerBound1}
    Fix any $t'$ satisfying $2t' \leq t \in \cL$. If $p,q \in V_{t'}$ are in the same connected component in $G_{t'}$, then $\rho_t(p),\rho_t(q)$ are in the same connected component in $G_t$. 
\end{claim}
\begin{proof}
There exists a set of points $p_1,p_2,\dots,p_l$ such that $f_{t'}(p_1) = p,f_{t'}(p_l) = q$, and $f_{t'}(p_1)f_{t'}(p_2),\dots,\allowbreak f_{t'}(p_l)$ is a path in $G_{t'}$, so that $ \|f_{t'}(p_i) - f_{t'}(p_{i+1}) \|_1 \leq t'$. Since $\|f_t(x) - x\|_1 \leq t/\beta$ for all $x \in P$, by the definition of the mapping $\rho_t$ it follows that $\|f_t(p_i) - f_t(p_{i+1})\|_1 \leq \|f_{t'}(p_i) - f_{t'}(p_{i+1}) \|_1 + 2/\beta \cdot(t + t') \leq t$ since $t \geq 2t'$. Thus $f_{t}(p_1)f_{t}(p_2),\dots,f_{t}(p_l)$ is a path in $G_{t}$, and the claim follows.     
\end{proof}
The above claim allows us to demonstrate that, for $t > t'$ and any $S \subset V_{t'}$ the number of connected components in $S$ is at least the number of connected components in $\rho_t(S)$. We also define $\kappa_t(S)$ to be the set of all connected components in $G_t$ which contain at least one point $\rho_t(p)$ for $p \in S$.

We now show that every point $p$ which is not dead has some non-trivial chance to be sampled and returned from the algorithm, and when it is it's contribution to the estimator is at least $x_t(p)$.

\begin{fact}\label{fact:LowerBound1}
Fix any level $t \in \cL$ and point $p \in V_t$. 
    If $p$ is alive, then there exists an integer $j \in [0 \; : \; \tau]$ such that $\prb{(z,p) \sim\hat{\cD}_{j}}{p=q} =(1-2\eps)^{2d} |V_t|^{-1}$, and such that $z_t^j(p) \geq x_{t}(p)$. 
\end{fact}
\begin{proof}
If $p$ is alive, then by definition $|\cB(p,t/\eps, V_t)| \leq \beta^2 \Delta^{10 \eps}$. First suppose that $p$ is of type $\ell$ for some $\ell \in [\tau]$. We claim that $j = \ell-1  \in [0 \; : \; \tau - 1]$ is the desired level, where we sample hash functions: $h_1 \sim \cH(\frac{t}{\eps^{\ell-1}}), h_2 \sim \cH(\frac{t}{\eps^{\ell+1}})$. In this case, $p$ is sampled and survives $h_1,h_2$ with probability $(1-2\eps)^{2d} |V_t|^{-1}$ in Algorithm \ref{alg:onepassBad}. Let $b = h_1(p), c= h_2(p)$, and note that since $h_1^{-1}(b)$ has diameter at most $\frac{t}{\eps^{\ell}}$, it follows that $h_1^{-1}(b) \cap h_2^{-1}(c) \subseteq h_1^{-1}(b) \subseteq \cB(p,\frac{t}{\eps^{\ell}},V_t)$, and since $p$ is of type $\ell$, we have $|\cB(p,\frac{t}{\eps^{\ell}},V_t)| \leq \beta^2 \Delta^{10 \eps}$, so it follows that $|h_2^{-1}(c) \cap h_1^{-1}(b)| \leq \beta^2 \Delta^{10\eps}$. Next, since $p$ survived $(h_1,h_2)$, we have 
\[ |h_2^{-1}(c)| \geq |\cB(p,\frac{t}{\eps^{\ell+1}},V_t)| \geq \beta^2 \Delta^{10 \eps} \]
    It follows that $p$ will be returned in Algorithm \ref{alg:onepassBad} whenever it is sampeld and survives $(h_1,h_2)$, which completes the claim for points of type $\ell$. Lastly, if $p$ is nearly-complete, when $j=\tau$, we sample sample hash functions: $h_1 \sim \cH(\frac{t}{\eps^{j}}), h_2 \sim \cH(\frac{t}{\eps^{j+2}})$ in Algorithm \ref{alg:onepassComplete} , and return $p$ if it is sampled, survives $(h_1,h_2)$, and $|h_1^{-1}(b) \cap h_2^{-1}(c)| \leq \beta^2 \Delta^{10 \eps}$. But since $h_1^{-1}(b)$ has diameter at most $t/\eps^{\tau + 1}$, we have 
    \[ |h_1^{-1}(b) \cap h_2^{-1}(c)| \leq |h_1^{-1}(b) | \leq   |\cB(p,t/\eps^{\tau + 1}, V_t)| \leq \beta^2 \Delta^{10 \eps} \]
    where the last inequality is the definition of nearly-complete. Note that this captures all cases by Fact \ref{fact:alltypes}, and completes the proof of the first claim. For the second claim, note that whenever $z_t^j(p) \neq 0$, it holds that $z_t^j(p) = 1/|S|$, where $S$ is some subset of the connected component in $G_t$ containing $p$, from which $z_t^j(p) \geq x_{t}(p)$ follows. 
\end{proof}

\begin{fact}\label{fact:4LB}
   Let $Z_t \sim \cD^t$ be the random variable returned by Algorithm \ref{alg:onepassMain} at level $t$. Then we have
\[ \ex{Z_{t}} \geq \frac{1}{4|V_t|} \sum_{q \in \Gamma_t} x_t(q) \]
\end{fact}
\begin{proof}
    As in Corollary \ref{cor:upperBound}, let $\pi_{j,q} = (\tau+2)^{-1}\prb{(z,p) \sim \hat{\cD}_j}{p=q}$, and let $\pi = \sum_{j=-1}^{\tau}\sum_{q \in V_t}\pi_{j,q}$. Applying Lemma \ref{lem:onePassMain}, we have that 
    \[ \sum_{j=-1}^{\tau} \pi_{j,q} \leq 4\frac{(1- 2\eps)^d}{\tau + 2} |V_t|^{-1}\]
    for any $q \in V_t$, from which $\pi \leq 4 (1- 2\eps)^d  (\tau+2)^{-1}$ follows. By Fact \ref{fact:LowerBound1}, we know that for every $q \in \Gamma_t$ the point $q$ is not dead, so we have $\sum_{j=-1}^{\tau}\sum_{q \in \Gamma}\pi_{j,q}z_t^j(q) \geq (\tau + 2)^{-1} (1- 2\eps)^d |V_t|^{-1} x_t(p)$
    \begin{align*}
    \ex{Z_{t}}& = \frac{1}{\pi}\sum_{j=-1}^{\tau}\sum_{q \in V_t}\pi_{j,q}z_t^j(q)   \geq \frac{1}{4|V_t|}\sum_{q \in \Gamma_t}x_t(q)
    \end{align*}
    which completes the proof. 
\end{proof}

\begin{fact}\label{fact:LB6}
    Fix any two levels $t \geq 2t'$ in $\cL$, and let $\cC$ be a subset of connected components in $G_{t'}$, and recall that $\kappa_t(\cC)$ is the subset of connected components which contain at least one point $\rho_t(x)$ for where $x  \in C$ for some connected component $C \in \cC$. Then 
    \[  \sum_{q \in \cC} x_{t'}(q) \geq \sum_{p \in \kappa_t(\cC)}  x_{t}(p)\]
    
\end{fact}
\begin{proof}
We abuse notation and write $q\in \cC$ to denote that $q \in C$ for some $c \in \cC$.
    We have $\sum_{q\in \cC }x_{t'}(q) = |\cC|$. Moreover, for any two points $x,y \in C$, for the same $C \in \cC$, we have that $\rho_t(x),\rho_t(y)$ are in the same connected component in $G_{t}$ by Claim \ref{claim:LowerBound1}. Thus $\sum_{p \in \kappa_t(\cC)} x_{t}(q) =|  \kappa_t(\cC)| \leq |\cC|$, which completes the proof. 
\end{proof}
A corollary of the above is that $\sum_{q \in V_{t'}} x_{t'}(q) \leq \sum_{q \in V_t} x_t(q)$ for any $t' \leq t$. 
We now prove the key fact which will be needed for our lower bound on the value of our estimator. 

\begin{fact}\label{fact:5LB}
 We have the following:
 \[ \sum_{t \in \cL} t \sum_{q \in V_{2t/\eps} \setminus \; \kappa_{2t/\eps}(\Gamma_t)} x_{2t/\eps}(q) \leq 2 \Delta^{-9\eps} (\log \Delta +1) \cost(P) \]
\end{fact}
\begin{proof}
Fix any $t \in \cL$.  
Note that $\kappa_{2t/\eps}(\Gamma_t)$ is the set of connected components $C$ in $G_{2t/\eps}$ such that there is at least one connected component $C' \in \Gamma_t$, with $\rho_t(C') \subset C$ which is alive.  It follows that $V_{2t/\eps} \setminus \; \kappa_{2t/\eps}(\Gamma_t)$ is the set of (points in) connected components $C$ in $G_{2t/\eps}$ such that for every connected component $C'$ in $G_t$ with $\rho_t(C') \subseteq C$, the component $C'$ is dead.  Let $k_t$ denote the number of such connected components in $V_{2t/\eps} \setminus \; \kappa_{2t/\eps}(\Gamma_t)$. We have
\[ \sum_{q \in V_{2t/\eps} \setminus \; \kappa_{2t/\eps}(\Gamma_t)} x_{t2/\eps}(q)  = k_t \]
On the other hand, notice that each of the $k_t$ connected components in $V_{2t/\eps} \setminus \; \kappa_{2t/\eps}(\Gamma_t)$ contain at least one point $\rho_{2t/\eps}(p)$ such that $p \in V_t$ is dead. For each $i \in [k_t]$ fix an arbitrary such dead point $p_i \in V_t$ such that $\rho_{2t/\eps}(p)$ is in the $i$-th connected component of $V_{2t/\eps} \setminus \; \kappa_{2t/\eps}(\Gamma_t)$. By the definition of dead,  we have $|\cB(p_i, t/\eps, V_t)| \geq \beta^2 \Delta^{10 \eps}$. Moreover, using that $\|f_t(p)-p\|_1 \leq t/\beta$ for any $t \in \cL$ and $p \in V_t$, we have
\[\|p_i - p_j\|_1 \geq \|\rho_{2t/\eps}(p_i) - \rho_{2t/\eps}(p_j)\|_1 - \frac{2}{\beta }(t/\eps + t) \geq  2t/\eps(1-2/\beta)\] where the last inequality follows from the fact that $\rho_{2t/\eps}(p_i)$ and $\rho_{2t/\eps}(p_j)$ are in seperate connected components in $G_{2t/\eps}$. It follows that $\cB(p_i, t/\eps, V_t)\cap \cB(p_j, t/\eps, V_t) = \emptyset$ for all $i,j \in [k_t]$, from which iut follows that 
\[|V_t|  = |V_T| \geq \Delta^{10 \eps} \beta^2 k_t \]
Where $T = \Delta^\eps \lfloor t /\Delta^\eps \rfloor$. Now by Lemma \ref{lem:quadtreeLemma}, we have
\[ \beta \cost(P)  \geq \frac{T}{\beta} (|V_T|-1) \geq \frac{T}{\beta} |V_T|/2 
 \geq \Delta^{9\eps} t \beta  k_t /2 \]
From which it follows that $  \sum_{q \in V_{2t/\eps} \setminus \; \kappa_{2t/\eps}(\Gamma_t)} t x_{2t/\eps}(q) \leq 2 \cost(P)/ \Delta^{9 \eps}$. Summing over all $\log(\Delta)+1$ values of $t \in \cL$ completes the proof. 
\end{proof}

We are now ready to prove the main lemma of the section. 

\begin{lemma}\label{lem:onePassLowerBound}
    Let $P \subset [\cord]^d$ be a dataset with the $\ell_1$ metric with $\min_{x,y \in P} \|x-y\|_1 \geq 100/\eps^2$. Then letting $\cD^t$ be the distribution of the output of Algorithm \ref{alg:onepassMain} when run on the graph $G_t$, we have:
    \[ \sum_{t \in \cL} t  \cdot |V_{t}| \cdot \exx{Z_{t} \sim \cD^{t}}{Z_{t}} \geq \frac{\eps}{4}\cost(P) \]
\end{lemma}
\begin{proof}
By Fact \ref{fact:LB6}, we have 
$\sum_{q \in \Gamma_t}x_{t}(q) \geq  \sum_{q \in \rho_{2t/\eps}(\Gamma_t)} x_{2t/\eps}(q)$.
Applying this inequality, along with Facts \ref{fact:4LB} and \ref{fact:5LB}, we have:
\begin{align*}
   \sum_{t \in \cL} t  \cdot |V_{t}| \cdot \exx{Z_{t} \sim \cD^{t}}{Z_{t}} &\geq \frac{1}{4}\sum_{t \in \cL} t   \sum_{q \in \Gamma_t} x_t(q) \\
     &\geq \frac{1}{4}\sum_{t \in \cL} t \sum_{q \in \rho_{2t/\eps}(\Gamma_t)} x_{2t/\eps}(q) \\
     & =\frac{1}{4}\left(\sum_{t \in \cL} t \sum_{q \in V_{2t/\eps}} x_{2t/\eps}(q) - \sum_{i=0}^{\log(\Delta) } t
 \sum_{q \in V_{2t/\eps}\setminus \rho_{2t/\eps}(\Gamma_t)}x_{2t/\eps}(q)  \right)\\
    & \geq\frac{1}{4}\left(\sum_{t \in \cL} t \sum_{q \in V_{2t/\eps}} x_{2t/\eps}(q) -2\Delta^{-9\eps} (\log \Delta +1) \cost(P)   \right)\\
    & \geq \eps \sum_{t \in \cL, t \geq 2/\eps} t \sum_{q \in V_{t}} x_{t}(q) - \frac{\eps}{100}\cost(P)\\
    & \geq \eps \cost(P) / 4 
\end{align*}
where the last line we used that all distances in the graph lie between $1/\eps^3 > 100/\eps^2$ and $\Delta$, 
which implies that $\sum_{t \in \cL, t < 2/\eps} t x_t(p) \leq (4/\eps)n \leq (\eps/10) \cost(P)$, 
and so using Proposition \ref{prop:xConstFactor}, we have
\[\sum_{t \in \cL, t \geq 2/\eps} t \sum_{q \in V_{t}} x_{t}(q)  \geq  \sum_{t \in \cL} t \sum_{q \in V_{t}} x_{t}(q)  -\eps  \cost(P) /2 \geq \cost(P)/3\] which completes the proof.
\end{proof}

%% file: OnePassSketching.tex

\section{A Single-Pass Sketching Algorithm for Euclidean MST}\label{sec:onePassSketch}
In this section, we demonstrate how to implement the one-pass friendly estimator from Section \ref{sec:estimatorOnePass} via a single-pass linear sketching algorithm. Specifically, we will prove the following theorem. 
\begin{theorem}\label{thm:onePassStreamMain}
Let $P \subset [\cord]^d$ be a set of $n$ points in the $\ell_1$ metric, represented by the indicator vector $x \in \R^{\cord^d}$, and fix $\eps >0$. Then there is a randomized linear sketching algorithm $F:\R^{\cord^d} \to \R^s$, where $s = \tilde{O}(\cord^{O(\eps)} (1-\eps)^{-2d} d^{O(1)})$, which given $F(x)$, with probability $1-1/\poly(n)$ returns a value $R$ such that 
\[ \cost(P) \leq  R \leq O\left(\frac{\log \eps^{-1}}{\eps^2}\right)\cost(P) \]
 The algorithm uses total space $\tilde{O}(s)$.  
\end{theorem}
The above theorem holds for $d$-dimensional $\ell_1$ space, but with a slightly exponential $(1-\eps)^{-2d}$ dependency on $d$. However, for the case of the $\ell_2$ metric, we can always assume that $d=O(\log n)$ by first embedding $n$ points in $(\R^d,\ell_2)$ into $(\R^{O(\log n)},\ell_1)$ with constant distortion. Specifically, we can first apply an AMS sketch \cite{alon1996space}, or a Johnson-Lindenstrauss Transform (i.e., dense Gaussian sketch), to embed $(\R^d,\ell_2)$ into $(\R^{O(\log n)},\ell_1)$ each of the $\binom{n}{2}$ distances are distorted by at most a $(1 \pm 1/10)$ factor. Naively, the JL transform requires full independence of the Guassian random variables, however this assumption can be removed via a standard application of Nisan's psuedorandom generator (see \cite{indyk2006stable} for details).  Then, noting that we can reduce to the case that, for such $d$, we have $d \cord = (nd)^{O(1)} = n^{O(1)}$, via Proposition \ref{prop:minDistance}, we obtain the following:

\begin{theorem}\label{thm:onePassMain}
Let $P \subset [\cord]^d$ be a set of $n$ points in the $\ell_2$ metric, represented by the indicator vector $x \in \R^{\cord^d}$, and fix $\eps >0$. Then there is a radnomized linear sketching algorithm $F:\R^{\cord^d} \to \R^s$, where $s = n^{O(\eps)}$, which given $F(x)$, with high probability returns a value $R$ such that 
\[ \cost(P) \leq  R \leq O\left(\frac{\log \eps^{-1}}{\eps^2}\right)\cost(P) \]
 The algorithm uses total space $\tilde{O}(s)$.  
\end{theorem}

To prove Theorem \ref{thm:onePassStreamMain}, we will utilize the following key sketching Lemma. Intuitively, the Lemma allows us to sample from a three part process, where at each point we sample a subset of the coordinates of a large vector $x$ (under some pre-defined grouping)  with probability proportional to its $\ell_p^p$ norm, and then sample a sub-subgroup of that vector, and so on. Since ultimately we will want to uniformly sample points to run Algorithm \ref{alg:onepassMain}, we will apply Lemma \ref{lem:mainLinearSketch} below with the setting of $p$ very close to zero so that $\|x\|_p^p \approx \|x\|_0$.

\begin{lemma}\label{lem:mainLinearSketch}
Fix integers $n_1,n_2, n_3 \geq 1$,  fix any $p \in (0,2]$, number of samples $S \geq 1$, parameter $\delta > 0$, and  let $n = n_1 n_2 n_3$. 
For $i \in [n_1]$, let $x_i = \{x_{i,1},x_{i,2},\dots,x_{i,n_2}\} \subset \R^{n_3}$ be a set of vectors. Define $\|x_{j,i}\|_p^p = \sum_{k=1}^{n_3} |x_{i,j,k}|^p$, $\|x_j\|_p^p = \sum_{i=1}^{n_2} \|x_{j,i}\|_p^p$ and $\|x\|_p^p = \sum_{i=1}^{n_1} \|x_i\|_p^p$. 

Then there is a randomized linear function $F:\R^{n} \to \R^s$,  where $s = O(\poly(\log n, \frac{1}{p}, \frac{1}{\delta}, S^2))$, which with returns a set of $S^2 +1$ samples $\{i_1,i_2^1,i_2^2,\dots,i_2^S\} \cup_{\ell=1}^S \{i_3^{\ell,1},i_3^{\ell,2},\dots,i_3^{\ell,S}\}$ from a distribution with total variational distance at most $\delta$ from the following:
\begin{itemize}
    \item $i_1 \sim [n_1]$ is sampled from from the distribution $(\pi_1^0,\dots,\pi_{n_1}^0),$ where $\pi_j^0 = \frac{\|x_j\|_p^p}{\|x\|_p^p}$.
    \item  Conditioned on $i_1 \sim [n_1]$, $(i_2^1,i_2^2,\dots,i_2^S) \sim_{\pi^1} [n_2]$ are each sampled independently from the distribution $\pi^1 = (\pi_1^1,\dots,\pi_{n_2}^1),$ where $\pi_j^1 = \frac{\|x_{i_1,j}\|_p^p}{\|x_{i_1}\|_p^p}$
    \item  For each $\ell \in [S]$, conditioned on $i_1 \sim [n_1]$ and $i_2^\ell \sim_{\pi^{\ell,2}} [n_2]$, the samples $(i_3^{\ell,1},i_3^{\ell,2},\dots,i_3^{\ell,S})$ are drawn independently from the distribution $\pi^{\ell,2} = (\pi_1^{\ell,2},\dots,\pi_{n_3}^{\ell,2})$, where $\pi_j^{\ell,2} =  \frac{|x_{i_0,i_1^\ell,j}|^p}{\|x_{i_1,j^\ell_2}\|_p^p}  $. 
\end{itemize}
Moreover, the function $F$ can be stored in $\tilde{O}(s)$ bits of space.
\end{lemma}

\subsection{Proof of Theorem \ref{thm:onePassStreamMain} Given Lemma \ref{lem:mainLinearSketch}.}
We now prove Theorem \ref{thm:onePassStreamMain} given Lemma \ref{lem:mainLinearSketch}. The remainder of the section will be devoted to proving the Lemma \ref{lem:mainLinearSketch}. 
\begin{proof}[Proof of Theorem \ref{thm:onePassStreamMain}]

    Our linear sketch will implement Algorithm \ref{alg:onepassMain}. We first note that the desired approximation $\Delta$ to the diameter of $P$, as is needed by Theorem \ref{thm:OnePassEstimator}, can be obtained by Lemma \ref{lem:diamapprox}. Then by Theorem \ref{thm:OnePassEstimator}, since each iteration of the while loop in Algorithm \ref{alg:onepassMain} sucseeds with probability at least $(1-\eps)^{-2d}$, it will suffice to show how to generate a single sample from each of Algorithms \ref{alg:onepassDead}, \ref{alg:onepassBad}, and \ref{alg:onepassComplete}. By Theorem \ref{thm:OnePassEstimator}, if we can produce at least $s = \tilde{O}( d^4 \log^4(\Delta) \Delta^{O(\eps)} \log( n) (1-\eps)^{2d})$ samples from each of  these algorithms (some of which may result in $\fail$), we will obtain the desired approximation with high probability. 

To simulate each of Algorithms \ref{alg:onepassDead}, \ref{alg:onepassBad}, and \ref{alg:onepassComplete}, note that each algorithm draws two hash functions $h_1 \sim \cH(t_1), h_2 \sim \cH(t_2)$, for some $t_1,t_2$. It then samples a hash bucket $b$ with probability proportional to $|h_2^{-1}(b)|$, and then samples $c$ with probability proportional to $|h_1^{-1}(c) \cap h_2^{-1}(b)|$, and then samples a point $p \sim h_1^{-1}(c) \cap h_2^{-1}(b)$. It then performs (possible only one of) the tests $|h_2^{-1}(b))| \leq \beta n^{10\eps}$ and $|h_2^{-1}(b) \cap h_1^{-1}(c)| > \beta n^{10\eps}$, and based on the results of these tests, and only in the case that $|h_2^{-1}(b) \cap h_1^{-1}(c)| \leq \beta n^{10\eps}$, it recovers the full set $h_2^{-1}(b) \cap h_1^{-1}(c)$, and outputs a value in $[0,1]$ based on the information it has recovered. We argue that this whole process can be simulated by sampling $b,c$ from the correct distribution, and then drawing $S = \tilde{O}( \beta \Delta^{\eps} )$ samples uniformly from each of $|h_2^{-1}(b) \cap h_1^{-1}(c)|$ and $|h_2^{-1}(b))|$. 

\begin{claim}
    Let $U$ be a finite universe, and $W$ a set of $T = O(t \log t)$ independent samples from $U$, such that each $u \in U$ is sampled with probability at least $\frac{1}{2|U|}$. Then with probability $1-\poly(t)$, $W$ contains at most $t$ samples if and only if $U$ contains at most $t$ distinct items.
\end{claim}
\begin{proof}
    The first direction follows from a coupon collector argument, so it suffices to show that if $|U| > t$ the set $W$ contains at least $t+1$ items. Then we can partition $U$ into disjoint subsets $S_1,\dots,S_{t+1}$ such that $\lfloor |U|/(t+1) \rfloor  \leq |S_i| \leq  2\lfloor |U|/(t+1) \rfloor $ for all $i$. Note that each sample $u \in W$ is drawn from $S_i$ with probability at least $\frac{1}{2(t+1)}$. It follows by the coupon collector argument that at least one sample from each set $S_i$ is drawn with  probability $1-\poly(t)$, in which case $W$ will contain at least $t+1$ distinct items, which completes the claim.

    \end{proof}

    The above claim demonstrates that taking at least $S = \tilde{O}( \beta n^{\eps} )$ samples from each $|h_2^{-1}(b) \cap h_1^{-1}(c)|$ and $|h_2^{-1}(b))|$, we can verify the tests with high probability. Moreover, given this many samples, we can totally recover $|h_2^{-1}(b) \cap h_1^{-1}(c)|$ when it has size less than $\beta n^{10\eps}$. 
    It suffices to give an algorithm that obtains the required samples, for which we will use Lemma \ref{lem:mainLinearSketch}. 
 
To apply Lemma \ref{lem:mainLinearSketch}, 
We first transform $x \mapsto y \in \R^{|N_t|}$ , where recall $N_t$ from Section \ref{sec:prelims} is the set of discritzed points at level $t$, by setting $y_{p} = \sum_{q \in [\cord]^d, f_t(q) = p}x_q$ (where we index the vectors in $\R^{|N_t|}, \R^{\cord^d}$ by points $p \in N_t, q \in [\cord]^d$). Note that this mapping is linear, and we have $\|y\|_0 = |V_t|$. Let $\cU_1,\cU_2$ be the universe of hash values for $h_1,h_2$ respectively. For each $i \in \cU_2$ we define the collection of vectors $x_i = \{x_{i,j}\}_{j \in \cU_1} \subset \R^{\cord^d}$, where 
$$x_{i,j,p} = \sum_{\substack{p \in N_t \\ h_2(p) = i, h_2(p) = j} } y_{p}$$ and here again we index the vectors $y,x_{i,j} \in \R^{\cord^d}$ by points $p \in N_t \subset [\cord]^d$. Then note $\|x_i\|_0 = |h_2^{-1}(i))|$ and $\|x_{i,j}\|_0 = |h_2^{-1}(i) \cap h_1^{-1}(j)|$, and each non-zero $x_{i,j,p}$ corresponds to a point $p \in h_2^{-1}(i) \cap h_1^{-1}(j) \subset V_t$. 
Also note that we can estimate $|V_t| $ to $(1 \pm \frac{1}{2})$ multiplicative error in $\tilde{O}(d\log\cord)$ space with an $L_0$ estimation linear sketch \cite{kane2010optimal}, as is needed for the estimator in Theorem \ref{thm:OnePassEstimator}. 

We can then apply Lemma \ref{lem:mainLinearSketch} on the collection $\{x_i\}_{i \in \cU_2}$ with the setting $S = \tilde{O}( \beta n^{\eps} )$ sample and $p = O(\frac{1}{s_0^2} \log(n))$, and variational distance error $\delta = \frac{1}{10 s_0}$. Since the vector is integral and has $\ell_\infty$ norm at most $n$, we have $\|x\|_p^p = (1 \pm \frac{1}{s_0^2})\|x\|_0$, and a similar bound holds for $\|x_i\|_p^p$ and the vectors therein. This multiplicative error in the sampling changes the expectation bounds from Lemmas \ref{lem:onePassUpperBound} and \ref{lem:onePassLowerBound} by at most $(1 \pm o(1))$, and therefore only affects the approximation obtained by Theorem \ref{thm:OnePassEstimator} by a constant (and note that the above claim holds even for such approximate samplers). 

Now if the distribution had no variational distance to the true $\ell_0$ sampling distribution, then running Lemma \ref{lem:mainLinearSketch} with the above parameters a total of $s_0 = \tilde{O}(\cord^{O(\eps)} (1+\eps)^{-2d} d^{O(1)})$ times would yield the desired approximation by Theorem \ref{thm:OnePassEstimator}. Now since we take a total of $s_0$ samples from the distribution, the probability that the correctness of the algorithm is affected by the $\gamma < 1/(10 s_0)$ variational distance is at most $1/10$. Formally, one can define a coupling between the true sampling distribution and the distribution output by each call to Lemma \ref{lem:mainLinearSketch}, such that with probability at least $1-\gamma$ each sample produced by the two coupled distributions is the same (this is possible by the definition of variational distance). Since we produce at most $s_0$ samples, by a union bound the probability that any sample in the coupling differs is at most $1/10$. Since the algorithm which used the true distribution was correct with probability at least $1-1/\poly(\cord^d)$ by Theorem \ref{thm:OnePassEstimator}, it follows that our estimator returns a correct approximation with probability at least $1-1/10 - 1/\poly(\cord^d)$, which can be boosted to any arbitrary high probability in $\cord^d$ by repeating $O(d \log \cord)$ times independently.

\end{proof}

\input{exponentials.tex}

\newcommand{\LZeroSample}{\operatorname{{\bf BlockSampleL0}}}

\subsection{Proof of Lemma \ref{lem:mainLinearSketch}}
The goal of this section is to prove the following Lemma \ref{lem:mainLinearSketch}. 
To do so, we first describe an idealized sampling routine in Figure \ref{fig:mainoutersketch2}, which we will implement with a linear sketch. 
We describe the routine for the case of $S=1$, and then show how to generalize it to larger $S$.

\begin{figure}[t!]
	\begin{framed}
		\begin{flushleft}	
  \textbf{Sample:}
  Fix a parameter $\gamma >0$, and generate the following collection of independent exponential random variables:
  \[\btau_1 = \{\bt_{i_1}\}_{i_1 \in [n_1]}, \qquad \btau_2 = \{\bt_{i_1,i_2}\}_{i_1 \in [n_1], i_2 \in [n_2]} ,\qquad  \btau_3 = \{\bt_{i_1,i_2,i_3}\}_{i_1 \in [n_1] , i_2 \in [n_2], i_3 \in [n_3]}\]
  \begin{itemize}
      \item Define $\bz^1 \in \R^{n_1}$ via $\bz_i^1 = \|x_i\|_p/\bt_i^{1/p}$ for $i \in [n_1]$, and set $\bi^*_1 = \arg \max_i \bz_i^1$ and $i^{**}_1 = \arg \max_{i \neq i^*_1} Z_i^1$. 
      \item Define $\bz^2 \in \R^{n_2}$ via $\bz_i^2 = \|x_{i^*_1,i}\|_p/\bt_{i^*_1,i}^{1/p}$ for $i \in [n_2]$, and set $i^*_2 = \arg \max_i \bz_i^2$ and $i^{**}_2 = \arg \max_{i \neq i^*_2} Z_i^2$.
      \item Define $\bz^3 \in \R^{n_3}$ via  $\bz_i^3 = \|x_{i^*_1,i^*_2,i}\|_p/\bt_{i^*_1,i^*_2}^{1/p}$ for $i \in [n_3]$, and set $i^*_3 = \arg \max_i \bz_i^2$ and $i^{**}_3 = \arg \max_{i \neq i^*_3} Z_i^2$.
  \end{itemize}
 For $j \in \{1,2,3\}$, let $\cE^j_1$ denote the event 
\begin{equation}\label{eqn:conditioningResult}
     |\bz_{i^*_j}^j|^p \geq \gamma \cdot \begin{cases} 
         \|x\|_p^p & \text{ if } j=1 \\
           \|x_{i_1^*}\|_p^p & \text{ if } j=2 \\
             \|x_{i_1^*,i_2^*}\|_p^p & \text{ if } j=3 \\
     \end{cases} , \qquad   \text{and} \qquad   |\bz_{i^*_j}^j|^p \geq (1+\gamma)|\bz_{i^{**}_j}^j|^p 
\end{equation}
And let $\cE^1_2,\cE_2^2,\cE^3_2$ denote the three events:
\[  \cE^1_2 \stackrel{\text{def}}{=}   \sum_{i \in [n_1]} \frac{\|x_i\|_p^p}{\bt_i} \leq \frac{4 \log(n/\gamma)}{\gamma} \|x\|_p^p, \qquad  \cE^2_2 \stackrel{\text{def}}{=}  \sum_{\substack{i \in [n_1]  \\ j \in [n_2]}} \frac{\|x_{i,j}\|_p^p}{\bt_i \bt_{i,j}} \leq \left(\frac{4 \log(n/\gamma)}{\gamma} \right)^2\|x\|_p^p \] 
\[\cE^3_2 \stackrel{\text{def}}{=} \sum_{\substack{i \in [n_1]  \\ j \in [n_2] \\ k \in [n_3]}} \frac{\|x_{i,j,k}\|_p^p}{\bt_i \bt_{i,j} \bt_{i,j,k}} \leq \left(\frac{4 \log(n/\gamma)}{\gamma}\right)^3 \|x\|_p^p\]
  \textbf{Ouput:} The values $(i_1^*,i_2^*,i_3^*)$ conditioned on the events $\cE^1_1\cap \cE^2_1\cap \cE^3_1 \cap \cE^1_2\cap \cE^2_2\cap \cE^3_2$.
			 \end{flushleft}
	\end{framed}\vspace{-0.2cm}\caption{Algorithm to generate one sample $(i_1,i_2,i_3)$,.}\label{fig:mainoutersketch2}
\end{figure}

\begin{proposition}\label{prop:rightDist}
Let $\cE^1_1, \cE^2_1, \cE^3_1 ,\cE^1_2,\cE^2_2, \cE^3_2$ be the events from Figure \ref{fig:mainoutersketch2}. Then we have
\[ \pr{\cE^1_1\cap \cE^2_1\cap \cE^3_1 \cap \cE^1_2\cap \cE^2_2\cap \cE^3_2} \geq 1-24\gamma\]
\end{proposition}
\begin{proof}
    By Lemma \ref{lem:gapexponentials},  the events $\cE^1_1, \cE^2_, \cE^3_1 $ each hold with probability at least $1-4\gamma$. Next, event $\cE_2^1$ holds with probability $1-2\gamma$ by Lemma \ref{lem:sumexponentials}. For $\cE_2^2$, applying Lemma \ref{lem:sumexponentials} twice we have
\[    \sum_{\substack{i \in [n_1]  \\ j \in [n_2]}} \frac{\|x_{i,j}\|_p^p}{\bt_i \bt_{i,j}} \leq \frac{4 \log(n/\gamma)}{\gamma} \sum_{i \in [n_1]} \frac{\|x_{i}\|_p^p}{\bt_i }  \leq \left(\frac{4 \log(n/\gamma)}{\gamma}\right)^2 \|x\|_p^p \] 
With probability $1-4\gamma$ by a union bound. Applying Lemma \ref{lem:sumexponentials} three times yields $\pr{\cE_2^3} \geq 1-6 \gamma$. Union bounding over all $6$ events completes the proof. 
\end{proof}

We now provide the proof of the streaming algorithm given Lemma \ref{lem:mainLinearSketch}

 In what follows, we describe the construction of the linear sketch that will recover a single sample $(i_1,i_2,i_3)$ from the sampling procedure in Figure \ref{fig:mainoutersketch2} so as to satisfy \ref{prop:rightDist}. Specifically, given sets of exponential $\tau_1,\tau_2,\tau_3$ drawn conditionally on the events $\cE_1 \cap \cE_2 \cap \cE_3$ as Figure \ref{fig:mainoutersketch2}, we show how to recover $(i_1^*,i_2^*,i_3^*)$ with a linear sketch.  We will use a separate linear sketch for each of the three coordinates. 
 We begin by demonstrating how to sample $i_1^*$.


 \paragraph{Recovering $i_1^*$.}
We first define a modified notion of the Count-Sketch estimator \cite{charikar2002finding}, which is a hashing based linear sketch which can be used to approximately recover the largest coordinate in a vector $x \in \R^n$. We will set the parameter $\gamma = \poly(\frac{1}{\log n}, \frac{1}{S}, \frac{1}{p}, \frac{1}{\delta})$, where $S$ is the number of samples needed by the algorithm.

Let $l = O(\log n)$ be a number of hash repetitions, and let $B = O( \log^2(n/\gamma)/\gamma^7)$ be a integer parameter specifying a number of hash buckets, each with a large enough constant. For each $i \in [l]$, we define a $4$-wise independent hash function $h_i:[n_1] \to [B]$. 
For $j \in [l]$ and $b \in [B]$, we will store $\lambda = O(\log n/(\eps_0 p)^4)$ linear functions of the input $\bA_{i,b,t}$, where $\eps_0 = \poly(\gamma)$. To do this, for every $(i_1,i_2,i_3,t) \in [n_1] \times [n_2] \times [n_3] \times [\lambda]$, we generate a $p$-stable random variable $\balpha_{i_1,i_2,i_3}^t \sim \cD_p$. 
\[\bA_{j,b,t} = \sum_{\substack{i_1 \in [n_1] \\  h_j(i_1) =  b}}\sum_{i_2 \in [n_2] }\sum_{i_3 \in [n_3]} \frac{\balpha_{i_1,i_2,i_3}^t x_{i_1,i_2,i_3}}{\bt_{i_1}^{1/p}} \]
We then set 
\begin{equation}\label{eqn:i1}  \bB_{j,b} = \median_{t \in [\lambda]} \left\{ \frac{|\bA_{j,b,t}|}{\median(|\cD_p|)|}  \right\} \qquad \text{and } \qquad     i_1 = \arg\max_{i \in [n_1]} \left\{ \median_{\ell \in [l]} \left\{ \bB_{\ell, h_\ell(i)} \right\} \right\}   
\end{equation}
where $\cD_p$ is the $p$-stable distribution.
Note that the entire data structure can be implemented with a linear sketch of size $O(lB\lambda)$. 
We claim that $i_1 = i_1^*$ with high probability.
\begin{proposition}\label{prop:bucketVal}
    Fix any $j \in [l], b \in [B]$, and let $\cB_{j,b}$ be as above. Then we have 
    \[ \bB_{j,b} = (1 \pm \eps_0) \left(   \sum_{\substack{i_1 \in [n_1] \\  h_j(i_1) =  b}} \frac{\|x_{i_1}\|_p^p}{\bt_{i_1}} \right)^{1/p}   \]
    With probability $1-1/\poly(n)$. 
\end{proposition}
\begin{proof}
By $p$-stability, for each $t \in [t_0]$ we have 
    \[\bA_{j,b,t} = \sum_{\substack{i_1 \in [n_1] \\  h_j(i_1) =  b}} \sum_{\substack{i_2 \in [n_2] \\i_3 \in [n_3]} }\frac{\balpha_{i_1,i_2,i_3}^t x_{i_1,i_2,i_3}}{\bt_{i_1}^{1/p}} = \alpha_t \left(\sum_{\substack{i_1 \in [n_1] \\  h_j(i_1) =  b}} \frac{\|x_{i_1}\|_p^p}{\bt_{i_1}}  \right)^{1/p} \]
    where $\alpha_t \sim \cD_p$ is a $p$-stable random variable.
    The proposition then follows by the correctness of $p$-stable sketches (Theorem \ref{thm:indykPrelims}).
\end{proof}

\begin{lemma}\label{lem:i1}
    Let $i_1$ be defined as in (\ref{eqn:i1}). Then given that the exponential $\tau_1$ satisfy event $\cE_1^1\cE_2^1$, we have that $i_1 = i^*_1$ with probability $1-1/\poly(n)$.
\end{lemma}
\begin{proof}
For any $j \in [l$, by the Proposition \ref{prop:bucketVal}, for any $i \in [n_1]$ we have
\[ \bB_{j,h_j(i)} = (1 \pm \eps_0) \left( \frac{\|x_{i}\|_p^p}{\bt_{i}}  + \sum_{\substack{i_1 \in [n_1] \setminus \{i\} \\  h_j(i_1) =  h_j(i)}} \frac{\|x_{i_1}\|_p^p}{\bt_{i_1}} \right)^{1/p}   \]
Let $K = \frac{\|x_{i_1^*}\|_p^p}{\bt_{i_1^*}}$. Then we have $\bB_{j,h_j(i)} \geq (1-\eps_0)K^{1/p}$ for all $j \in [l]$, thus $\median_{j \in [l]} \bB_{j,h_j(i_1^*)} \geq (1-\eps_0)K^{1/p}$. 
Next, suppose that $i \neq i_1^*$. 
Since this is the case, conditioned on $\cE_1^1$ we have $(\bz_i^1)^p<K/(1+\gamma) \leq (1-\gamma+O(\gamma^2))K$, so
\[ \bB_{j,h_j(i)} < (1 + \eps_0) \left( (1-\gamma + O(\gamma^2))K  + \sum_{\substack{i_1 \in [n_1] \setminus \{i\} \\  h_j(i_1) =  h_j(i)}} \frac{\|x_{i_1}\|_p^p}{\bt_{i_1}} \right)^{1/p}   \]
We now bound the second term inside the parenthesis above. Since $h_j$ is a uniform hash function, we have
\[ \ex{\sum_{\substack{i_1 \in [n_1] \setminus \{i\} \\  h_j(i_1) =  h_j(i)}} \frac{\|x_{i_1}\|_p^p}{\bt_{i_1}}} \leq \frac{1}{B} \sum_{\substack{i \in [n_1]\\ j \in [n_2]}} \frac{\|x_{i_1}\|_p^p}{\bt_{i_1}} \leq \frac{4\log (n/\gamma)}{\gamma B} \|x\|_p^p \leq  \frac{4\log (n/\gamma)}{\gamma^2 B}  K\]
Where the second to last inequality used the event $\cE_2^1$, and the last inequality used  $\mathcal{E}_1$. 
So by Markov's inequality, since $B > 20 (4 \log (n/\gamma)/\gamma^3)$, with probability $9/10$ we have 
\[ \bB_{j,h_j(i)} < (1 + \eps_0)  (1-\gamma/2 + O(\gamma^2))^{1/p}K^{1/p}   < (1-\eps_0) K^{1/p} \]
Where we used that $\eps_0 = \poly(\gamma)$ and $\gamma < \poly(1/p)$. It follows that  $\median_{j \in [l]} \bB_{j,h_j(i)} < (1-\eps_0)K^{1/p}$ for any $i \neq i_1^*$ with high probability using a Chernoff bound over the $l$ repetitions,  which completes the proof. 
\end{proof}


 \paragraph{Recovering $i_2^*$.}
 We now demonstrate how to recover $i_2^*$ given $i_1^*$ and the sets of exponential $\tau_1,\tau_2$ drawn conditionally on the events $\cE_1^1 \cap \cE_2^1 \cap \cE_1^2 \cap \cE_2^2$. Our sketch will again be another Count-Sketch with the same parameters  $l = O(\log n)$ repetitions and $B = O( \log^2(n/\gamma)/\gamma^7)$ from earlier. 
 Then for each $i \in [l]$, we define a $4$-wise independent hash function $h_i^2:[n_1] \times [n_2] \to [B]$. 
For $j \in [l]$ and $b \in [B]$, we will store $\lambda = O(\log n/(\eps_0 p)^4)$ linear functions  $\bA_{j,b,t}'$ of the input, where $t \in [\lambda]$ and $\eps_0 = \poly(\gamma)$. Again, for every $(i_1,i_2,i_3,t) \in [n_1] \times [n_2] \times [n_3] \times [\lambda]$, we generate a $p$-stable random variable $\balpha_{i_1,i_2,i_3}^t \sim \cD_p$, and set
\[\bA_{j,b,t}' = \sum_{\substack{(i_1,i_2) \in [n_1] \times [n_2] \\  h_j^2(i_1,i_2) =  b}} \sum_{i_3 \in [n_3]} \frac{\balpha_{i_1,i_2,i_3}^t x_{i_1,i_2,i_3}}{\bt_{i_1}^{1/p} \bt_{i_1,i_2}^{1/p}} \]
We then set 

\begin{equation}\label{eqn:i2} \bB_{j,b}' = \median_{t \in [\lambda]} \left\{ \frac{|\bA_{j,b,t}'|}{\median(|\cD_p|)|}  \right\}, \qquad \text{and} \qquad    i_2= \arg\max_{i \in [n_2]} \left\{ \median_{\ell \in [l]} \left\{ \bB_{\ell, h_\ell^2(i_1^*,i)}' \right\} \right\}  
\end{equation}

Note that the entire data structure can be implemented with a linear sketch of size $O(lB\lambda)$. 
We claim that $i_2 = i_2^*$ with high probability. First, the proof of the following Proposition is identical to the proof of Proposition \ref{prop:bucketVal}.

\begin{proposition}\label{prop:bucketVal2}
    Fix any $j \in [l], b \in [B]$, and let $\bB_{j,b}'$ be as above. Then we have 
    \[ \bB_{j,b}' = (1 \pm \eps_0) \left(    \sum_{\substack{(i_1,i_2) \in [n_1] \times [n_2] \\  h_j^2(i_1,i_2) =  b}} \frac{\|x_{i_1,i_2}\|_p^p}{\bt_{i_1} \bt_{i_1, i_2}} \right)^{1/p}   \]
    With probability $1-1/\poly(n)$. 
\end{proposition}

\begin{lemma} \label{lem:i2} 
    Let $i_2$ be defined as in (\ref{eqn:i2}). Then given that the exponential $\tau_1,\tau_2$ satisfy event $\cE_1^1 \cap \cE_2^1 \cap \cE_1^2 \cap \cE_2^2$, we have that $i_2 = i^*_2$ with probability $1-1/\poly(n)$.
\end{lemma}
\begin{proof}
The proof proceeds similarly to Lemma \ref{lem:i1}. 
For any $j \in [\ell]$, by the Proposition \ref{prop:bucketVal2}, for any $i \in [n_2]$ we have
\[ \bB_{j,h_j^2(i_1^*, i)}' = (1 \pm \eps_0) \left( \frac{\|x_{i_1^*,i}\|_p^p}{\bt_{i_1^*} \bt_{i_1^*, i}}  +  \sum_{\substack{(i_1,i_2) \in [n_1] \times [n_2] \\  h_j^2(i_1,i_2) =  h_j^2(i_1^*,i)}}  \frac{\|x_{i_1, i_2}\|_p^p}{\bt_{i_1} \bt_{i_1, i_2}} \right)^{1/p}   \]
Let $K = \frac{\|x_{i_1^*,i_2^*}\|_p^p}{\bt_{i_1^*} \bt_{i_1^*, i_2^*}} $. By events $\cE_1^1$ and $\cE_1^2$ we have 
\begin{equation}\label{eqn:K2}
    K = \frac{\|x_{i_1^*,i_2^*}\|_p^p}{\bt_{i_1^*} \bt_{i_1^*, i_2^*}}  \geq  \gamma \frac{\|x_{i_1^*}\|_p^p}{\bt_{i_1^*} }  \geq \gamma^2 \|x\|_p^p 
\end{equation}

Then we have $\bB_{j,h_j^2(i_1^*,i_2^*)}' \geq (1-\eps_0) K^{1/p}$ for all $j \in [l]$, thus $\median_{j \in [l]} \bB_{j,h_j^2(i_1^*,i_2^*)}'\geq (1-\eps_0)k^{1/p}$. 
Next, suppose that $i \neq i_2^*$. 
Since this is the case, conditioned on $\cE_1^2$ we have 
\[\frac{\|x_{i_1^*,i}\|_p^p}{\bt_{i_1^*}\bt_{i_1^*, i}}  \leq  \frac{1}{1+\gamma}\frac{\|x_{i_1^*,i_2^*}\|_p^p}{ \bt_{i_1^*}\bt_{i_1^*, i_2^*}} \leq (1-\gamma + O(\gamma^2)) K\]
Thus
\[\bB_{j,h_j^2(i_1^*, i)}'< (1 + \eps_0) \left( (1-\gamma +O(\gamma^2)K +  \sum_{\substack{(i_1,i_2) \in [n_1] \times [n_2] \\  h_j^2(i_1,i_2) =  h_j^2(i_1^*,i)}}  \frac{\|x_{i_1, i_2}\|_p^p}{\bt_{i_1} \bt_{i_1, i_2}} \right)^{1/p}    \]

We now bound the second term inside the parenthesis above. Again, since $h^2$ is a uniform hash function, we have 
\[ \ex{\sum_{\substack{(i_1,i_2) \in [n_1] \times [n_2] \\  h_j^2(i_1,i_2) =  h_j^2(i_1^*,i)}}  \frac{\|x_{i_1, i_2}\|_p^p}{\bt_{i_1} \bt_{i_1, i_2}}} \leq \frac{1}{B}\sum_{(i_1,i_2) \in [n_1] \times [n_2] } \frac{\|x_{i_1, i_2}\|_p^p}{\bt_{i_1} \bt_{i_1, i_2}} \leq \frac{16 \log^2(n/\gamma)}{\gamma^2B}\|x\|_p^p \leq \frac{16 \log^2(n/\gamma)}{\gamma^4 B} K\]
Where the second to last inequality used $\cE_2^2$, and the last used (\ref{eqn:K2}). 
So by Markov's inequality, taking $B > 20 \cdot (16 \log^2(n/\gamma)/\gamma^5)$, with probability $9/10$ we have 
\[ \bB_{j,h_j^2(i_1^*, i)}' < (1 + \eps_0)  (1-\gamma/2 + O(\gamma^2))^{1/p}K^{1/p}   < (1-\eps_0) K^{1/p} \]
Where we used that $\eps_0 = \poly(\gamma)$ and $\gamma < \poly(1/p)$. It follows that  $\median_{j \in [l]} \bB_{j,h_j^2(i_1^*, i)}' < (1-\eps_0)K^{1/p}$ for any $i \neq i_2^*$ with high probability using a chernoff bound over the $l$ repetitions,  which completes the proof. 

\end{proof}


 \paragraph{Recovering $i_3^*$.}
 We now demonstrate how to recover $i_3^*$ given $i_1^*,i_2^*$ and the sets of exponential $\tau_1,\tau_2,\tau_3$ drawn conditionally on the events $\cE^1_1\cap \cE^2_1\cap \cE^3_1 \cap \cE^1_2\cap \cE^2_2\cap \cE^3_2$. Our sketch will be once again be another Count-Sketch with the same parameters  $l = O(\log n)$ repetitions and $B = O( \log^2(n/\gamma)/\gamma^7)$ from earlier. 
 Then for each $i \in [l]$, we define a $4$-wise independent hash function $h_i^3:[n_1] \times [n_2] \times [n_3] \to [B]$. 
For $j \in [l]$ and $b \in [B]$, we will store $\lambda = O(\log n/(\eps_0 p)^4)$ linear functions  $\bA_{i,b,t}''$ of the input, where $t \in [\lambda]$ and $\eps_0 = \poly(\gamma)$. Again, for every $(i_1,i_2,i_3,t) \in [n_1] \times [n_2] \times [n_3] \times [\lambda]$, we generate a $p$-stable random variable $\balpha_{i_1,i_2,i_3}^t \sim \cD_p$, and set
\[\bA_{j,b,t}'' = \sum_{\substack{(i_1,i_2,i_3) \in [n_1] \times [n_2] \times [n_3] \\  h_j^3(i_1,i_2,i_3) =  b}}  \frac{\balpha_{i_1,i_2,i_3}^t x_{i_1,i_2,i_3}}{\bt_{i_1}^{1/p} \bt_{i_1,i_2}^{1/p}  \bt_{i_1,i_2,i_3}^{1/p}} \]
We then set 

\begin{equation}\label{eqn:i3}\bB_{j,b}'' = \median_{t \in [\lambda]} \left\{ \frac{|\bA_{j,b,t}^1|}{\median(|\cD_p|)|}  \right\}, \qquad \text{and} \qquad    i_3= \arg\max_{i \in [n_3]} \left\{ \median_{\ell \in [l]} \left\{ \bB_{\ell, h_\ell^2(i_1^*,i_2^*,i)}'' \right\} \right\}  
\end{equation}

Note that the entire data structure can be implemented with a linear sketch of size $O(lB\lambda)$. 
We claim that $i_3 = i_3^*$ with high probability. The proof of the following Proposition is again identical to the proof of Proposition \ref{prop:bucketVal}.

\begin{proposition}\label{prop:bucketVal3}
    Fix any $j \in [l], b \in [B]$, and let $\bB_{j,b}''$ be as above. Then we have 
    \[ \bB_{j,b}'' = (1 \pm \eps_0) \left(    \sum_{\substack{(i_1,i_2) \in [n_1] \times [n_2] \\  h_j^2(i_1,i_2) =  b}} \frac{\|x_{i_1,i_2}\|_p^p}{\bt_{i_1} \bt_{i_1, i_2}} \right)^{1/p}   \]
    With probability $1-1/\poly(n)$. 
\end{proposition}

The proof of the following Lemma is essentially identical to that of Lemma \ref{lem:i2}
\begin{lemma} \label{lem:i3} 
    Let $i_2$ be defined as in (\ref{eqn:i2}). Then given that the exponential $\tau_1,\tau_2,\tau_3$ satisfy event $\cE^1_1\cap \cE^2_1\cap \cE^3_1 \cap \cE^1_2\cap \cE^2_2\cap \cE^3_2$, we have that $i_3 = i^*_3$ with probability $1-1/\poly(n)$.
\end{lemma}

\begin{proof}
The proof proceeds similarly to Lemma \ref{lem:i1}. 
For any $j \in [\ell]$, by the Proposition \ref{prop:bucketVal3}, for any $i \in [n_2]$ we have
\[ \bB_{j,h_j^3(i_1^*, i_2^*, i)}'' = (1 \pm \eps_0) \left( \frac{\|x_{i_1^*,i_2^*, i}\|_p^p}{\bt_{i_1^*} \bt_{i_1^*, i_2^*} \bt_{i_1^*, i_2^*, i} }  +  \sum_{\substack{(i_1,i_2, i_3) \in [n_1] \times [n_2] \times [n_3] \\  h_j^3(i_1,i_2,i_3) =  h_j^3(i_1^*,i_2^*,i)}}  \frac{\|x_{i_1, i_2,i_3}\|_p^p}{\bt_{i_1} \bt_{i_1, i_2} \bt_{i_1,i_2,i_3}} \right)^{1/p}   \]

Let $K = \frac{\|x_{i_1^*,i_2^*, i_3^*}\|_p^p}{\bt_{i_1^*} \bt_{i_1^*, i_2^*}  \bt_{i_1^*, i_2^*, i_3^*}  } $. By events $\cE_1^1, \cE_1^2, \cE_1^3$ we have 
\begin{equation}\label{eqn:K3}
    K = \frac{\|x_{i_1^*,i_2^*, i_3^*}\|_p^p}{\bt_{i_1^*} \bt_{i_1^*, i_2^*}  \bt_{i_1^*, i_2^*, i_3^*}  } \geq  \gamma \frac{\|x_{i_1^*, i_2^*}\|_p^p}{\bt_{i_1^*}  \bt_{i_1^*, i_2^*}  }  \geq \gamma^2  \frac{\|x_{i_1^*}\|_p^p}{\bt_{i_1^*}   }   \geq \gamma^3 \|x\|_p^p 
\end{equation}

Then we have $\bB_{j,h_j^3(i_1^*,i_2^*,i_3^*)}'' \geq (1-\eps_0) K^{1/p}$ for all $j \in [l]$, thus $\median_{j \in [l]} \bB_{j,h_j^3(i_1^*,i_2^*,i_3^*)}'' \geq (1-\eps_0)k^{1/p}$. 
Next, suppose that $i \neq i_3^*$. 
Since this is the case, conditioned on $\cE_2^3$ we have 
\[\frac{\|x_{i_1^*,i_2^*,i}\|_p^p}{\bt_{i_1^*}\bt_{i_1^*, i_2^*, i}}  \leq  \frac{1}{1+\gamma}\frac{\|x_{i_1^*,i_2^*,i_3^*}\|_p^p}{ \bt_{i_1^*}\bt_{i_1^*, i_2^*} \bt_{i_1^*,i_2^*,i_3^*}} \leq (1-\gamma + O(\gamma^2) )K\]
Thus

\[\bB_{j,h_j^3(i_1^*,i_2^*,i)}'' < (1 + \eps_0) \left( (1-\gamma +O(\gamma^2))K+  \sum_{\substack{(i_1,i_2, i_3) \in [n_1] \times [n_2] \times [n_3] \\  h_j^3(i_1,i_2,i_3) =  h_j^3(i_1^*,i_2^*,i)}}  \frac{\|x_{i_1, i_2,i_3}\|_p^p}{\bt_{i_1} \bt_{i_1, i_2} \bt_{i_1,i_2,i_3}}  \right)^{1/p}    \]

We now bound the second term inside the parenthesis above. Again, since $h^3$ is a uniform hash function, we have 
\begin{align*}
    \ex{\sum_{\substack{(i_1,i_2, i_3) \in [n_1] \times [n_2] \times [n_3] \\  h_j^3(i_1,i_2,i_3) =  h_j^3(i_1^*,i_2^*,i)}}  \frac{\|x_{i_1, i_2,i_3}\|_p^p}{\bt_{i_1} \bt_{i_1, i_2} \bt_{i_1,i_2,i_3}} } & \leq \frac{1}{B}\sum_{(i_1,i_2, i_3) \in [n_1] \times [n_2] \times [n_3] }  \frac{\|x_{i_1, i_2,i_3}\|_p^p}{\bt_{i_1} \bt_{i_1, i_2} \bt_{i_1,i_2,i_3}}  \\
    &\leq \frac{64 \log^3(n/\gamma)}{\gamma^3 B}\|x\|_p^p \\
    &\leq \frac{64 \log^3(n/\gamma)}{\gamma^6 B} K
\end{align*}

Where the second to last inequality used $\cE_2^3$, and the last used (\ref{eqn:K3}). 
So by Markov's inequality, taking $B > 20 \cdot (64 \log^3(n/\gamma)/\gamma^7)$, with probability $9/10$ we have 
\[ \bB_{j,h_j^3(i_1^*,i_2^*,i)}'' < (1 + \eps_0)  (1-\gamma/2 + O(\gamma^2))^{1/p}K^{1/p}   < (1-\eps_0) K^{1/p} \]
Where we used that $\eps_0 = \poly(\gamma)$ and $\gamma < \poly(1/p)$. It follows that  $\median_{j \in [l]} \bB_{j,h_j^3(i_1^*,i_2^*,i)}''< (1-\eps_0)K^{1/p}$ for any $i \neq i_3^*$ with high probability using a Chernoff bound over the $l$ repetitions,  which completes the proof. 

\end{proof}
We are now ready to conclude this section with a proof of Lemma \ref{lem:mainLinearSketch}. 
\begin{proof}[Proof of Lemma \ref{lem:mainLinearSketch}]
    We first describe how to generalize the above procedure to any number of samples $S$. In this case, we generate $1+S^2$ sets of exponentials. Specifically,  we generate
    \[\btau_1 = \{\bt_{i_1}\}_{i_1 \in [n_1]}, \qquad \btau_2^{\ell_1} = \{\bt_{i_1,i_2}^{\ell_1}\}_{i_1 \in [n_1], i_2 \in [n_2]} ,\qquad  \btau_3^{\ell_1,\ell_2} = \{\bt_{i_1,i_2,i_3}^{\ell_1,\ell_2}\}_{i_1 \in [n_1] , i_2 \in [n_2], i_3 \in [n_3]}\]
Where for each $\ell_1,\ell_2$ range over all values in $[S]$. Then for each $(\ell_1,\ell_2) \in [S]^2$, we condition on the events from Figure \ref{fig:mainoutersketch2} holding for $\btau_1,\btau_2^{\ell_1}, \btau_3^{\ell_1,\ell_2}$. By Proposition \ref{prop:rightDist}, these hold for this trio with probability $1-24\gamma$. Since there are $S^2$ such choices of trios, by a union bound the events from Figure \ref{fig:mainoutersketch2} hold for all trios with probability $1-24 \cdot S^2 \gamma > 1-\delta$ by the setting of $\gamma$. Conditioned on these events, we recover the desired samples $(i_1^*,i_2^*,i_3^*)$ corresponding to all $S^2$ sets $\btau_1,\btau_2^{\ell_1}, \btau_3^{\ell_1,\ell_2}$ of exponentials. The proof that each of these samples comes from the correct distribution follows from standard facts about exponential variables (see Fact \ref{fact:order}). Note that the $1/\poly(n)$ probability of failure from the Chernoff bounds used in Lemmas \ref{lem:i1}, \ref{lem:i2}, \ref{lem:i3} can be absorbed into the variatonal distance. 

Finally, for the space complexity bounds, note that we need only store the $p$-stable random variables and the exponential distribution, as the count-sketch hash functions are only $4$-wise independent (and therefore can be stored in $O(\log n)$ bitse each). Since $p$-stable variables have tails bounded by $\prb{\bx \sim \cD_p}{|\bx|>t} <O(1/t^p)$ for $p \in (0,2]$,  it follows that each of the $\poly(n)$ $p$-stable random variables generated by the algorithm is bounded by $X = \exp( \log^{O(1/p)}(n))$ with probability $1-1/\poly(n)$. Since the exponential distribution has tails which decay faster than $O(1/t^p)$, the same holds true for the exponentials generated by the algorithm. One can then rounding all random variables to a multiple of $1/X$, inducing a total additive error of $\exp( -\log^{O(1/p)}(n))$ into each entry of the sketch. Then each variable can be stored in  a total of $O(\max\{\frac{1}{p} , 1\} \cdot \poly(\log n))$ random bits (see arguement in Appendix A6 of \cite{kane2010exact} for generating $p$-stables). It suffices to demonstrate that this additive error does not effect the output of the algorithm significantly, for which it suffices to demonstrate that conditioned on the events $\btau_1,\btau_2^{\ell_1}, \btau_3^{\ell_1,\ell_2}$ holding for the unrounded exponential random variables, the algorithm outputs the same values after the rounding. This follows from the fact that the proofs of Lemmas \ref{lem:i1}, \ref{lem:i2}, and \ref{lem:i3} all still hold even when am additive  $\exp( -\log^{O(1/p)}(n))$ error introduced into each bucket $\bB_{i,j}$ (in fact, they hold even if a larger additive error of $\poly(\gamma) \cdot K^{1/p}$ is introduced, where $K$ is as in those Lemmas).

Given the bit-complexity bounds on the random variables being generated, it now suffices to derandomize the $p$-stable and exponential variables. To do this, we use the reordering trick of Indyk \cite{indyk2006stable} to employ Nisan's pseudorandom generator \cite{nisan1990pseudorandom}. Specifically, since the sketch is linear, its value is independent of the order of the updates to $x$, so we can order the coordinates of $x$ to arrive in a stream so that first all coordinates of $x_1$ arrive, and then $x_2$. Within the block of updates to $x_1$, we first add all coordinates to $x_{1,1}$ then $x_{1,2}$, and so on additionally with $x_{1,1,1}$ before $x_{1,1,2}$. Since the exponentials corresponding to each block $x_i$ and $x_{i,j}$ are only used within that block, once the block is completed in the stream we can throw them out, which allows us a tester to full construct the linear sketch $F(x) \in \R^s$ used by algorithm in space at most $O(s)$. The derandomziation then follows from \cite{nisan1990pseudorandom}.  
\end{proof}

\subsection{Approximating the Diameter in a Stream}\label{sec:diam}
In this section, we give a linear sketch using $O((1-d)^\eps \polylog(n))$ space which approximates the diameter of an $n$-point set $P \subset [\cord]^d$. 

\begin{lemma}\label{lem:diamapprox}
    Let $P \subset [\cord]^d$ be a set of $n$ points in the $\ell_1$ metric. Then there is a linear sketching algorithm that uses space $O((1-\eps)^{-d} \poly(d, \log(n,\cord))$ space and returns an integer $\Delta \geq 1$, such that with probability $1-\poly(\frac{1}{n})$ we have

    \[\diam(P)\leq \Delta \leq  \frac{4}{\eps} \diam(P)    \]
\end{lemma}
\begin{proof}
Let $x \in \R^{\cord^d}$ be the indicator vector for $P$. 
    For each $t=2^i$, $i \in \log (d \cord)$, we run a procedure to test if we should output $\Delta=t$. The procedure is as follows. We draw a single hash function $h \sim \cH(t)$ from Lemma \ref{lem:newLSH}, and letting $\cU$ be the universe of possible hash values of $h$, define a vector $z^t \in \R^{|\cU|}$ where $z_i^t = \sum_{p \in [\cord]^d, h(p) = i} x_{p}$. We then use the $(1 \pm 1/10)$ $\ell_0$ estimator of \cite{kane2010optimal}, which requires space $\tilde{O}(\log(|\cU|)) = \tilde{O}(d \cord)$ and distinguishes the case of $\|z^t\|_0 = 1$ from $\|z^t\|_0 > 1$. 

    Now note that if $\|z\|_0 = 1$, then we immediately conclude that $\diam(P) \leq 2t/\eps$ due to the diameter bound from Lemma \ref{lem:newLSH}. The algorithm repeats the hashing a total of $ \ell =  O((1-\eps)^{-d} \log(\cord n d))$ times, obtaining $z_1^t,z_2^t,\dots,z_\ell^t$ and reports $\Delta=2 t_0/\eps$ where $t_0$ is the smallest value for which $\|z^{t_0}_l\|_0 = 1$ for some $l \in [\ell]$. The first inequality from the Lemma immediately follows, thus it suffices to show that at level $t^*$ for any $\diam(P)\leq  t^* \leq 2 \diam(P)$ we have at least one $z^{t^*}_l$ with $\|z^{t^*}_l\|_0 = 1$.     To see this, let $p \in P$ be an arbitrary point. With probability $(1-\eps)^{-d}$, $p$ survives $h \sim \cH(t^*)$, in which case $\cB(p,\diam(P)) \subseteq \cB(p,t^*) \subseteq  h^{-1(h(x)}$. Since $\cB(p,\diam(P))  \cap P = P$, this completes the claim. Taking $\ell$ many repetitions at each level guarantees the result with probability $1-\poly(\frac{1}{n})$, and one can then union bound over all $\log(d \cord)$ levels to complete the proof. 
\end{proof}

%% file: exponentials.tex

\def\MST{\text{EMD}}
\def\sv{\mathsf{v}} \def\su{\mathsf{u}}
\def\bcalI{\boldsymbol{\mathcal{I}}}
\def\bfeta{\boldsymbol{\eta}}
\def\rr{\mathbf{r}} \def\bone{\mathbf{1}}
\def\vv{\mathbf{v}} \def\LS{\mathsf{LS}}
\def\Exp{\text{Exp}} \def\ALG{\mathsf{ALG}}
\def\cc{\mathbf{c}} \def\bDelta{\mathbf{\Delta}}
\def\bh{\mathbf{h}} \def\bsv{\boldsymbol{\sv}}
\def\uu{\mathbf{u}} \def\calU{\mathcal{U}}
\def\aa{\mathbf{a}} \def\bE{\boldsymbol{E}}
\def\bOmega{\boldsymbol{\Omega}} \def\balpha{\boldsymbol{\alpha}}
\def\bomega{\boldsymbol{\omega}} \def\bA{\mathbf{A}}
\newcommand{\calbQ}{\boldsymbol{\calQ}}

\subsection{Stable Distributions and Sketching Tools}

We now introduce several sketching tools that will be necessary for our one-pass linear sketch. The first of these is Indyk's $p$-stable sketch, and in particular the analysis of this sketch for $p$ near $0$. Specifically, we will later need to use this sketch with $p \to 0$ as $n \to \infty$. We first define stable random variables.
\begin{definition}\label{def:stable}
	A distribution $\mathcal{D}$ is said to be $p$-stable if whenever $\bx_1,\dots,\bx_n \sim \mathcal{D}$ are drawn independently and $a \in \R^n$ is a fixed vector, we have
	\[	\sum_{i=1}^n a_i \bx_i \mathop{=}^{d} \|a\|_p \bx	\]
	where $\bx \sim \mathcal{D}$ is again drawn from the same distribution, and the ``$\displaystyle\mathop{=}^d$'' symbol above denotes distributional equality. 
\end{definition}
Such $p$-stable distributions are known to exist for all $p \in (0,2]$. For $p \in (0,2]$, we write $\cD_p$ to denote the standard $p$-stable distribution (see \cite{nolan2009stable}). 
More generally, we refer to \cite{nolan2009stable} for a thorough discussion on $p$-stable random variables, and to \cite{indyk2006stable} for their use in streaming/sketching computation. We will utilize the following standard method for generating $p$-stable random variables for $p < 1$. For $p \in [1,2]$, standard methods of generation can be found in \cite{nolan2009stable}.\footnote{Since the tails of $p$ stables are bounded by $1/t^{p}$, we need only have an explicit form here for generation for the purpose when $p$ is near $0$. }

\begin{proposition}[\cite{chambers1976method}]\label{prop:genstable}
	Fix any $p \in (0,1)$. Then a draw from a $p$-stable distribution $\bx \sim \cD_p$ can be generated as follows: (i) generate $\theta \sim [-\frac{\pi}{2},\frac{\pi}{2}]$, (ii) generate $\boldr \sim [0,1]$, and set 
		\begin{align}
		\bx = \frac{\sin(p \theta)}{\cos^{1/p}(\theta)} \cdot \left(\frac{\cos(\theta(1-p))}{\ln(1/\boldr)}\right)^{\frac{1-p}{p}}. \label{eq:gen-p-stable}
		\end{align}
\end{proposition}

We are now ready to state the Indyk $p$-stable sketch. The analysis of this sketch for the case of $p$ near $0$ is given in Theorem 18 of \cite{chen2022new}. 
\begin{theorem}[Indyk's $p$-Stable Sketch \cite{indyk2006stable}] \label{thm:indykPrelims}
	Consider any $p,\eps,\delta \in (0,2]$, and set $t = O\left(\frac{\log\delta^{-1}}{ (p\eps)^4}\right)$ Let $\bx_1,\dots,\bx_t \sim \cD_p$ independently. Then, with probability at least $1-\delta$, 
	\[1-\eps\leq	\median_{i \in [t]} \left\{	\frac{|\bx_i|}{\median(|\cD_p|)|}	\right\}  \leq 1+\eps.		\]
	Where $\median(|\cD_p|)|$ is the value such that
	\[			\prb{X \sim \cD_p}{|X| < \median(|\cD_p|)|} =1/2	\]
 For $p$ bounded away from $0$, one requires only $t = O(\eps^{-2} \log \delta^{-1})$. 
\end{theorem}

\subsection{Exponential Order Statistics}\label{sec:expo}
We review some properties of the order statistics of independent non-identically distributed exponential random variables. 
Let $(\bt_1,\dots,\bt_n)$ be independent exponential random variables where $\bt_i$ has mean $1/\lambda_i$ (equivalently, $\bt_i$ has rate $\lambda_i>0$), abbreviated as $\bt_i\sim\Exp(\lambda_i)$. Recall that $\bt_i$ is given by the cumulative distribution function (cdf) $\Pr[\bt_i\le x] = 1-e^{-\lambda_i x}$. Our algorithm will require an analysis of the distribution of values $(\bt_1,\dots,\bt_n)$. We begin by noting that constant factor scalings of an exponential variable result in another exponential variable.

\begin{fact}[Scaling of exponentials]\label{fact:scale}
	Let $\bt\sim \Exp(\lambda)$ and $\alpha > 0$. Then $\alpha \bt$ is distributed as $\Exp(\lambda/\alpha)$.
\end{fact}
\begin{proof}
	The cdf of $\alpha \bt$ is given by $ \Pr[\bt < x/\alpha]=1-e^{-\lambda x/\alpha}$, which is the cdf of $\Exp(\lambda/\alpha)$.
\end{proof}

The distribution of the maximum of a collection of exponentials has a simple formulation:

    
\begin{fact}[\cite{nagaraja2006order}]\label{fact:order}
	Let $\bt=(\bt_1,\dots,\bt_n)$ be independently exponentials with 
	  $\bt_i\sim \Exp(\lambda_i)$. Then $\arg \max_{j \in [n]} \bt_j =i$ with probability $\lambda_i / \sum_{j \in [n]} \lambda_j$. 
\end{fact}
We will use the following two bound on a collection of scaled inverse exponentials, which are proven in \cite{chen2022new}. 

\begin{lemma}[Lemma 5.4~\cite{chen2022new}]\label{lem:sumexponentials}
	Fix $n\in \N$ and $x \in \R^n$ to be any fixed vector, and consider independent draws $\bt_1,\dots,\bt_n \sim \Exp(1)$. For any $\gamma \in (0,1/2)$, 	
	\[	\Prx_{\bt_1,\ldots,\bt_n}\left[\sum_{i\in[n]}\frac{|x_i|}{\bt_i} \geq \frac{4\log (n/\gamma)}{\gamma} \|x\|_1 \right]\leq 2\gamma	.\]	
\end{lemma}

\begin{lemma}[Lemma 5.5~\cite{chen2022new}]\label{lem:gapexponentials} For $n\in \N$, let $x \in \R^n$ be any fixed vector. Let $\bt_1,\dots,\bt_n \sim \Exp(1)$ be i.i.d. exponentially distributed, and $\gamma > 0$ be smaller than some constant. Then, letting $i^*=\arg\max \frac{|x_i|}{t_i}$, we have that
	\begin{align}
	 \frac{|x_{i^*}|}{\bt_{i^*}}\ge \gamma \|x\|_1\qquad\text{and}\qquad  \frac{|x_{i^*}|}{\bt_{i^*}}\ge (1+\gamma)\max_{i\neq i^*}\left\{\frac{|x_i|}{\bt_i}\right\}, \label{Eq:gap_exponentials}	 
	 \end{align}
 holdss with probability at least $1-4\gamma$.
\end{lemma}

%% file: LowerBound.tex
\newcommand{\BHH}{\text{BHH}}
\begin{theorem}
\label{thm:lb}
    There is no 1-pass streaming algorithm with memory $o(\sqrt{n})$ 
    that achieves a better than $1.1035$ approximation to the Euclidean minimum spanning tree under the $\ell_2$ metric when
    the dimension is at least $c \log n$, for a large
    enough constant $c$.
\end{theorem}
\begin{proof}
Our proof will embed approximately $n$ points in an $n$-dimensional Euclidean space. To obtain the bound 
for Euclidean spaces of dimension $c \log n$, it is enough
to apply the Johnson-Lindenstrauss lemma with parameter 
$\eps$, for a small enough $\eps$.

Our lower bound is based on the \emph{Boolean Hidden Matching} problem introduced by~\cite{BarYossef}, see also~\cite{Gavinsky}.
\begin{definition}[Boolean Hidden Matching Problem]
In the Boolean Hidden Hypermatching Problem $\BHH_{t,n}$ Alice gets a vector $x \in \{0, 1\}^n$ with $n = 2kt$ and $k$ is an integer and Bob gets a perfect $t$-hypermatching $M$ on the $n$ coordinates of $x$, i. e., each edge has exactly $t$ coordinates, and a string
$w \in \{0, 1\}^{n/t}$. 
We denote the vector of length $n/t$ given by 
$(\oplus_{1\le i \le t} x_{M_1,i} , \ldots, \oplus_{1\le i \le t} x_{M_{n/t},i} )$ by
$Mx$ where $(M_{1,1},\ldots , M_{1,t}),\ldots ,(M_{n/t,1},\ldots , M_{n/t,t})$ are the edges of $M$. The problem is to answer \texttt{yes} if $Mx \oplus w = 1^{n/t}$ and \texttt{no} if $Mx \oplus w = 0^{n/t}$, otherwise the algorithm may answer arbitrarily.   
\end{definition}
Verbin and Yu~\cite{VerbinY11} prove a lower bound of 
$\Omega(n^{1-1/t})$ for the randomized one-way communication
complexity for $\BHH_{t,n}$.
\begin{theorem}[Verbin and Yu~\cite{VerbinY11}]
    The randomized one-way communication complexity of $\BHH_{t,n}$ 
    is $\Omega(n^{1-1/t})$.
\end{theorem}

For our reduction, we would like to set $t=2$ and use $w = 0^{n/t}$; 
we refer to this specific case of the problem as $\BHH_{2,n}^0$.

We further make use of the following lemma to further 
simplify the problem we work with.
\begin{lemma}[Lemma 10 in \cite{BuryS15}]
For any $t \ge 2$, the communication complexity of 
$\BHH^0_{t,4n}$ is lower bounded by the communication
complexity of $\BHH_{t,n}$.
\end{lemma}

We will thus focus on $\BHH^0_{2,n}$, which by the above theorems
requires a one-way communication complexity of $\Omega(n^{1-1/t})$.
We can thus rephrase the problem as follows.
The input consists of a set of $2n$ elements 
$\{u_1,\ldots,u_n, v_1,\ldots, v_n\}$.
Alice receives a set $M_A$ of $n/2$ 
distinct indices in $[n]$. Bob receives a set $M_B$ of $n/2$ pairs 
$(u_i, u_j)$. Alice should send a message to Bob so that Bob can distinguish
between the following cases:
\begin{itemize}
\item YES: For each pair $(u_i, u_j) \in M_B$, either both $i,j \in M_A$, or none of $i,j \in M_A$; and
\item NO: For none of the pairs $(u_i, u_j) \in M_B$ both $i,j \in M_A$.
\end{itemize}

We show how to reduce Euclidean MST in $n$-dimensions -- one can reduce the 
dimension to $O(\log n / \epsilon^2)$ and only lose a $(1-\eps)$ multiplicative
factor in the hardness bound by using the Johnson-Lindenstrauss Lemma.
We create the following instance.
For each index $i \in [n]$, we associate the point $e_i$, namely the point
whose coordinates are all 0 except for the $i$th coordinate which is 1.
For each pair $(u_i, u_j)$ in $[n] \times [n]$, we create the point $p_{i,j}$ 
whose coordinates are all 0 except for coordinates $i$ and $j$ which are 1.
Note that the distance between any pair of points hence created is at least 1.

The Euclidean MST instances goes as follows: Alice's set consists of all the
points $e_i$ such that $i \in M_A$ plus the origin; and Bob's set consists
of all the points $p_{i,j}$ such that $i,j$ in $M_B$ plus the origin. 
We consider the Euclidean MST instance consists of the union of Alice's and 
Bob's sets.

We now claim that if the boolean hidden matching instance is a YES instance, 
then the induced Euclidean MST instance has cost at most $n$.
Indeed, consider the following solution: connect each point $e_i$ (s.t. $i \in
M_A$) to the origin, for an overall cost of $n/2$. Then, for each $p_{i,j}$ 
s.t. $(u_i,u_j) \in M_B$, we have that $i \in M_A$ (and $j \in M_B$). This 
implies that connecting each point $(u_i,u_j) \in M_B$ to $e_i$ yields a connected tree, since $e_i$ is indeed in the instance. The cost of such connections is at most $n/2$, and the total cost is thus at most $n$.

We now turn to the NO case. We show that the total cost of the induced Euclidean MST instance is at least $(3+\sqrt{2}) n  /4$. Indeed, since it is a NO instance, there exists a set of $n/4$ pairs $(u_i,u_j)$ such that none of 
$i,j$ is in $M_A$. Each such point $p_{i,j}$ is thus at distance at least $\sqrt{2}$ from any other point $p_{i', j'}$ where $i' \neq i$ or $j' \neq j$ 
and from the origin and at distance at least $\sqrt{3}$ from any other point
$e_{\ell}$ that is in the instance.
Moreover, since the pairwise distance between any point of the instance is
at least 1, the remaining $3n/4$ points that are not at the origin are contributing 1 to the cost. Thus, the Euclidean MST cost of the instance is at least $n ((3/4) + (\sqrt{2}/4))$.

It thus requires Alice to send $\Omega(\sqrt{n})$ bits to Bob in order for Bob to approximate the Euclidean MST length within a factor of 
$(3+\sqrt{2})/4 \sim 1.1035..$.
\end{proof}